%% file: main.tex
\documentclass[11pt]{article}

\usepackage{amsthm}
\usepackage{graphicx} 
\usepackage{array} 

\usepackage{amsmath, amssymb, amsfonts, verbatim}
\usepackage{hyphenat,epsfig,subcaption,multirow}
\usepackage[font=small,labelfont=bf]{caption}
\usepackage{nicefrac}
\usepackage{paralist}

\usepackage[usenames,dvipsnames]{xcolor}
\usepackage[ruled]{algorithm2e}

\DeclareFontFamily{U}{mathx}{\hyphenchar\font45}
\DeclareFontShape{U}{mathx}{m}{n}{
      <5> <6> <7> <8> <9> <10>
      <10.95> <12> <14.4> <17.28> <20.74> <24.88>
      mathx10
      }{}
\DeclareSymbolFont{mathx}{U}{mathx}{m}{n}
\DeclareMathSymbol{\bigtimes}{1}{mathx}{"91}

\usepackage{tcolorbox}
\tcbuselibrary{skins,breakable}
\tcbset{enhanced jigsaw}

\usepackage[normalem]{ulem}
\usepackage[compact]{titlesec}

\definecolor{DarkRed}{rgb}{0.5,0.1,0.1}
\definecolor{DarkBlue}{rgb}{0.1,0.1,0.5}

\usepackage{nameref}
\definecolor{ForestGreen}{rgb}{0.1333,0.5451,0.1333}
\definecolor{Red}{rgb}{0.9,0,0}
\usepackage[linktocpage=true,
	pagebackref=true,colorlinks,
	linkcolor=DarkRed,citecolor=ForestGreen,
	bookmarks,bookmarksopen,bookmarksnumbered]
	{hyperref}
\usepackage[noabbrev,nameinlink]{cleveref}
\crefname{property}{property}{Property}
\creflabelformat{property}{(#1)#2#3}
\crefname{equation}{eq}{Eq}
\creflabelformat{equation}{(#1)#2#3}

\usepackage{bm}
\usepackage{url}
\usepackage{xspace}
\usepackage[mathscr]{euscript}

\usepackage{tikz}
\usetikzlibrary{arrows}
\usetikzlibrary{arrows.meta}
\usetikzlibrary{shapes}
\usetikzlibrary{backgrounds}
\usetikzlibrary{positioning}
\usetikzlibrary{decorations.markings}
\usetikzlibrary{patterns}
\usetikzlibrary{calc}
\usetikzlibrary{fit}
\usetikzlibrary{decorations}

\usepackage[framemethod=TikZ]{mdframed}

\usepackage[noend]{algpseudocode}
\makeatletter
\def\BState{\State\hskip-\ALG@thistlm}
\makeatother

\usepackage{cite}
\usepackage{enumitem}
\setlist[itemize]{leftmargin=20pt}
\setlist[enumerate]{leftmargin=20pt}

\usepackage[margin=1in]{geometry}

\newtheorem{theorem}{Theorem}
\newtheorem{lemma}{Lemma}[section]
\newtheorem{proposition}[lemma]{Proposition}
\newtheorem{corollary}[theorem]{Corollary}
\newtheorem{claim}[lemma]{Claim}
\newtheorem{fact}[lemma]{Fact}

\newtheorem*{claim*}{Claim}
\newtheorem*{assumption*}{Assumption}
\newtheorem*{proposition*}{Proposition}
\newtheorem*{lemma*}{Lemma}
\newtheorem*{problem5*}{Problem}

\newtheorem{observation}[lemma]{Observation}
\newtheorem{property}{Property}

\newtheorem*{theorem*}{Theorem}

\crefname{lemma}{Lemma}{Lemmas}
\crefname{claim}{claim}{claims}
\crefname{property}{Property}{Properties}

\newtheorem{mdresult}{Result}
\newenvironment{result}{\begin{mdframed}[backgroundcolor=lightgray!40,topline=false,rightline=false,leftline=false,bottomline=false,innertopmargin=2pt]\begin{mdresult}}{\end{mdresult}\end{mdframed}}

\theoremstyle{definition}

\newtheorem{mdproblem}{Problem}

\newtheorem*{mdproblem*}{Problem}
\newenvironment{Problem*}{\begin{mdframed}[hidealllines=false,innerleftmargin=10pt,backgroundcolor=gray!10,innertopmargin=5pt,innerbottommargin=5pt,roundcorner=10pt]\begin{mdproblem*}}{\end{mdproblem*}\end{mdframed}}
\newtheorem{mddefinition}[lemma]{Definition}
\newenvironment{Definition}{\begin{mdframed}[hidealllines=false,innerleftmargin=10pt,backgroundcolor=white!10,innertopmargin=5pt,innerbottommargin=5pt,roundcorner=10pt]\begin{mddefinition}}{\end{mddefinition}\end{mdframed}}
\newtheorem*{mddefinition*}{Definition}
\newenvironment{Definition*}{\begin{mdframed}[hidealllines=false,innerleftmargin=10pt,backgroundcolor=white!10,innertopmargin=5pt,innerbottommargin=5pt,roundcorner=10pt]\begin{mddefinition*}}{\end{mddefinition*}\end{mdframed}}
\newtheorem{mdremark}{Remark}
\newenvironment{Remark}{\begin{mdframed}[hidealllines=false,innerleftmargin=10pt,backgroundcolor=gray!10,innertopmargin=5pt,innerbottommargin=5pt,roundcorner=10pt]\begin{mdremark}}{\end{mdremark}\end{mdframed}}

\newtheoremstyle{restate}{}{}{\itshape}{}{\bfseries}{~(restated).}{.5em}{\thmnote{#3}}
\theoremstyle{restate}

\allowdisplaybreaks

\renewcommand{\qed}{\nobreak \ifvmode \relax \else
      \ifdim\lastskip<1.5em \hskip-\lastskip
      \hskip1.5em plus0em minus0.5em \fi \nobreak
      \vrule height0.75em width0.5em depth0.25em\fi}

\renewcommand{\leq}{\leqslant}
\renewcommand{\geq}{\geqslant}

\renewcommand{\ge}{\geq}

\input{macros}

\title{$\mathcal{O}(\log\log{n})$ Passes is Optimal for Semi-Streaming \\ Maximal Independent Set}
\author{Sepehr Assadi\footnote{(sepehr@assadi.info) Cheriton School of Computer Science, University of Waterloo, and Department of Computer Science, Rutgers University. 
Supported in part by an Alfred P. Sloan Fellowship, a University of Waterloo startup grant, and an NSF CAREER grant CCF-2047061. \smallskip}
\and Christian Konrad\footnote{(christian.konrad@bristol.ac.uk) School of Computer Science, University of Bristol. Supported by EPSRC New Investigator Award EP/V010611/1. \smallskip}
\and Kheeran K. Naidu\footnote{(kheeran.naidu@bristol.ac.uk) School of Computer Science, University of Bristol. Supported
by EPSRC DTP studentship EP/T517872/1. \smallskip}
\and Janani Sundaresan\footnote{(jsundaresan@uwaterloo.ca) Cheriton School of Computer Science, University of Waterloo. \smallskip}} 

\date{}

\begin{document}
\maketitle

\pagenumbering{roman}

\input{abstract}

\clearpage

\setcounter{tocdepth}{3}
\tableofcontents
\clearpage
\pagenumbering{arabic}
\setcounter{page}{1}

\input{intro}

\input{overview}

\input{prelim}

\input{dup-embed}

\input{distribution}

\input{analysis}

\section*{Acknowledgement} 

Sepehr Assadi would like to thank Yu Chen and Sanjeev Khanna for several discussions on proving multi-pass semi-streaming lower bounds for MIS as part of their collaborations in~\cite{AssadiCK19a} (on one-pass lower bounds for MIS) and~\cite{AssadiCK19b} (on poly-pass lower bounds for lexicographically-first MIS). He is also thankful to Gillat Kol and Zhijun Zhang for their collaboration in~\cite{AssadiKZ22} on proving distributed sketching lower bounds for MIS.

\bigskip

\bibliographystyle{halpha-abbrv}
\bibliography{general}

\clearpage
\appendix

\part*{Appendix}
\input{info}

\input{appendix-dup}

\end{document}

%% file: macros.tex
\newcommand{\Leq}[1]{\ensuremath{\underset{\textnormal{#1}}\leq}}
\newcommand{\Geq}[1]{\ensuremath{\underset{\textnormal{#1}}\geq}}
\newcommand{\Eq}[1]{\ensuremath{\underset{\textnormal{#1}}=}}

\newcommand{\inner}[2]{\ensuremath{\langle #1 \; , \; #2 \rangle}}
\newcommand{\tvd}[2]{\ensuremath{\norm{#1 - #2}_{\mathrm{tvd}}}}
\newcommand{\Ot}{\ensuremath{\widetilde{O}}}
\newcommand{\eps}{\ensuremath{\varepsilon}}
\newcommand{\Paren}[1]{\Big(#1\Big)}
\newcommand{\Bracket}[1]{\Big[#1\Big]}
\newcommand{\bracket}[1]{\left[#1\right]}
\newcommand{\paren}[1]{\ensuremath{\left(#1\right)}\xspace}
\newcommand{\card}[1]{\left\vert{#1}\right\vert}

\newcommand{\IN}{\ensuremath{\mathbb{N}}}

\newcommand{\norm}[1]{\ensuremath{\|#1\|}}

\newcommand{\set}[1]{\ensuremath{\left\{ #1 \right\}}}
\newcommand{\poly}{\mbox{\rm poly}}

\DeclareMathOperator*{\Exp}{\ensuremath{{\mathbb{E}}}}
\DeclareMathOperator*{\Prob}{\ensuremath{\textnormal{Pr}}}
\renewcommand{\Pr}{\Prob}

\newcommand{\Ex}{\Exp}

\newenvironment{ourbox}{\begin{mdframed}[hidealllines=false,innerleftmargin=10pt,backgroundcolor=white!10,innertopmargin=10pt,innerbottommargin=5pt,roundcorner=10pt]}{\end{mdframed}}

\newenvironment{tbox}{\begin{tcolorbox}[
		enlarge top by=5pt,
		enlarge bottom by=5pt,
		 breakable,
		 boxsep=0pt,
                  left=4pt,
                  right=4pt,
                  top=10pt,
                  arc=0pt,
                  boxrule=1pt,toprule=1pt,
                  colback=white
                  ]
	}
{\end{tcolorbox}}

\newcommand{\event}{\ensuremath{\mathcal{E}}}

\newcommand{\rv}[1]{\ensuremath{{\mathsf{#1}}}\xspace}
\newcommand{\rA}{\rv{A}}
\newcommand{\rB}{\rv{B}}
\newcommand{\rC}{\rv{C}}
\newcommand{\rD}{\rv{D}}

\newcommand{\rG}{\rv{G}}
\newcommand{\rProt}{\rv{\Prot}}
\newcommand{\rR}{\rv{R}}

\newcommand{\supp}[1]{\ensuremath{\textnormal{\text{supp}}(#1)}}
\newcommand{\distribution}[1]{\ensuremath{\textnormal{dist}(#1)}\xspace}

\newcommand{\kl}[2]{\ensuremath{\mathbb{D}(#1~||~#2)}}
\newcommand{\II}{\ensuremath{\mathbb{I}}}
\newcommand{\HH}{\ensuremath{\mathbb{H}}}
\newcommand{\mi}[2]{\ensuremath{\def\mione{#1}\def\mitwo{#2}\mireal}}
\newcommand{\mireal}[1][]{
  \ifx\relax#1\relax%
    \II(\mione \,; \mitwo)%
  \else%
    \II(\mione \,; \mitwo\mid #1)%
  \fi
}
\newcommand{\en}[1]{\ensuremath{\HH(#1)}}

\newcommand{\itfacts}[1]{\Cref{fact:it-facts}-(\ref{part:#1})\xspace}

\newcommand{\cc}[1]{\ensuremath{\textsc{CC}(#1)}}

\newcommand{\ic}[2]{\ensuremath{\textsc{IC}(#1,#2)}}

\newcommand{\PP}{\ensuremath{\mathbb{P}}}
\newcommand{\cP}{\ensuremath{\mathcal{P}}}
\newcommand{\GG}{\ensuremath{\mathcal{G}}}
\newcommand{\cH}{\ensuremath{\mathcal{H}}}
\newcommand{\cA}{\ensuremath{\mathcal{A}}}
\newcommand{\cB}{\ensuremath{\mathcal{B}}}
\newcommand{\cHstar}{\ensuremath{\mathcal{H}^*}}
\newcommand{\startvert}{\ensuremath{\textnormal{\textsc{start}}}}
\newcommand{\finalvert}{\ensuremath{\textnormal{\textsc{final}}}}

\newcommand{\embed}{\ensuremath{\textnormal{\textsc{embed}}}}

\newcommand{\prot}{\ensuremath{\pi}}
\newcommand{\Prot}{\ensuremath{\Pi}}

\newcommand{\raoconst}{\ensuremath{c_{\textnormal{comp}}}}
\newcommand{\probconst}{\ensuremath{\rho}}

\newcommand{\reald}{\ensuremath{\GG^{\textnormal{\textsc{real}}}}}
\newcommand{\faked}{\ensuremath{\GG^{\textnormal{\textsc{fake}}}}}

\newcommand{\kstar}{\ensuremath{k^{\star}}}

%% file: abstract.tex
\begin{abstract}

\bigskip

In the semi-streaming model for processing massive graphs, an algorithm makes multiple passes over the edges of a given $n$-vertex graph and is tasked with computing the solution to a problem using $O(n \cdot \poly\!\log{\!(n)})$ space.
Semi-streaming algorithms for Maximal Independent Set (MIS) that run in $O(\log\log{n})$ passes have been known for almost a decade, however, the best lower bounds can only rule out single-pass algorithms. 
We close this large gap by proving that the current algorithms are \emph{optimal}:
{Any semi-streaming algorithm for finding an MIS with constant probability of success requires $\Omega(\log\log{n})$ passes.}
This settles the complexity
of this fundamental problem in the semi-streaming model, and constitutes one of the first optimal multi-pass lower bounds in this model. 

\medskip

We establish our result by proving an \textbf{optimal round vs communication tradeoff} for the (multi-party) communication complexity of MIS. 
The key ingredient of this result is a new technique, called \textbf{hierarchical embedding}, for performing round elimination: we show how to pack \emph{many} but \emph{small} hard $(r-1)$-round instances of the problem
into a single $r$-round instance, in a way that enforces any $r$-round protocol to effectively solve all these $(r-1)$-round instances also. 
These embeddings are obtained via a novel application of results from extremal graph theory---in particular dense graphs with many disjoint unique shortest paths---together with a newly designed graph product, 
and are analyzed via information-theoretic tools such as direct-sum and message compression arguments.

\end{abstract}

%% file: intro.tex

\section{Introduction}\label{sec:intro}

In the semi-streaming model for processing graphs, the edges of an
$n$-vertex graph $G = (V, E)$ are presented to an algorithm one-by-one in an arbitrarily ordered stream.
A semi-streaming algorithm then is allowed to make one or few passes over this stream and use $\Ot(n) := O(n \cdot \poly\!\log{\!(n)})$ space 
to solve a given problem. The semi-streaming model has been at the forefront of research on processing massive graphs since its introduction in~\cite{FeigenbaumKMSZ05} almost two decades ago. 

We study the \emph{Maximal Independent Set (MIS)} problem, namely, finding \emph{any} independent set of the graph that is not a proper subset of another independent set. 
An $O(\log{n})$-pass semi-streaming algorithm for MIS follows from Luby's parallel algorithm~\cite{Luby85} (see also~\cite{LattanziMSV11,KumarMVV13}).  
This was improved to an $O(\log\log{n})$-pass algorithm in~\cite{AhnCGMW15}  (see also~\cite{GhaffariGKMR18,Konrad18}). Despite significant attention, this has remained the state of the art for almost a decade now. 
At the same time, the only streaming lower bounds known for MIS are the $\Omega(n^2)$ space lower bounds for \emph{one}-pass algorithms obtained 
independently in~\cite{AssadiCK19a,CormodeDK19}.\footnote{There is also an $\Omega(n^{1/5})$-pass lower bound for semi-streaming algorithms that compute the \emph{lexicographically first MIS (LFMIS)}~\cite{AssadiCK19b}; 
however, it is known that LFMIS is a much harder (and quite different) problem than MIS (in most settings, including semi-streaming) and thus this result is not related to our discussion for finding any MIS.} 

 We prove that the $O(\log\log{n})$ passes in the algorithm of~\cite{AhnCGMW15} is optimal. 
\begin{result}[Formalized in~\Cref{cor:mis-stream}]\label{res:main}
	For any $p \geq 1$, any $p$-pass streaming algorithm for finding any maximal independent set of  $n$-vertex graphs with constant success probability requires $n^{1+1/(2^{p}-1)-o(1)}$ space. 
	In particular, semi-streaming algorithms require $\Omega(\log\log{n})$ passes. 
\end{result}
\vspace{-0.3cm}

\Cref{res:main} fully settles the pass-complexity of MIS in the semi-streaming model. This answers a fundamental open question in the graph streaming literature---see, e.g.,~\cite{CormodeDK19,Dark20}---on whether 
it is possible to even prove  \emph{any} multi-pass lower bound for MIS -- our lower bound now matches, up to $n^{o(1)}$ factors, the tradeoff of the algorithm of~\cite{AhnCGMW15} for \emph{every} number of passes.  

We establish~\Cref{res:main} by proving a stronger \textbf{rounds vs communication tradeoff} for MIS in the standard (multi-party) communication model. Here, the input graph is \emph{edge}-partitioned between multiple 
players. In each round, player one sends a message to player two, who sends a message to player three, and they continue like this until the last player, who sends a message back to the first one. 
We would like the message of the last player in the last round to reveal an MIS of the input graph (see~\Cref{sec:cc} for a formal definition). 
It is a well-known fact that communication lower bounds in this model also imply streaming lower bounds (see \Cref{prop:cc-stream}). 
\begin{result}[Formalized in~\Cref{thm:mis-cc}]\label{res:cc}
	For any $r \geq 1$, any $r$-round $(r+1)$-party protocol for finding any maximal independent set of $n$-vertex graphs with constant success probability requires $n^{1+1/(2^{r}-1)-o(1)}$ communication.
\end{result}
\vspace{-0.3cm}

The bounds in~\Cref{res:cc} are again optimal, up to $n^{o(1)}$ factors, in light of the algorithms in~\cite{AhnCGMW15,GhaffariGKMR18,Konrad18}. 
\Cref{res:main} now follows immediately from~\Cref{res:cc} using the standard reduction between the two models. 
Given the central role MIS 
plays in most distributed communication models---wherein communication is often a bottleneck---our result on the communication complexity of MIS is of its own independent interest.  

\subsection{Our Techniques}

We shall go over our techniques in detail in the streamlined overview of our approach in~\Cref{sec:overview}. For now, we only mention the high level ideas behind our techniques.

Our main approach in proving~\Cref{res:cc}---which implies~\Cref{res:main} immediately---is by finding a way to generalize and adapt the natural and easy-to-work-with ``round elimination lemma''-type arguments of~\cite{MiltersenNSW95} 
to \emph{graph} problems.  We achieve this using a new technique which we call \textbf{hierarchical embedding} that consists of two separate parts.

The first part is of \emph{combinatorial} nature. We design a new family of extremal graphs that have {many} disjoint {induced} subgraphs that are essentially---but not (necessarily) entirely---{smaller} copies of the same ``outer'' graph (informally, 
the outer graph is repeated many times as a smaller induced copy inside itself). 
These graphs form a considerable generalization of \emph{Ruzsa-Szemer{\'e}di (RS)} graphs~\cite{RuzsaS78} that, starting from~\cite{GoelKK12}, have found numerous applications in 
proving streaming and communication lower bounds, e.g., in~\cite{Kapralov13,AssadiKLY16,AssadiKL17,AssadiKO20,Kapralov21,CormodeDK19,KonradN21,AssadiR20,ChenKPSSY21,Assadi22,AssadiS23,KonradK24}.
We construct these graphs via a new \emph{graph product} applied to a (variant of) another family of extremal graphs studied in the literature on hopsets and spanners under the terms 
``directed graphs requiring large number of shortcuts''~\cite{Hesse03} or ``graphs with many long disjoint shortest paths''~\cite{AbboudB16} (see, e.g.,~\cite{CoppersmithE06,HuangP21,BodwinH22,LuWWX22,WilliamsXX23}). 

The second part is of  \emph{information-theoretic} nature. We use  \emph{information complexity}~\cite{ChakrabartiSWY01,BarakBCR10} and \emph{direct-sum} results to generalize the $\Theta(n)$-party {simultaneous} communication lower bounds of~\cite{AlonNRW15,AssadiKZ22} 
to the $O(\log\log{n})$-party communication model that we study. This involves dealing with protocols with {more interaction} in each round and a much larger bandwidth \emph{per} player. 
One component of this part is a \emph{message compression} technique due to~\cite{HarshaJMR07} (see also~\cite{BravermanG14}), which we use to handle protocols in our arguments that are \emph{statistically close} to having low information cost, but may not
have an actually low information cost themselves. 

Finally, we note that our hierarchical embedding approach adapts and generalizes the very recent work of \cite{KonradK24} for proving \emph{two}-pass lower bounds
for approximating matchings: in the combinatorial part, we allow for embedding a much richer family of graphs compared to those of~\cite{KonradK24} by going beyond RS graphs;
in the information-theoretic part, we use entirely different arguments that allow for proving lower bounds for any large number of passes and not just two.

\subsection{More Context and Related Work}

MIS is a fundamental problem with natural connections to many other classical problems, such as vertex cover, matching, and vertex/edge coloring (see, e.g.~\cite{Linial87}). 
Hence, MIS has been studied extensively in most models of computation on (massive) graphs, including 
LOCAL~\cite{Luby85,Linial87,BarenboimEPS12,KuhnMW16,Ghaffari16,BalliuBHORS19,DRozhonG20,GhaffariG23}, Dynamic graphs~\cite{AssadiOSS18,OnakSSW18,AssadiOSS19,BehnezhadDHSS19,ChechikZ19}, the Massively 
Parallel Computation model (MPC)~\cite{GhaffariGKMR18,DBehnezhadBDFHKU19,GhaffariU19}, Distributed Sketching~\cite{AssadiKO20,AssadiKZ22}, Local Computation Algorithms (LCAs)~\cite{RubinfeldTVX11,AlonRVX12,GhaffariU19,Ghaffari22}, 
Sublinear-time~\cite{NguyenO08,YoshidaYI09,AssadiS19,Behnezhad21}, and semi-streaming~\cite{AhnCGMW15,CormodeDK19,AssadiCK19a} (this is by no means a comprehensive list; see these papers for further references). 
We now mention some of these and related work that are most relevant to us. 

\subsubsection*{MIS as a subroutine and a ``barrier'' for other semi-streaming algorithms}
 The $O(\log\log{n})$-pass algorithm of \cite{AhnCGMW15} for MIS 
was designed as part of a semi-streaming implementation of the so-called \emph{Pivot} algorithm~\cite{AilonCN05} which computes the \emph{random order MIS} to approximate \emph{Correlation Clustering} 
(improving upon~\cite{ChierichettiDK14}; see also~\cite{BlellochFS12,FischerN18}). In general, computing MIS, and, in particular, random order MIS, is a highly useful subroutine in various algorithmic problems; see, e.g.~\cite{YoshidaYI09,BlellochFS12,BehnezhadDHSS19,Behnezhad21,HaqiZ23,CambusKLPU24}. 
At the same time, this also meant that for many of these
problems, the best bounds remained stuck at the same $O(\log\log{n})$ passes of the MIS. 

More recently, there have been several attempts to bypass this ``barrier'' by relying on weaker subroutines than MIS. 
Most notably, for correlation clustering,~\cite{Cohen-AddadLMNP21,AssadiW22} designed ``non Pivot'' algorithms that use an entirely different subroutine based on sparse-dense decompositions (at 
a cost of much larger approximations); and,~\cite{BehnezhadCMT22,BehnezhadCMT23,CambusKLPU24} designed semi-streaming algorithms for more relaxed versions of Pivot  that can be 
implemented in much fewer passes. 
Another example of similar nature is the $O(1)$-pass semi-streaming algorithm of~\cite{CambusKPU23} for a relaxation of MIS called the \emph{$2$-ruling set}%
\footnote{For every $\beta \geq 1$, a $\beta$-ruling set of a graph $G$ is an independent set such that every other vertex is within distance $\beta$ of some vertex in the ruling set. Thus, an MIS is a $1$-ruling set.}, 
which can replace MIS sometimes, e.g., in metric facility location~\cite{BernsHP12} (see~\cite{KonradPRR19,AssadiD21} for other streaming algorithms for ruling sets). 

Our lower bound in~\Cref{res:main} now definitively confirms the necessity of relying on these relaxations (and their accompanied complications), as there is provably no way of implementing \emph{any} MIS-based algorithm
in the semi-streaming model in fewer than $O(\log\log{n})$ passes. 

\subsubsection*{Distributed lower bounds for MIS}
 There is a quite large body of work on proving lower bounds for MIS in distributed models that focus on \emph{locality} (instead of communication cost and message lengths). This is a topic orthogonal to ours
and we instead refer the reader to~\cite{KuhnMW16,BalliuBHORS19,Suomela20}. We only mention that these
results are also based on a technique called \emph{round elimination} which is conceptually similar to communication complexity round eliminations (as in~\cite{MiltersenNSW95} or our paper)
but technically they appear to be entirely disjoint (see~\cite{Suomela20} for more details on distributed round elimination).  

Much more closer to our work is the distributed sketching (a.k.a., \emph{broadcast} Congested Clique) lower bound of~\cite{AssadiKZ22} (building on~\cite{AlonNRW15}). 
In their model, there is a player for each vertex of the graph that can only see incident edges of this vertex (so $n$ players in total). In each round, each player, 
 \emph{simultaneously} with others, can send $\poly\!\log{\!(n)}$-size messages which will be seen by everyone in the next round. 
 \cite{AssadiKZ22} proved that $\Omega(\log\log{n})$ rounds of communication is needed in this model to compute an MIS. 
 The distributed sketching model is algorithmically weaker than the communication model of~\Cref{res:cc} in three aspects: $(i)$ (much) more players, $(ii)$ much less interaction per round, and, $(iii)$ the more stringent requirement of $\poly\!\log{\!(n)}$-size messages \emph{per} vertex in \emph{worst case}, as opposed to \emph{on average}, namely, $\Ot(n)$ communication per round in total. In particular, 
 this model is not even strong enough to implement the MIS algorithms of~\cite{AhnCGMW15,GhaffariGKMR18,Konrad18} (primarily due to condition $(iii)$; see~\cite{AssadiKZ22}), and the best known upper bound in this model for MIS is still the $O(\log{n})$ rounds of Luby's algorithm~\cite{Luby85}.
 
 Given the above, the lower bounds in~\cite{AssadiKZ22} do \emph{not} imply any streaming lower bounds\footnote{This is in a strong sense; for instance,~\cite{AssadiKZ22} also proves an $\Omega(\log\log{n})$-round
 distributed sketching lower bound for \emph{maximal matching}, despite this problem admitting a simple \emph{one}-pass semi-streaming algorithm~\cite{FeigenbaumKMSZ05}.}. Nonetheless, 
 they \emph{do} rule out a certain restricted family of semi-streaming algorithms: these are algorithms that compute a linear sketch of the neighborhood of each vertex \emph{individually}, 
 as was introduced in~\cite{AhnGM12} to design algorithms on \emph{dynamic} streams with both edge insertions and deletions. 
As such,~\cite{AssadiKZ22} posed the question of whether their results can be generalized to all \emph{dynamic} streaming algorithms. Our~\Cref{res:main} is now 
answering an even more stronger version of the question by proving the lower bound directly for \emph{insertion-only} streams. 

\subsubsection*{MIS in MPC and Congested-Clique models} 

The key technique behind the semi-streaming MIS algorithm of~\cite{AhnCGMW15} was proving the ``residual sparsity property'' of the \emph{randomized} greedy MIS: computing MIS of a ``small'' random subset of vertices allows for sparsifying the graph significantly. 
This technique has been quite influential for both designing other semi-streaming algorithms, e.g., for approximating matchings~\cite{Konrad18b,AzarmehrB24}, 
as well as for designing algorithms for MIS in other models, e.g.~\cite{AssadiOSS19,BehnezhadDHSS19,ChechikZ19}
for dynamic graphs,~\cite{GhaffariGKMR18,Konrad18} for MPC or Congested-Clique algorithms. 

In particular, the current best MPC and Congested-Clique algorithms~\cite{GhaffariGKMR18,Konrad18} for MIS are {direct} translation of the algorithm of~\cite{AhnCGMW15} to these models. 
Similar to the semi-streaming model, it is also an important open question in the area to obtain more efficient algorithms for MIS in these models. 
Proving unconditional lower bounds in MPC and Congested-Clique models imply strong circuit lower bounds (see, e.g.~\cite{DruckerKO14,RoughgardenVW16}) and is thus beyond the reach of current techniques. But, our~\Cref{res:main} can still act as a strong \emph{guide} here: any  algorithm in these models for breaking the $O(\log\log{n})$-round barrier must exploit the ``extra'' power of these models over semi-streaming algorithms and should not be implementable in the semi-streaming model (unlike many MPC and Congested-Clique algorithms that do have semi-streaming counterparts).  

\subsubsection*{The classic four local graph problems in the semi-streaming model} 

There has been a growing interest in studying ``locally checkable'' graph problems in the semi-streaming model, beyond their origins in distributed computing; see, e.g.~\cite{AssadiCK19a,KonradPRR19,AssadiD21,CambusKPU23,FlinGHKN24} 
and references therein. In particular, the (pass-)complexity of the  `classic four local graph problems'  is as follows: \textbf{Maximal matching}: a \emph{one}-pass semi-streaming algorithm via the greedy algorithm~\cite{FeigenbaumKMSZ05};
\textbf{$(\Delta+1)$-vertex coloring}: a \emph{one}-pass semi-streaming algorithm via palette sparsification~\cite{AssadiCK19a}; \textbf{$(2\Delta-1)$-edge coloring}: this problem is not well-defined in the semi-streaming model given its output size is more than the allowed space (although, there have been exciting recent work on this problem in the \emph{W-streaming} model
that augments the standard semi-streaming model with a \emph{write-only} tape~\cite{BehnezhadDHKS19,BehnezhadS23,ChechikMZ23,GhoshS23});
and finally, \textbf{Maximal independent set}: our~\Cref{res:main} combined with the algorithm of~\cite{AhnCGMW15} completes the picture for MIS and establishes $\Theta(\log\log{n})$ passes 
as its pass-complexity. Consequently, we now have a complete understanding of these four classical problems in the semi-streaming model. 

\subsubsection*{Optimal multi-pass streaming lower bounds} 

Finally, we remark that there has been tremendous progress recently in proving multi-pass graph streaming lower bounds~\cite{GuruswamiO16,AssadiCK19b,AssadiR20,ChenKPSSY21,ChenKPSSY21b,Assadi22,AssadiS23,KonradK24}; see~\cite{Assadi23} for a recent survey of these results. Yet, our techniques in~\Cref{res:main} are the \emph{first} ones that allow for proving (even nearly) \emph{optimal} pass lower bounds for semi-streaming algorithms in this
line of work. For instance,~\cite{GuruswamiO16,ChenKPSSY21} prove $\Omega(\log{n})$ pass lower bounds for reachability or perfect matching, and~\cite{AssadiS23} proves $\Omega(\log{(1/\eps)})$-passes for $(1-\eps)$-approximate matchings (conditionally). However, the best known algorithms require $n^{1/2+o(1)}$ and $n^{3/4+o(1)}$ passes for the first two problems~\cite{AssadiJJST22}, and $O(1/\eps^2)$~\cite{AssadiLT21} or $O(\log{(n)}/\eps)$~\cite{AhnG18,Assadi24} passes for the second one. 
We hope our techniques pave the path for proving other \emph{optimal} multi-pass lower bounds as well.

\section*{\emph{A note on the presentation of this paper}}
Given the generality---and conceptual simplicity---of our approach, we believe the ideas in this paper can be of general interest beyond the semi-streaming model. 
As such, to enhance the readability of our paper and make it more accessible to non-experts, we have provided ample intuition, discussions, and figures throughout, which has contributed significantly to the length of the paper.
However, the main parts of the proofs can be stated fairly succinctly. In particular, to obtain a complete understanding of all the key technical details, a reader familiar with the background on multi-pass semi-streaming lower bounds
can directly jump to~\Cref{sec:dup-embed} (ignoring~\Cref{sec:dup-construct} in the first read) and then check~\Cref{sec:cc-lower} followed by~\Cref{sec:analysis} until
the beginning of~\Cref{sec:proof-communication-lb}.

%% file: overview.tex
\section{Technical Overview}\label{sec:overview}

We now present a streamlined overview of our technical approach. As stated earlier, our main result is a (multi-party) round vs communication lower bound for MIS in~\Cref{res:cc}; we obtain
our semi-streaming lower bound from this result immediately, using the standard connection between the two models (see~\Cref{sec:cc}). Thus, in this section, 
we solely focus on the communication complexity of MIS and postpone the streaming result to the formal proofs.
We emphasize that this section
oversimplifies many details and the discussions will be informal for the sake of intuition. 

We start by presenting the intuition on what would have been the \emph{ideal} lower bound approach for this problem. We then discuss 
how we are \emph{actually} implementing this ideal approach and discuss our main technical ingredient, namely, \emph{hierarchical embeddings}, that enables this implementation.  

\subsection{An ``{Ideal}'' Lower Bound Approach}\label{sec:ideal}

Almost all \emph{round-sensitive} communication lower bounds, at a high level, rely on some form of \textbf{round elimination} arguments (see, e.g.~\cite[Section 1.3]{MiltersenNSW95}): one shows that a too-good-to-be-true $r$-round protocol 
can be used to also create a too-good-to-be-true $(r-1)$-round protocol, and continue this until ending up with a $0$-round protocol which does something non-trivial, a contradiction. These arguments can take
various forms, wherein the number of players, the domains of the inputs, or even the underlying problems may change (dramatically) from one round to another.  

One of the simplest approaches here is the following ``Round Elimination Lemma'' of~\cite{MiltersenNSW95}: an $r$-round hard instance of the problem is created by picking multiple $(r-1)$-round hard instances 
of the same problem $I_1,\ldots,I_k$ given to the players. The players also receive additional inputs that 
``point'' to one of these instances $I^* \in \set{I_1,\ldots,I_k}$ as being the \textbf{special} one in a way that, to obtain the final answer, they need to solve $I^*$. However, 
this pointer is hidden to the players at first (e.g., is only given to the player not speaking first) and so they cannot reveal enough information about $I^*$ in their first round; thus, they now need 
to solve $I^*$ which is a hard $(r-1)$-round instance in $(r-1)$ rounds, with ``almost no help'' from the first round, implying the lower bound.

Let us now see the \emph{one-round} lower bounds for MIS in~\cite{AssadiCK19a,CormodeDK19} in this context. 

\paragraph{Prior one-round lower bounds.} The proofs in~\cite{AssadiCK19a,CormodeDK19} are as follows. Alice receives \emph{two {identical} copies} of a random bipartite graph $G$ with $\Theta(n^2)$ edges.
Bob receives the following graph on the same set of vertices:  pick a pair $\set{u,v}$ of vertices on the two sides of the bipartition of $G$ and connect all vertices that are \emph{not} a copy of these vertices to each other. 
See~\Cref{fig:overview1}. Computing the MIS of such a graph requires the players to figure out if $(u,v)$ is an edge in $G$ or not. This is because picking any vertex $w$ other than $u$ or $v$ in the MIS ``forces out'' all the remaining vertices (possibly even the copies of $u$ and $v$ 
in the \emph{same} copy as $w$), except for $u,v$ of the other copy, that are now the only candidates to join the MIS (again, see~\Cref{fig:overview1}). 

\begin{figure}[h!]
\centering
\includegraphics[scale=0.4]{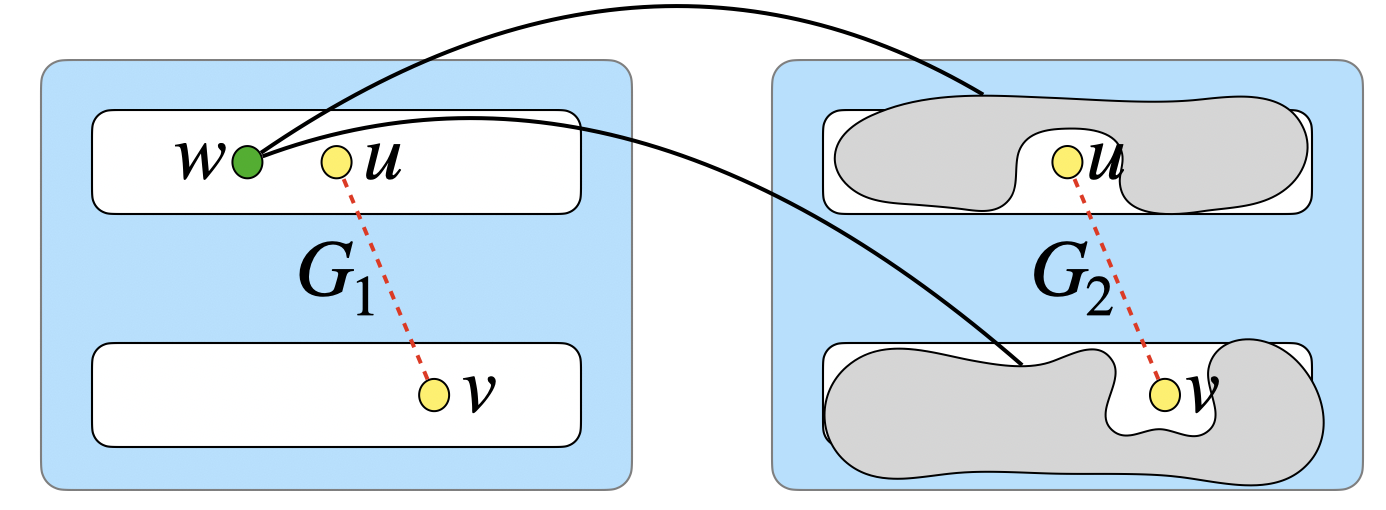}
\caption{Alice's input is $G_1$ and $G_2$ which are copies of the same  random bipartite graph $G$. Bob's input connects every vertex other than copies of $u$ and $v$ to each other. If we 
pick a vertex $w$ in the left to join the MIS, only vertices $u$ and $v$ in the right remain to join the MIS (the figure only shows some of $w$'s edges). 
}\label{fig:overview1}
\end{figure}

We can now do a round elimination to prove the lower bound\footnote{The lower bounds in~\cite{AssadiCK19a,CormodeDK19} instead directly reduce the problem to the one-round communication
complexity of the \emph{Index} problem~\cite{Ablayev93,KremerNR95}; but, they can also be proven in this (almost equivalent) way.}. Think of the $0$-round ``hard'' problem as finding the MIS of a graph on a pair of vertices which may or may not be connected. 
The $1$-round hard problem above consists of $\Theta(n^2)$ copies of this problem, with Bob's input making exactly one of them, corresponding to the pair $\set{u,v}$, the special one; for the players to solve this $0$-round problem, 
Alice should have communicated enough information about $\set{u,v}$, which is not possible with $o(n^2)$ size messages, given Alice does not know which of the $\Theta(n^2)$ instances is special. 

\paragraph{Beyond one-round lower bounds.} Suppose we want to extend this approach beyond a single round. We can use the same approach of using two identical copies of the same \textbf{base graph} for Alice and 
give Bob an input that connects these copies. Instead of a random graph however, we need Alice's base graph to be a collection of \emph{independent} copies of the $1$-round hard instances. Here, `independent' copies not only mean in a probabilistic sense, but also in a graph-theoretic sense: the edges of one instance should not appear between vertices of another as otherwise the instances become ``corrupted'', namely, they are no longer necessarily hard $1$-round instances. On top of this, we also need these instances to be quite \emph{large}, so that a $1$-round lower bound for them incurs a significant communication cost in terms of the \emph{entire} graph size (this is 
not a problem for the $1$-round case as a $0$-round protocol cannot have any communication no matter the input size). 

A bit more specifically, to make this strategy work even for just a two-round $n^{1+\eps}$ communication lower bound, for any constant $\eps > 0$, we need Alice's base graph to have the following properties: 
\begin{enumerate}[itemsep=0pt, leftmargin=8pt]
\item It must contain $(n^{1+\eps})$ $1$-round instances so Alice cannot reveal too much about a random instance; 
\item Each $1$-round instance needs to have $(n^{(1+\eps)/2})$ vertices so that the quadratic $1$-round lower bound implies that the second round of the 
protocol needs $(n^{(1+\eps)/2})^2 = n^{1+\eps}$ communication; 
\item The instances cannot have edges between each others' vertices so as not to ``corrupt'' one another. 
\end{enumerate}

Given that one-round instances have quadratic number of edges, satisfying all these constraints means having a base graph on $n$ vertices with $n^{2+2\eps}$ edges or even pairs of vertices!\footnote{Making 1-round instances sparser or similar modifications
does not work either, as that forces us to pick even larger $1$-round instances in Line 2, leading to the same exact contradiction.}

This highlights an inherent difficulty (basically, impossibility) of implementing this natural type of round elimination argument for proving communication lower bounds on graphs. 
In fact, almost all prior round-sensitive communication lower bounds on graphs (which are super linear in $n$) rely on more complicated arguments.
For instance, the recent lower bounds in~\cite{AssadiR20,ChenKPSSY21,Assadi22,AssadiS23} 
require proving a lower bound for much harder intermediate problems to be able to carry out the inductive argument (e.g., see \emph{permutation hiding generators} in~\cite{ChenKPSSY21,AssadiS23}; see also~\cite{Assadi23} for 
a summary of these recent advances). The only exception is the very recent work of~\cite{KonradK24} that shows how to extend a similar one-round lower bound approach (but, for the matching problem in~\cite{GoelKK12})
to \emph{two}-rounds, and prove a two-pass semi-streaming lower bound for approximate matching (we mention similarities and differences of these two approaches throughout this section).

\subsection{Our \emph{Actual} Lower Bound Approach: Hierarchical Embeddings}\label{sec:actual}

Even though the above approach seems quite hopeless, 
we show how to actually implement it
modulo a crucial twist: we will not only start with many hard instances of the $(r-1)$-round problem in our $r$-round instance, but we will 
get the players to also solve \emph{many} of these instances in the subsequent rounds (this is inspired by the aforementioned two-round lower bound of~\cite{KonradK24}). Specifically, we have $q \cdot p$ ``hard'' $(r-1)$-round  \emph{sub-instances} $\set{I_{i,j}}_{i \in [q] , j \in [p]}$, each on $b$ vertices for some \emph{large} parameters $q$, $p$, and $b$ as a function of $n$. These instances are again created based on a single base graph which is copied twice and given to the players; they are then connected via a clique structure, which identifies a \textbf{sub-instance group} $I_{t,*} := I_{t,1},\ldots,I_{t,p}$ for a random chosen $t \in [q]$ 
as special instances. Restricting ourselves again to a two-round $n^{1+\eps}$-communication lower bound, what we need from this base graph is the following:
\vspace{-0.15cm}
\begin{enumerate}[itemsep=0pt, leftmargin=10pt]
\item We need $q \cdot p = n^{1+\eps}$ so Alice's first message does not reveal too much information; 
\item We also need $p \cdot b^2 = n^{1+\eps}$ so Bob's message cannot solve \emph{all} $p$ special instances (each on $b$ vertices admitting a $1$-round $\Omega(b^2)$-communication lower bound) in the second round; 
\item \textbf{Inducedness property:} None of $q \cdot p$ sub-instances can have an edge between vertices of another sub-instance so as to not ``corrupt'' one another. Also, each sub-instance group $I_{i,*}$ for $i \in [q]$ should be \emph{vertex disjoint}
so that finding an MIS of their vertices requires solving all sub-instances. 
\end{enumerate}

Satisfying these constraints requires having a base graph on $n$ vertices with $q \cdot p \cdot b^2$ vertex-pairs. \emph{If} we could set $p \approx q \approx n/b$, then, at least we will pass the \emph{basic} test of ``vertex-pair counting''. This in turn gives us a 
base graph with $\approx n^2/b^2$ sub-instances among which $\approx n / b$ are special (each admitting a 1-round $\Omega(b^2)$-communication lower bound). Optimizing the parameters by setting $n^2/b^2 = (n/b) \cdot b^2 = n^{1+\eps}$, implies that 
we can \emph{hope for} obtaining an $\approx n^{4/3}$ lower bound (for $b \approx n^{1/3}$) for two-round algorithms this way.

\begin{Remark}\label{rem:loose-motivation}
\vspace{-5pt}
{This is precisely the tradeoff achieved by the protocols for MIS in~\cite{AhnCGMW15,GhaffariGKMR18} in two rounds. In the first round, they sample
$n/b$ vertices and send all their edges in $n^2/b^2$ communication, compute their MIS, and remove their neighbors. They prove that this reduces the maximum degree of remaining vertices to $\approx b$ only, and so in the next 
round they can send the remaining edges in $\approx n \cdot b$ communication and compute an MIS. This approach
allows us to mimic the \textbf{same exact tradeoffs} for the lower bound for any number of passes.}
\end{Remark}

\vspace{-10pt}
Let us see how a base graph with these parameters should look like, especially given that the third requirement is quite strict: 
the induced subgraph of the base graph on vertices of $I_{i,*}$ for $i \in [q]$, should be a vertex-disjoint union of graphs in $I_{i,*}$, see~\Cref{fig:overview2}. 

\begin{figure}[h!]
\centering
\includegraphics[scale=0.55]{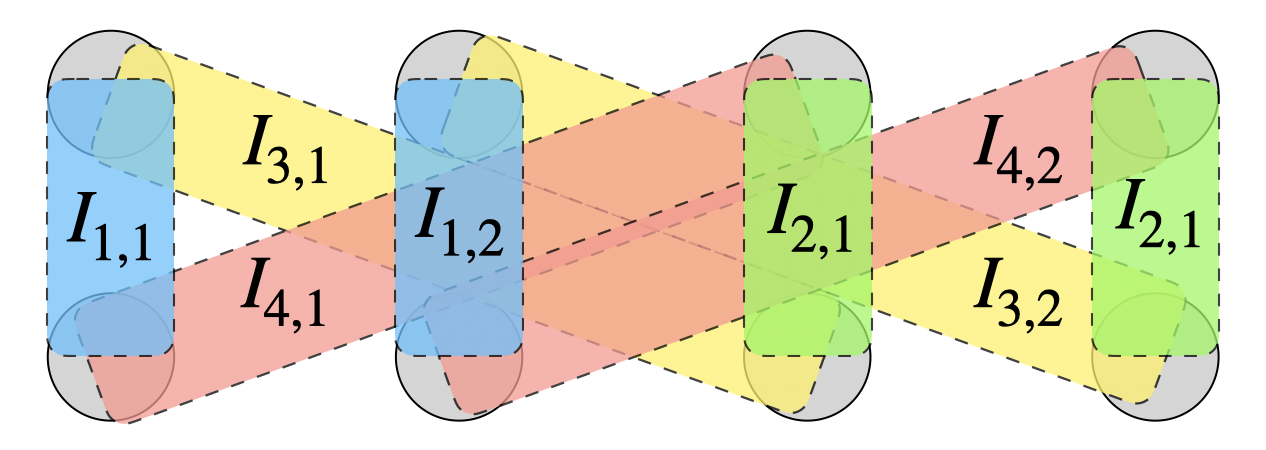}
\caption{An illustration of the base graphs discussed above. Here $q=4$ and $p=2$ and each box denotes a single sub-instance (with different colors for different $I_{i,*}$ for $i \in [4]$). 
There are no edges between vertices of $I_{i,*}$ except for edges of the same sub-instances, which themselves are vertex-disjoint. While the figure may suggest sub-instances as being bipartite graphs, that crucially cannot be true in the actual construction.}\label{fig:overview2}
\end{figure}

\paragraph{Combinatorial considerations.} Do such graphs even exist? Notice that inserting each group of sub-instances $I_{i,*}$ adds $p \cdot b^2 \approx n \cdot b$ edges to the graph, but prohibits 
roughly $\approx p^2 \cdot b^2 \approx n^2$ vertex-pairs to ever appear as an edge in the graph (to ensure the inducedness property, vertices of a single sub-instance in $I_{i,*}$ cannot have edges to any of the remaining $(p-1) \cdot b$ vertices of the sub-instances of $I_{i,*}$). This seems to suggest that by the time we have added barely a super-constant number of sub-instance groups $I_{i,*}$, i.e., when $q$ is becoming super-constant and way before $q \approx n/b$, we will run out of edges and
thus cannot continue. In other words, these graphs can only be sparse, and are hence not at all suitable for a super-linear-in-$n$ lower bound.  

The above argument however is flawed because the vertex-pairs {prohibited} by these sub-instance groups can be \emph{shared} among each other (as opposed to edges that cannot be shared). Indeed, 
succinctly stated, the  base graphs for $r$-round instances are disjoint collections of \emph{induced} subgraphs that are \emph{vertex-disjoint} unions of hard $(r-1)$-round instances. 
This definition is reminiscent of \emph{Ruzsa-Szemer{\'e}di (RS)} graphs~\cite{RuzsaS78}: these are graphs that consist of many disjoint \emph{induced matchings}.  
The same ``flawed reasoning'' also applies to RS graphs and yet, in reality, they can become quite dense, with even ${{n}\choose{2}} - o(n^2)$ edges and induced matchings of size $ p = n^{1-o(1)}$~\cite{AlonMS12}.

The difference with RS graphs for us is that we need \emph{induced subgraphs} that are much more complicated than matchings. One of our main contributions is an 
almost-optimal construction of these graphs for a large family of induced subgraphs (which can be any \emph{$n^{o(1)}$-colorable} graph).
Our construction relies on two separate ingredients: $(i)$ another family of extremal graphs referred to as 
``directed graphs requiring large number of shortcuts''~\cite{Hesse03} or ``graphs with many long disjoint shortest paths''~\cite{AbboudB16} studied extensively in the context of hopsets and spanners lower bounds
(see, e.g.,~\cite{CoppersmithE06,HuangP21,BodwinH22,LuWWX22,WilliamsXX23}); and, $(ii)$ a graph product, which we call \emph{embedding product}, that allows for ``packing'' a large number 
of complex induced subgraphs inside a single graph of part $(i)$. See~\Cref{sec:overview-dup} for a detailed overview.  

\paragraph{Information-theoretic considerations.} 
We have so far focused on the {combinatorial} aspects of our approach. The next step are the \emph{information theoretic} arguments for proving the lower bound. 
Here, our work can be seen as generalizing and unifying two previously disjoint sets of techniques for proving round-sensitive communication lower bounds (discussed in more detail in~\Cref{sec:overview-round-elimination}): 
\vspace{-5pt}
\begin{itemize}[itemsep=0pt, leftmargin=10pt]
	\item The round elimination arguments of~\cite{AlonNRW15} and~\cite{AssadiKZ22} for approximate matchings and MIS, respectively, 
	for quite different and ``algorithmically weaker'' models of communication (e.g., one cannot implement the MIS algorithms of~\cite{AhnCGMW15,GhaffariGKMR18} in these models). 
	These lower bounds in particular work with a much larger number of players (proportional to vertices in the graph) and do not allow interaction between the players in each single round. 
	
	\item The line of work, starting from~\cite{GoelKK12},  
	that use \emph{RS graphs} for proving semi-streaming lower bounds and in particular, a very recent result of~\cite{KonradK24}
	that proves a \emph{two}-pass lower bound for the matching problem. These prior works typically rely on different reductions to various communication problems---which are often different
	from the original underlying problem---say, \emph{Index} in~\cite{GoelKK12} and \emph{HiddenStrings} in~\cite{KonradK24}.
	Our approach, on the other hand, takes advantage of a ``self-reducibility'' property of the problem which allows for implementing a round elimination argument over multiple passes/rounds. 
\end{itemize}
\vspace{-12pt}
\paragraph{The hierarchical embedding technique.} Let us now reiterate how our \emph{hierarchical embedding} technique works at a high level. 
For technical reasons, our $r$-round lower bound will be against $(r+1)$-party (number-in-hand) protocols instead of just two parties\footnote{We suspect the lower bound also works for two players (and possibly without much modifications). But, since 
``few''-party protocols already imply semi-streaming lower bounds---and one needs more than two parties to achieve the sharp bounds on passes (see, e.g.~\cite{GuhaM08}) achieved in~\Cref{res:main}---we did not pursue that direction in this paper.}.

We pick $q_r \cdot p_r$ hard $(r-1)$-round sub-instances with $r$-players in a set $\mathcal{I}:=\set{I_{i,j}}_{i \in [q_r], j \in [p_r]}$, 
each on $b_r$ vertices, for parameters $q_r, p_r$, and $b_r$ to be determined soon.

We use our \textbf{combinatorial constructions} to {embed} all these instances into a {single} base graph $G$ 
with the inducedness property mentioned earlier (and then copied identically twice as before). For every $a \in [r]$, the input of the $a^{\textnormal{th}}$ player in the $r$-round instance is the collection of the inputs of all $a^{\textnormal{th}}$ players in sub-instances
in $\mathcal{I}$. The $(r+1)^{\textnormal{th}}$ player in the $r$-round instance gets a clique-subgraph as before that ``points'' to one of the sub-instance groups $I_{t,*}$ for a random $t\in [q_r]$. 

This way, the input consists of a \emph{hierarchy} of different-round instances of the problem, and the players have to {solve} multiple root-to-leaf paths of this hierarchy to obtain a solution to the entire input as well.  
Our \textbf{round elimination argument} then shows that the first round of the protocol cannot reveal enough information about the special sub-instance group to allow solving \emph{all} of them in the remaining $(r-1)$ rounds. 
See~\Cref{fig:hierarchical} for an illustration. 

\begin{figure}[h!]
	\centering
	\subcaptionbox{The hierarchical instances form a conceptual tree, and the players need to ``solve'' certain root-to-leaf paths.}%
	[1\linewidth]{
		\includegraphics[scale=0.3]{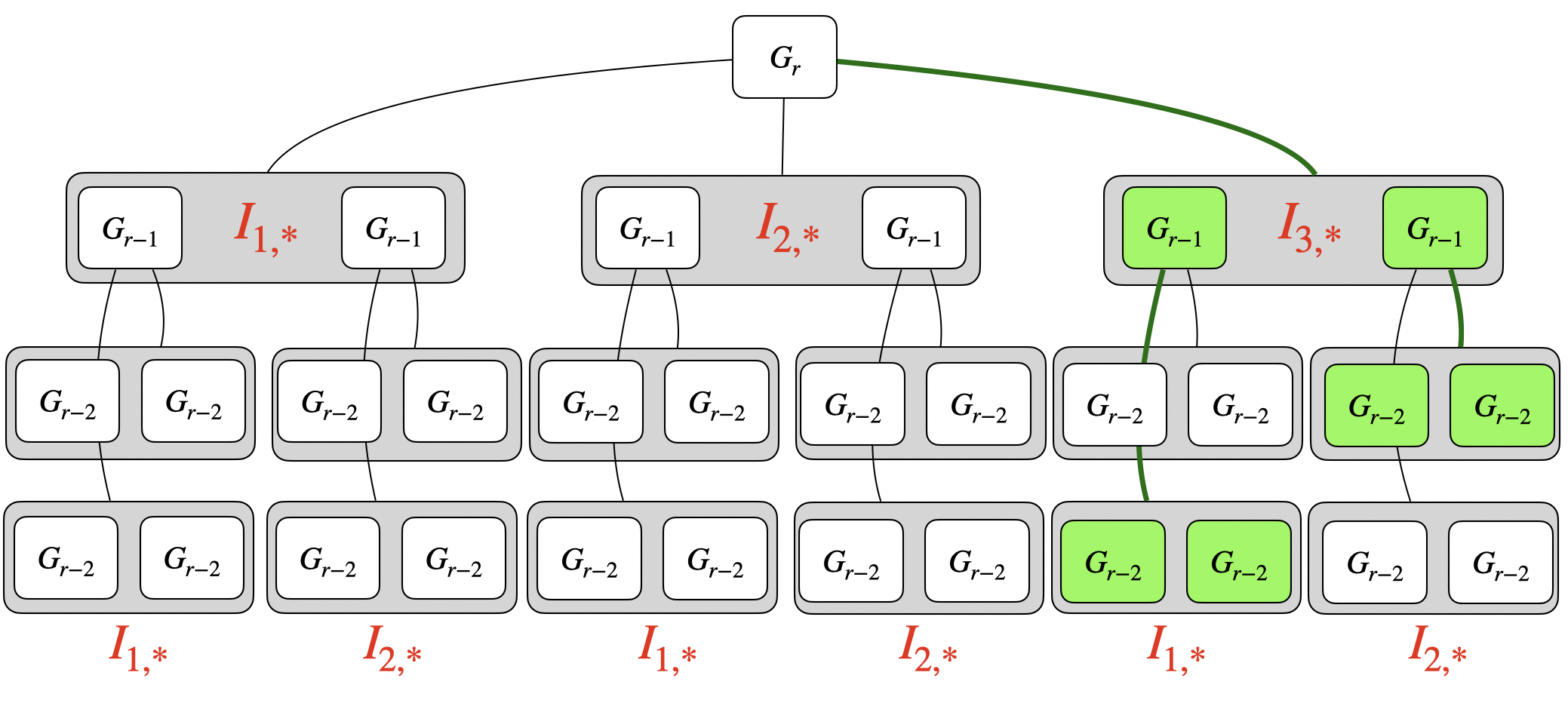}
	} 
	
	\subcaptionbox{The input graph is also a hierarchy of \emph{induced} subgraphs (this figure shows these graphs in a vertex-disjoint way for simplicity of drawing, but in reality, vertices of these graphs
	are highly ``tangled''). }%
	[1\linewidth]{
	\includegraphics[scale=0.28]{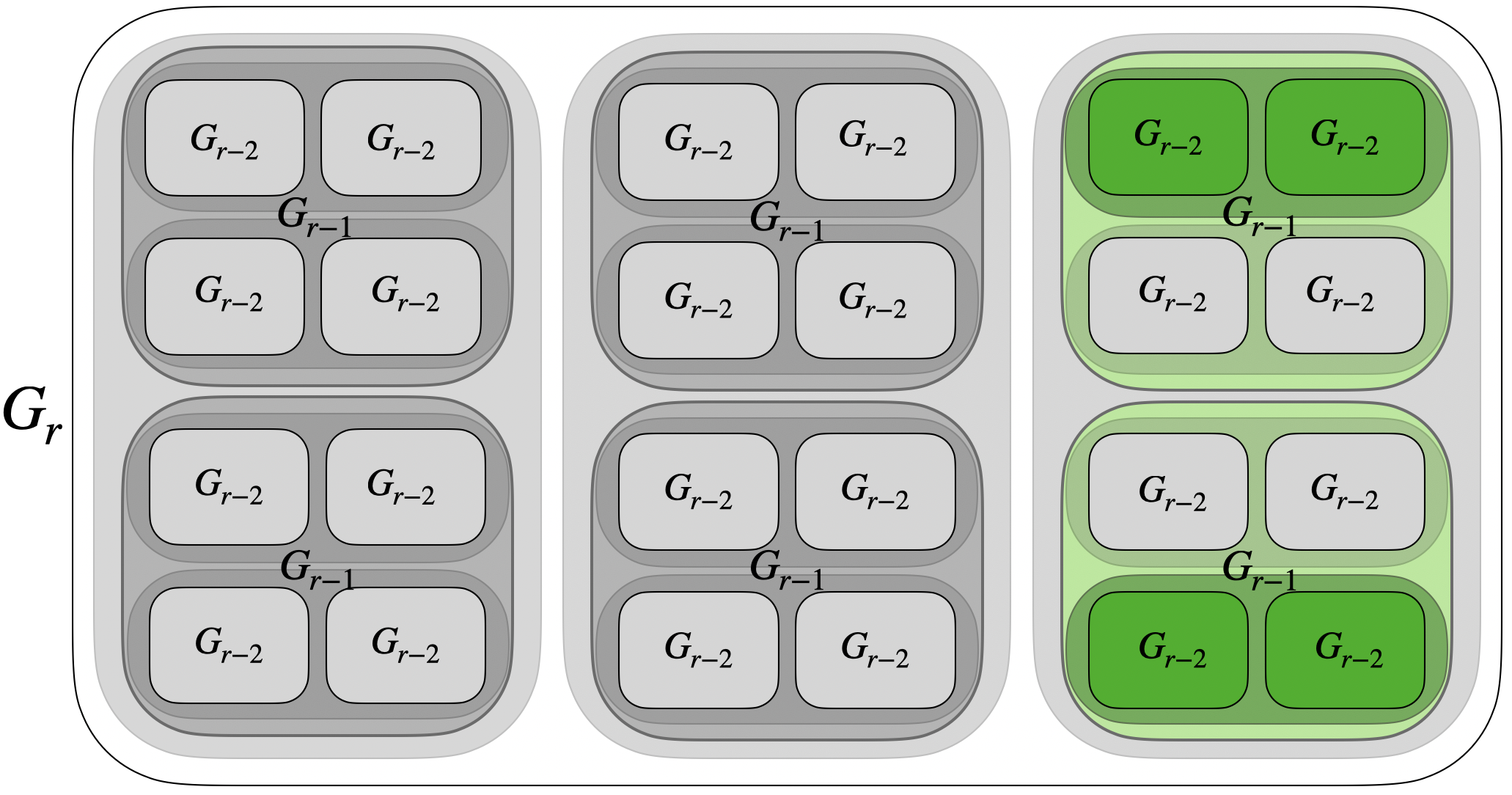}
	}
	\caption{An illustration of the hierarchical embedding technique.}
	\label{fig:hierarchical}
\end{figure}

Finally, the choice of the parameters $q_r,p_r$ and $b_r$ is as follows. Let $C(n,r)$ denote the lower bound on the communication cost we hope to prove for $r$-round protocols on $n$-vertex graphs. Then, 
\begin{align*}
	&q_r \cdot p_r \approx C(n,r): \text{so players cannot reveal much about $p_r$ special sub-instances in the first round}  \\
	&p_r \cdot C(b_r , r-1) \approx C(n,r): \text{so solving the $p_r$ special $b_r$-size sub-instances in $r-1$ rounds is hard}. 
\end{align*}
Picking these parameters optimally, while accounting for the restrictions imposed by the combinatorial construction (and an $n^{o(1)}$-``loss'' on size of these graphs compared to absolute best-imaginable bounds), 
gives us our $\approx n^{1+1/(2^{r}-1)}$ communication lower bound for $r$-round $(r+1)$-party protocols. 

\begin{Remark}
As should be clear from this discussion, our hierarchical embedding technique is quite general and is not particularly tailored to the MIS problem. Hence, it seems quite plausible to extend this approach to various other graph problems as well, 
making hierarchical embeddings a general technique for proving multi-pass graph streaming lower bounds. 
\end{Remark}

\subsection{Combinatorial Part: Dense Graphs with Many Induced Subgraphs}\label{sec:overview-dup}

We now discuss the construction of our \emph{base graphs} for our hierarchical embeddings. Recall that we need an $n$-vertex graph $G$ with $q \approx n/b$ \textbf{subgraph groups} $\set{H_{i,*}}_{i \in [q]}$, each with $p \approx n/b$ \emph{vertex-disjoint} subgraphs $\set{H_{i,1},\ldots,H_{i,p}}_{i \in [q]}$ on $b$ vertices each. In addition, the \textbf{inducedness property} requires that there are no other edges between the vertices of each subgraph group (other than the edges of the subgraphs themselves). 
Refer back to~\Cref{fig:overview2} from earlier for an illustration.

Our starting point is the recent RS-graph based constructions in streaming lower bounds in~\cite{AssadiS23,KonradK24}. In particular,~\cite{AssadiS23} created a graph with ``high entropy permutations'' in place of edges of RS graphs and~\cite{KonradK24} took this even further with a graph that has many ``small'' RS graphs embedded inside one larger RS graph (roughly speaking, by changing
each edge of the outer RS graph, with a copy of the inner RS graph). These constructions, however, are still insufficient for us as they can support quite limited ``inner'' graphs, while 
our hierarchical embeddings require a graph consisting of highly complex subgraphs corresponding to hard instances of MIS with one less round. For instance, both constructions in~\cite{AssadiS23,KonradK24} inherently 
can only support \emph{bipartite} inner graphs while our hard instances definitely cannot be bipartite as finding MIS of bipartite graphs is easy with $\approx n$ communication in one round. 

Our constructions work with an \emph{outer} graph, which, instead of induced matchings in RS graphs, contains many \textbf{Unique Path Collections (UPCs)} that can {almost} be seen as ``stretching'' an induced 
matching (a collection of paths of length $1$), to longer paths (see~\Cref{def:upc} for the formal definition and~\Cref{fig:overview5,fig:upc} for illustrations). The `almost' part however is due to a subtle difference: the property of UPCs 
goes beyond only their \emph{induced} subgraphs -- instead, these paths are not only induced on their vertices, but also are the \emph{unique} shortest path between their endpoints even 
if one uses outside vertices. See~\Cref{fig:overview5} for an example. 
\begin{figure}[h!]
\centering
\includegraphics[scale=0.42]{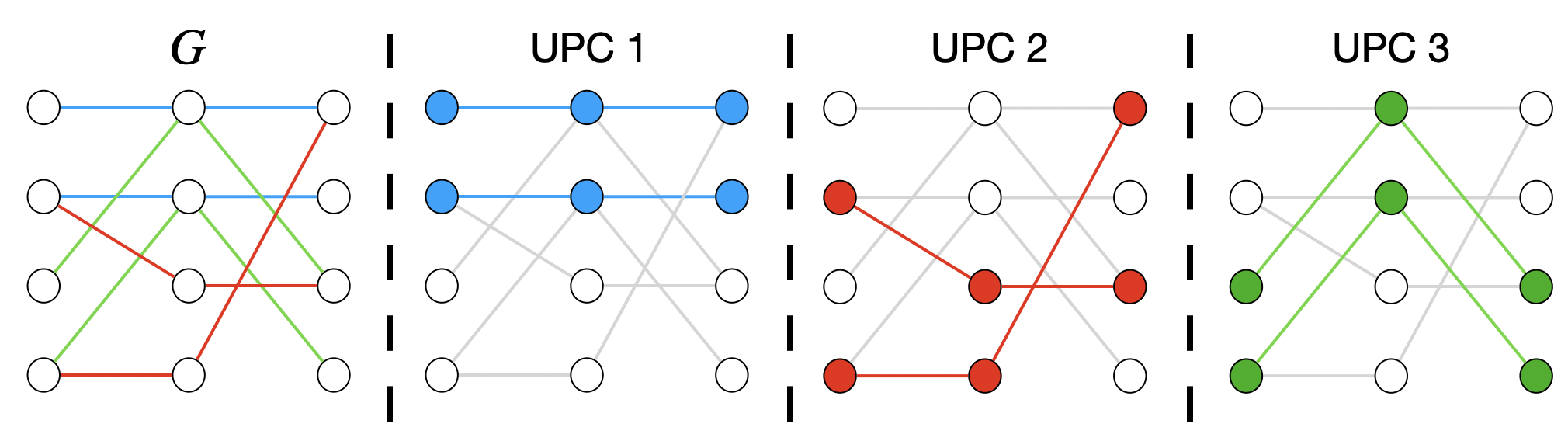}
\caption{An example of UPCs in the graph $G$ on the left. In each UPC, there is no path of length $2$ from a starting vertex to an end vertex unless the path belongs to the UPC itself.}\label{fig:overview5}
\end{figure}

Such outer graphs have been already studied extensively in the context of shortcutting sets and spanners; see, e.g., ~\cite{Hesse03,AbboudB16,CoppersmithE06,HuangP21,BodwinH22,LuWWX22,WilliamsXX23} and references therein. 
For our purpose, we need a slight variation that follows from the approach in~\cite{AbboudB16}, which is itself based on the original RS graph constructions in~\cite{RuzsaS78} (the variation allows for ``packing'' 
many paths \emph{together} in a UPC, as opposed to the \emph{path-wise} guarantee of prior work). Nevertheless, it appears that the way we use these graphs is completely different 
from prior work. We are also not aware of any prior applications of these graphs to streaming or communication lower bounds, or to the maximal independent set problem in any other computational model\footnote{In terms of parameters, we are working with almost the exact opposite of prior work: we aim to \emph{maximize} their density and size of UPCs---crucially of size $n^{2-o(1)}$ and $n^{1-o(1)}$, respectively---while \emph{minimizing} their diameter---crucially of size $n^{o(1)}$---unlike prior work that aim to achieve a diameter of $n^{\Omega(1)}$; we only \emph{have to} increase the diameter per each level of our hierarchical embedding to accommodate the recursive strategy.}. 
 
 Having obtained these outer graphs, our \textbf{embedding product} works by replacing, for every $i \in [q]$, each \emph{path} in the $i^{\textnormal{th}}$ UPC of these graphs with the subgraphs in $H_{i,*}$ that we would like to have in our base graph.
 The only requirement we need from subgraphs $\set{H_{i,j}}_{i,j}$ is that they should have a small \emph{chromatic number}; in particular, as long as they are \emph{$k$-colorable} graphs, we can also embed them in UPCs with paths of length $k-1$ by mapping each of their color classes to one of the vertices of the path in a careful way. The strong guarantee of our outer graphs can then be used to argue that these $k$-colorable graphs cannot have inserted an edge between vertices of each other (i.e., break
 the inducedness property), as such an edge can be traced backed to a ``shortcutting path'' in the outer graph which cannot exist. See~\Cref{fig:overview6} for an illustration. 
 \begin{figure}[h!]
\centering
\includegraphics[scale=0.30]{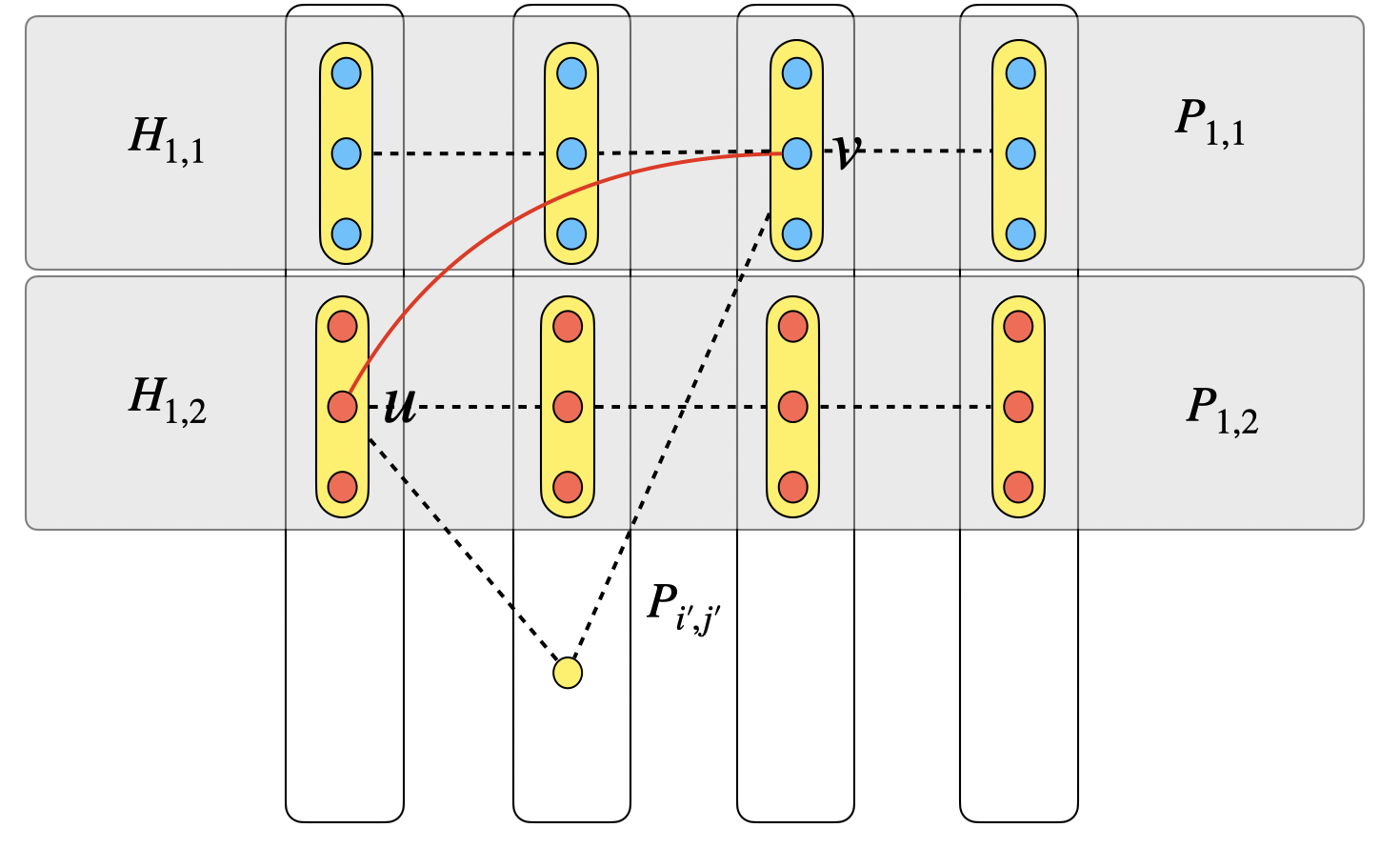}
\caption{If the solid (red) edge between $u$ and $v$ is inserted by any other $k$-colorable graph $H_{i,j}$ for $i \neq 1$, we will also find a shortcut between ``outer'' versions of $u$ and $v$, which cannot exist in the outer graph.}\label{fig:overview6}
\end{figure}

Finally, in our application to MIS, we start with $2$-colorable graphs as our base case instances that are hard for $0$-round protocols. The identical copying plus adding a clique between some parts of the two copies that we mentioned earlier, with some modifications, leads to a $4$-colorable hard instance for $1$-round protocols. Following the same pattern thus leads to $2^{r+1}$-colorable graphs for $r$-round protocols. Our construction allows for setting $q \approx p \approx n^{1-o(1)}$
for embedding any $n^{o(1)}$-colorable graphs, and since we will only need to continue the lower bound for $r \approx \log\log{n}$ rounds, we will have more than enough room to carry out the entire plan.

\subsection{Information-Theoretic Part: Round Elimination via Direct-Sum}\label{sec:overview-round-elimination}

We will now discuss the information-theoretic parts of our lower bound. The combinatorial parts of our lower bound
gives us a quite strong boost here: we can now \emph{solely} focus on the ability of the protocols in ``traversing'' the hierarchical tree of instances we discussed earlier (see~\Cref{fig:hierarchical}), 
without worrying about \emph{any} interference between underlying instances of the problem (there is none owing to the inducedness property of our construction). In other words, 
the communication lower bound here is practically independent of the MIS problem and is {essentially} (although not quite) only about the abstract problem of traversing the hierarchical tree: 
\begin{itemize}[itemsep=0pt]
	\item The first $r$ players receive  $q_r \cdot p_r$ ``hard'' $(r-1)$-round sub-instances $\mathcal{I} := \set{I_{i,j} \mid i \in [q_r], j \in [p_r]}$; the last player receives an index $t \in [q_r]$; 
	\item The goal for the players is to solve \emph{all} \textbf{special} sub-instances of $I_{t,*}$ in $r$ rounds. 
\end{itemize}
 
This problem is similar in spirit to the \emph{tree-pointer jumping} problem of~\cite{ChakrabartiCM08}.  
The crucial difference is that  tree-pointer jumping only has \emph{one} special instance per level (think of $p_r = 1$ above), whereas for us, 
we need the players to solve \emph{many} ``small'' sub-instances---which, individually, are not hard to solve within the budget of even the $(r-1)$-round protocols (this is precisely our deviation from 
the ideal-type lower bounds to our actual lower bound). A similar problem was also very recently introduced by~\cite{KonradK24} {only} for $2$-round protocols as their \emph{HiddenStrings} problem. 
However, the techniques in~\cite{KonradK24}---based on the two-way communication complexity of the Index problem~\cite{JainRS09}---are inherently only for $2$-round protocols and do not generalize beyond. 

Our second technical contribution is a lower bound for this problem (in a bit less abstract form). 

\paragraph{An intuitive ``proof''.} Suppose we have a protocol $\prot$ with $s=o(p_r \cdot q_r)$ communication on the distribution $\GG_r$ of $r$-round instances. 
We can first show that, in its first round, $\prot$ can only reveal $O(s/q_r) = o(p_r)$ bits in total about all instances in $I_{t,*}$ (as these are an {unknown} $1/q_r$ fraction of the inputs in $\mathcal{I}$). 
This means that \emph{on average} $\pi$ has only revealed 
$o(1)$ bits about any one instance $I_{t,j}$ for $j \in [p_r]$. Thus, in the remaining $r-1$ rounds, it needs to solve $p_r$ copies of the $(r-1)$-round problem, where each copy, on average, is only $o(1)$ {statistical distance} away from the hard $(r-1)$-round 
distribution $\GG_{r-1}$. This second problem has a flavor
of a direct-sum result\footnote{Recall that a direct-sum result argues that solving $m$ \emph{independent} copy of a problem becomes $m$ times harder than solving one copy; in this context, solving $m$ copies needs $m$ times more communication than one copy.}. Hence, using a direct-sum style argument (say, as in~\cite{BarakBCR10}), 
we should also be able to say that in the remaining $(r-1)$ rounds, $\prot$ needs to communicate $p_r$ times the communication cost of solving one $(r-1)$-round hard instance. 
This gives us the tradeoff stated at the very end of~\Cref{sec:actual}.

Unfortunately, making this intuition precise is hindered by the fact that in this direct-sum scenario, the $p_r$ $(r-1)$-round instances, in addition to their marginal deviation of $o(1)$ from $\GG_{r-1}$, also 
have become \emph{correlated} with each other for a total of $o(p_r)$ bits. In other words, the first message of $\prot$ can also change their \emph{joint} distribution by revealing $o(p_r)$
bits about them collectively.  
In general, one cannot hope for a direct-sum result to hold over such non-independent instances. 

\paragraph{Our actual proof.} To tackle this obstacle, we apply an \emph{information complexity direct-sum} to the previous round elimination arguments of~\cite{AlonNRW15,AssadiKZ22}. 
Specifically, we show how to obtain a $(r-1)$-round protocol $\prot'$ for a \emph{single} hard $(r-1)$-round instance using the protocol $\prot$: 
\begin{ourbox}
\textbf{High level description of the protocol $\prot'$ for $I_{r-1} \sim \GG_{r-1}$:}
\vspace{-0.15cm}
\begin{enumerate}[itemsep=0pt]
	\item \textbf{Sampling step:} \emph{Sample} the first message $\Prot_1$ of $\prot$ and a family of sub-instances $\mathcal{I}$ from $\GG_{r}$, i.e., $(\Prot_1,\mathcal{I}) \sim \GG_{r}$. Then, 
	sample $k \in [p_r]$ uniformly at random and \emph{replace} $I_{t,k}$ in $\mathcal{I}$ with the input instance $I_{r-1}$ (where $t$ is the index of the special sub-instance group of $\mathcal{I}$). 
	\item \textbf{Simulation step:} Run the protocol $\prot$ on the modified $\mathcal{I}$ from its \emph{second} round onwards and output the answer of $\prot$ on the special sub-instance $I_{t,k}$. 
\end{enumerate}
\end{ourbox}

\vspace{-0.15cm}
The analysis consists of two parts, one for each step of the protocol. The first part shows how the players can sample the joint variables $(\Prot_1,\mathcal{I})$ and embed $I_{r-1}$ inside $\mathcal{I}$---which crucially is independent of $\Prot_1$ 
in this sampling, although it should \emph{not} be in the distribution of by $\prot$---using a distribution that is \emph{close} to $\GG_{r} \mid I_{i,j} = I_{r-1}$. 
This sampling is inspired by~\cite{AlonNRW15,AssadiKZ22} with various modifications to account for the considerable differences between the two communication models. 
Proving this part is based on a similar argument as in the intuitive ``proof'' above that says that the marginal distribution of $I_{i,j}$ is not affected that much by the first-round message $\Prot_1$. 

The second part of the argument, however, deviates entirely from~\cite{AlonNRW15,AssadiKZ22}, because in their communication model, the new protocol $\prot'$ already has the desired 
communication cost for an $(r-1)$-round protocol (given their focus on maximum message-size \emph{per} vertex of the graph); thus, their arguments finish at this point. 
However, for us, the protocol $\prot'$ has the same communication cost as $\prot$ even though it is being run on a \emph{much smaller} instance. Hence, there is no contradiction that 
$\prot'$ can solve $\GG_{r-1}$ in $(r-1)$ rounds as its communication cost is way above the bar anyway. 

To handle this part, we instead work with the \emph{information cost} of the protocols (in place of their communication cost) for a natural \emph{multi}-party generalization of two-party \emph{external} information cost in~\cite{BarakBCR10} 
(see~\Cref{sec:ic} for the definition). Our approach is then to show that even though $\prot'$ may communicate as much as $\prot$, its information cost is a factor $1/p_r$ smaller than the information cost of $\prot$. 
This is because $I_{r-1}$ is embedded in one of the $p_r$-many special sub-instances randomly whose identity is \emph{hidden} to the protocol $\prot$. But, this creates another challenge: information cost, unlike communication, 
is a function of the underlying distribution, and in $\prot'$, the distribution of the inputs that $\prot$ is simulated on is not the original distribution $\GG_{r}$, only statistically close to it. 
Thus, even though we know that $\prot$ has a \emph{low} information cost on $\GG_r$, we cannot guarantee it also has a low information cost in the simulation step
and conclude that $\prot'$ has a low information cost\footnote{In general, information cost of a protocol on two statistically close distributions can be vastly different.}. 

The challenge outlined above is not unique to our problem and has been studied before also, e.g., for direct-\emph{product} results in~\cite{BravermanRWY13a,BravermanRWY13b,BravermanW15} but typically
for \emph{internal} information cost and two-party protocols. For our purpose, we use a \emph{message compression} technique due to~\cite{HarshaJMR07} (see also~\cite{BravermanG14}) to reduce the \emph{communication} cost 
of the protocol $\prot$ in the simulation step down to its information content. This leads to some minimal loss in the parameters that can be easily handled
in our setting. But now, the change in the underlying distribution of the simulation step cannot affect the communication cost, which allows us to obtain a too-good-to-be-true protocol $\prot'$ for $\GG_{r-1}$, which leads to our desired contradiction.

\medskip

This concludes our overview.  In the rest of the paper, we formalize the ideas and approaches discussed in this section and prove~\Cref{res:main} and~\Cref{res:cc} formally.

%% file: prelim.tex
\section{Preliminaries}

\paragraph{Notation.} For any $t \geq 1$, define $[t] := \set{1,\ldots,t}$. Some times, for a clarity of exposition, we may write $(a) \uparrow (b)$ to denote $a^b$. 
For any tuple $(x_1,\ldots,x_t)$ and $i \in [t]$,  
we define $x_{<i} := (x_1,\ldots,x_{i-1})$. We define $x_{>i}$ and $x_{-i}$, analogously. For a set of tuples $\set{(x,y) \mid x \in X, y \in Y}$ for some 
sets $X$ and $Y$, and $x \in X$, we define $(x,*) := \set{(x,y) \mid y \in Y}$; we define $(*,y)$ for $y \in Y$ analogously. 


When there is room for confusion, we use sans-serif letters for random
variables (e.g. $\rA$) and normal letters for their realizations (e.g. $A$). We use $\distribution{\rA}$ and $\supp{\rA}$ to denote the distribution of $\rA$ and its support, respectively. 

For random variables $\rA,\rB$, we use $\en{\rA}$ to denote the \emph{Shannon entropy}
 and $\mi{\rA}{\rB}$ to denote the \emph{mutual information}. For two distributions $\mu$ and $\nu$ on the same support, $\tvd{\mu}{\nu}$ denotes their \emph{total variation distance} 
 and $\kl{\mu}{\nu}$ is their \emph{KL-divergence}.~\Cref{app:info} contains the definitions of these notions and standard
information theory facts that we use in this paper. 

\subsection{(Multi-Party) Communication Complexity}\label{sec:cc}

We work in the standard multi-party number-in-hand communication model; we provide some basic definitions here and refer the interested reader to the excellent textbooks~\cite{KushilevitzN97,RaoY20} on communication complexity for more details. 

There are $k \geq 1$ players $P_1,\ldots,P_k$ in this model who receive a partition of the edges of an input graph $G=(V,E)$. 
The players follow a protocol $\pi$ to solve a given problem on $G$, say, finding an MIS.
They have access to a \emph{shared} tape of randomness, referred to as \emph{public randomness}, in addition to their own \emph{private randomness}.  
In each \emph{round}, $P_1$ writes a message on a board visible to all parties, followed by $P_2$, all the way
to $P_k$. Each message can only depend on the private input of the sender, content of the blackboard, public randomness, and private randomness of the sender. 
At the end of the last round, the last player $P_k$ writes the answer to the blackboard. 

\begin{Definition}\label{def:cc}
	For any protocol $\prot$, the \textbf{communication cost} of $\prot$, denoted by $\cc{\prot}$, is defined as the worst-case (maximum) length of messages, measured in bits, communicated by all players 
	on any input. We assume that all \textbf{transcripts}, i.e., the set of all messages written onto the blackboard, in $\prot$ have the same worst-case length (by padding).
\end{Definition}

\paragraph{Communication complexity and streaming.} The following standard result---dating back (at the very least) to the seminal paper of~\cite{AlonMS96} that introduced the streaming model---relates communication cost of protocols and space of streaming algorithms. 

\begin{proposition}[cf.~\cite{AlonMS96}]\label{prop:cc-stream}
	For any $p \geq 1, s \geq 1$, and $\delta \in (0,1)$, suppose there is a $p$-pass $s$-space $\delta$-error streaming algorithm $A$ for some problem $\PP$. Then, for any integer $k \geq 1$, there also exists a $k$-party protocol $\prot$ with $p$ rounds, communication cost $\cc{\prot} \leq p \cdot k \cdot s$, and error probability at most $\delta$, for the same problem $\PP$. 
\end{proposition}
\begin{proof}
	Consider the stream $\sigma = \sigma_1 \circ \cdots \circ \sigma_k$ where $\sigma_i$ is the input to the player $P_i$ in $\prot$ (ordered arbitrarily in the stream). 
	Player $P_1$ runs $A$ on $\sigma_1$ and writes the memory content on the board, which allows player $P_2$ to continue running $A$ on $\sigma_2$, and so on and so forth. 
	This allows the players to run one pass of $A$ using communication cost at most $k \cdot s$. 
	The players can continue this in $p$ rounds, faithfully simulating the $p$ passes of the algorithm, and at the end $P_k$ can read the answer of $A$ from the memory and outputs it on the blackboard. 

	This way, $\prot$ will have the same probability of success as $A$, uses at most $p \cdot k \cdot s$ communication and $p$ rounds of communication. 
\end{proof}
\Cref{prop:cc-stream} allows us to translate communication lower bounds into streaming lower bounds. In particular, this result reduces proving~\Cref{res:main} to proving~\Cref{res:cc}. 

\subsection{Information Cost and Message Compression}\label{sec:ic}
We also work with the notion of \emph{information cost} of protocols that originated in~\cite{ChakrabartiSWY01} and has since
found numerous applications (see, e.g.~\cite{Weinstein15} for an excellent survey of this topic). 

There are various definitions of information cost that have been considered depending on the application. The definition we use---which is a natural multi-party variant of the standard \emph{external information cost}~\cite{BarakBCR10}---is best suited for our application.  

\begin{Definition}\label{def:ext-info}
	For any multi-party protocol $\prot$ whose inputs are distributed according to the distribution $\mu$, the \textbf{(external) information cost} is defined as: 
	\[
		\ic{\prot}{\mu} := \mi{\rG}{\rProt \mid \rR},
	\]
	where $\rG$ denotes the random variable for the input graph sampled from $\mu$, $\rProt$ denotes the set of all messages written on the board, and $\rR$ is the public randomness. 
\end{Definition}
In words, multi-party external information cost measures the information revealed by the entire protocol about the entire graph to an \emph{external} observer. 
The following result follows from~\cite{BarakBCR10}. 
\begin{proposition}[cf.~\cite{BarakBCR10}]\label{prop:ic-cc}
	For any multi-party protocol $\prot$ on any input distribution $\mu$, 
	\[
		\ic{\prot}{\mu} \leq \cc{\prot}. 
	\]
\end{proposition}
\begin{proof}
	We have, 
	\begin{align*}
		\ic{\prot}{\mu} \Eq{(1)} \mi{\rG}{\rProt \mid \rR} \Eq{(2)} \en{\rProt \mid \rR} - \en{\rProt \mid \rG,\rR} \Leq{(3)} \en{\rProt \mid \rR} \Leq{(4)} \en{\rProt} \Leq{(5)} \log{\card{\supp{\rProt}}} \Eq{(6)} \cc{\prot}; 
	\end{align*}
	here, $(1)$ is by the definition information cost, $(2)$ is by the definition of mutual information, (3) is by the non-negativity of entropy (\itfacts{uniform}), (4) is because conditioning can only reduce the entropy (\itfacts{cond-reduce}), 
	(5) is because uniform distribution has the highest entropy (\itfacts{uniform}), and (6) is by the definition of worst-case communication cost. 
\end{proof}

\paragraph{Message compression.} We use \emph{message compression} to reduce communication cost of protocols close to their information cost in certain settings.  
The following result is due to~\cite{HarshaJMR07} which was further strengthened slightly in~\cite{BravermanG14}. We follow the textbook presentation in~\cite{RaoY20}. 

\begin{proposition}[c.f.{\cite[Theorem 7.6]{RaoY20}}]\label{prop:msg-compress}
	Suppose Alice knows two distributions $\cA, \cB$ over the same set ${U}$ and Bob only knows $\cB$. Then, there is a protocol for Alice and Bob to sample an element according to $\cA$ 
	by Alice sending a single message of size 
	\begin{align}
			\kl{\cA}{\cB} + 2 \log (1+\kl{\cA}{\cB}) + O(1)  \label{eq:kl-msg-compress}
	\end{align}
	
	bits \underline{in expectation}. This protocol has no error. 
\end{proposition}
\noindent
We note that somewhat weaker bounds on the KL-Divergence in \Cref{eq:kl-msg-compress} already suffice for our purposes. Thus, to simplify the exposition, we use,
\begin{align}
		\kl{\cA}{\cB} + 2 \log (1+\kl{\cA}{\cB}) + O(1)  &\leq 	\kl{\cA}{\cB} + 10 	\kl{\cA}{\cB}  + O(1) \tag{as $\log(1+x) \leq 5 x$ for every $x \geq 0$ } \\
		&\leq \raoconst \cdot \paren{\kl{\cA}{\cB}  + 1},\label{eq:msg-compress-final}
\end{align}
for some \textbf{absolute constant} $\raoconst  \geq 1$ that we use from now on in our proofs.

%% file: dup-embed.tex
\newcommand{\Gdup}{\ensuremath{G_{\textsc{dup}}}}

\section{Disjoint-Unique-Paths Graphs and Embedding Products} \label{sec:dup-embed} 

Our lower bound for MIS relies on two graph-theoretic ingredients: $(i)$ an extremal family of graphs which can be seen as a generalization of \emph{RS graphs}~\cite{RuzsaS78}, 
and $(ii)$ a graph product which we call the \emph{embedding product} for this particular
family of graphs. We now define these constructions and establish their key properties. To continue, we need the following basic definitions (see \Cref{fig:layered}). 

\begin{Definition}\label{def:layered-graph}
	For any $n, k \geq 1$, a graph $G = (V, E)$ is called a $(n,k)$-\textbf{layered graph} if there exists an equipartition of $V$ into $k$ \emph{independent sets} $(V_1, V_2, \ldots, V_k)$ of size $n/k$ each, referred to
	as \textbf{layers} of $G$. We additionally say that $G$ is a \textbf{strictly-layered graph} if all edges of $G$ are between consecutive layers. 
\end{Definition}
\vspace{-0.2cm}
\begin{Definition}\label{def:layered-path}
	For any $n, k \geq 1$, a path $P$ in an $(n,k)$-strictly-layered graph $G$ is called a \textbf{layered path} if it has edges between consecutive layers and exactly one vertex per layer.
	For any layered path $P$, we define $\startvert(P)$ as the vertex of $P$ in the layer $V_1$ of $G$ and $\finalvert(P)$ as the vertex of $P$ in the layer $V_k$. 
\end{Definition}

\begin{figure}[h!]
	\centering
	\subcaptionbox{ A $(12,4)$-layered graph.}%
	[.4\linewidth]{
		\includegraphics[scale=0.36]{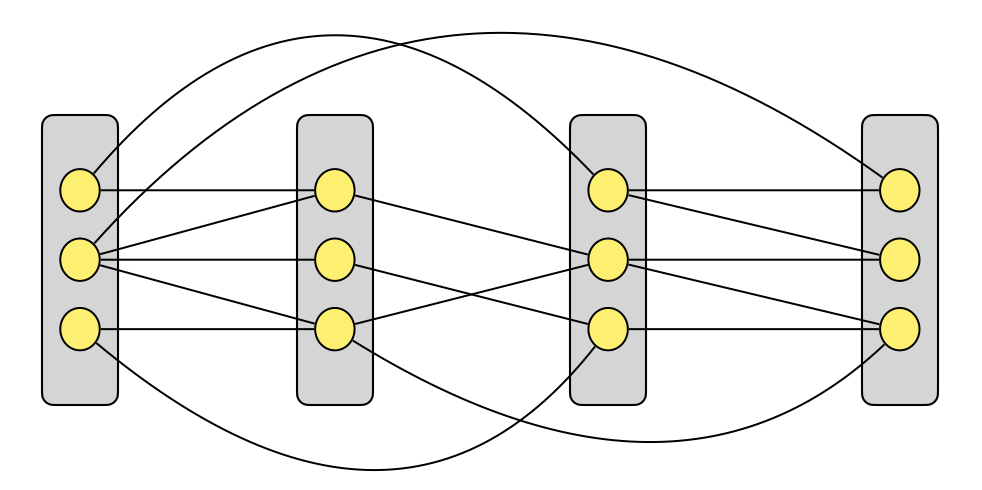}
	} 
	\hspace{0.2cm} 
	\subcaptionbox{ A $(12,4)$-strictly-layered graph and a layered path inside it.}%
	[.55\linewidth]{
	\includegraphics[scale=0.38]{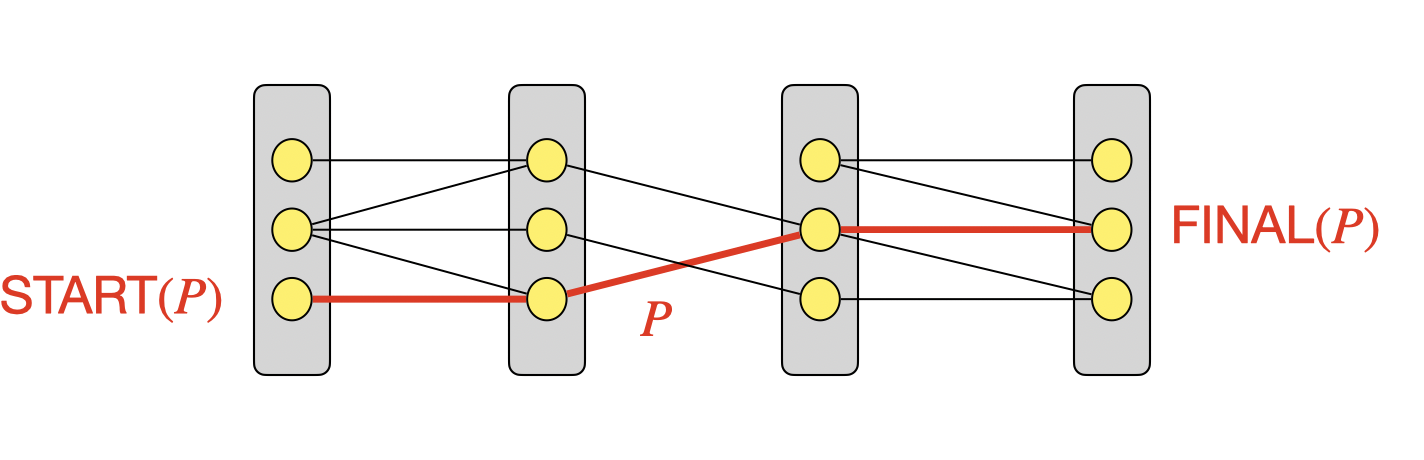}
	}
	\caption{An illustration of layered graphs and strictly-layered graphs with a layered path.}
	\label{fig:layered}
\end{figure}

\subsection{Disjoint-Unique-Paths (DUP) Graphs}\label{sec:dup}

The following definition captures the key concept we need from our extremal graphs. 

\begin{Definition}\label{def:upc}
A set of paths $\cP$ in a \emph{strictly}-layered graph $G$ is called a \textbf{unique path collection (UPC)} iff $(i)$ $\cP$ is a vertex-disjoint set of layered paths, 
and, $(ii)$ for any pair of vertices $s$ in $\startvert(\cP) := \set{\startvert(P) \mid P \in \cP}$ and $t$ in $\finalvert(\cP) := \set{\finalvert(P) \mid P \in \cP}$, 
the only layered path between $s$ and $t$ in the entire graph $G$ is a path in $\cP$ (if one exists, $s$ and $t$ are the two ends of the same path $P \in \cP$). 
\end{Definition}
\begin{figure}[h!]
\centering
\includegraphics[scale=0.35]{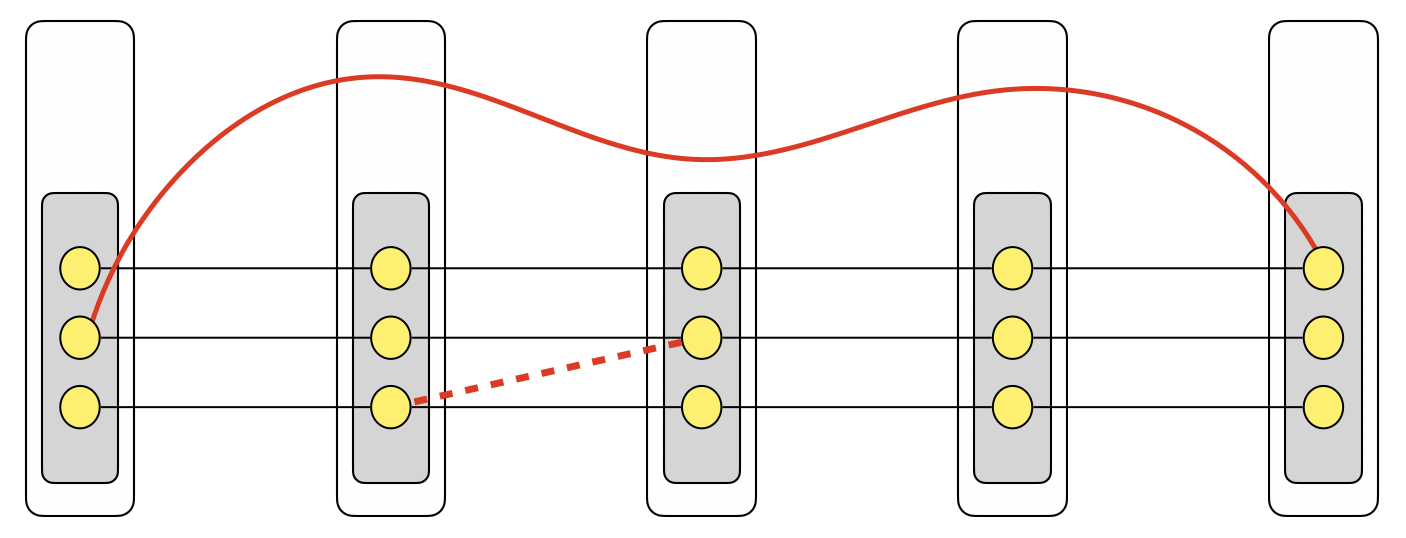}
\caption{An illustration of a UPC in a strictly-layered graph. Neither the dashed red edges (between the vertices of different paths in a UPC) nor the edges of the solid red path (using outside vertices) are allowed.
}\label{fig:upc}
\end{figure}

See~\Cref{fig:upc} for an example of UPCs. We use UPCs to define our extremal graph family. 

\begin{Definition}\label{def:unique-paths}
	For any $n,k, p,q \geq 1$, a $(p,q,k)$-\textbf{disjoint-unique-paths (DUP)} graph is a $(n,k+1)$-strictly-layered graph $G$ whose edges can be partitioned into $q$ 
	UPCs $\cP_1,\ldots,\cP_q$, each consisting of exactly $p$ layered paths $P_{i,1},\ldots,P_{i,p} \in \cP_i$ for $i \in [q]$. 
	
\end{Definition}
The following result shows the existence of $(p,q,k)$-DUP graphs with $p,q = n^{1-o(1)}$ (that is necessary for our application) and $k = \omega( \poly\!\log{(n)})$ (that is sufficient for our application). 
The existence of such results is a-priori surprising as UPCs of size $p$ are ``locally sparse'' and restrict many edges from appearing in the graph, and yet large $q$ implies that the graph is ``globally dense''. 
Nevertheless, we can use the same ideas as they are used in the construction of RS graphs~\cite{RuzsaS78} or for `graphs with many long disjoint shortest paths' in~\cite{AbboudB16} to construct DUP graphs. 
\begin{proposition}\label{prop:dup}
	For any sufficiently large $n \in \IN$ and any given $k \geq 1$ satisfying $k \leq 2^{(\log{n})^{1/4}}$, there exists a $(p,q,k)$-DUP graph on $n$ vertices with the following parameters for two absolute constants $\eta_p,\eta_q > 0$,
	\[
		p = \frac{n}{\exp\paren{\eta_p \cdot (\ln{n})^{3/4}}} \qquad \text{and} \qquad q = \frac{n}{\exp\paren{\eta_q \cdot (\ln{n})^{3/4}}}.
	\]
\end{proposition}
Given that this result does not seem to have appeared in prior work, we provide a self-contained proof in~\Cref{sec:dup-construct}, which may also provide some more intuition about these graphs.

\subsection{Embedding Products}\label{sec:embed}

The next key ingredient of our construction is a graph product that treats DUP graphs as an ``outer'' graph and replace their \emph{entire} paths in UPCs with 
different ``inner'' layered graphs (not necessarily strictly-layered ones). See~\Cref{fig:embed} for an illustration.

\begin{Definition}\label{def:embed}
	Let $\Gdup$ be a $(p,q,k)$-DUP graph on $n_1$ vertices and layers $U_1,\ldots,U_{k+1}$ for some $p,q,k \geq 1$. 
	Let $\cH := \set{H_{i,j} \mid i \in [q] , j \in [p]}$ be a family of $(n_2,k+1)$-layered graphs where all the graphs in $\cH$ are on the same layers $W_1,\ldots,W_{k+1}$ but may have different edges. 

We define the \textbf{embedding of $\cH$ into} $\Gdup$, denoted by $G:=\embed(\cH \rightarrow \Gdup)$, as the following $((n_1 \cdot n_2)/(k+1) , k+1)$-layered graph: 
	\begin{itemize}[leftmargin=10pt]
		\item \textbf{Vertices:} We have layers $V_1,\ldots,V_{k+1}$ where, for every $\ell \in [k+1]$, $V_\ell := U_\ell \times W_\ell$; 
		\item \textbf{Edges:} For any layered path $P_{i,j} = (u_1, u_2, ..., u_{k+1})$ of UPC $\cP_i$ in $\Gdup$ with $i \in [q]$ and $j \in [p]$, and any edge $(x,y) \in H_{i,j}$ between layers $W_{\ell_x}$ and $W_{\ell_y}$, 
		we add an edge $e$ between $(u_{\ell_x},x)$ and $(u_{\ell_y},y)$ to $G$. We say that the edge $e$ is added \emph{w.r.t.}\ the path $P_{i,j}$. 
	\end{itemize}
\end{Definition}

\begin{figure}[h!]
	\centering
	\subcaptionbox{A family $\cH$ including $(12,4)$-layered graphs $H_{1,1}$ and $H_{1,2}$ (remaining graphs are unspecified) on the left and a DUP graph $\Gdup$ with $p=2$, unspecified $q$, and $k=3$ on the right. 
	}%
	[1\linewidth]{
		\includegraphics[scale=0.38]{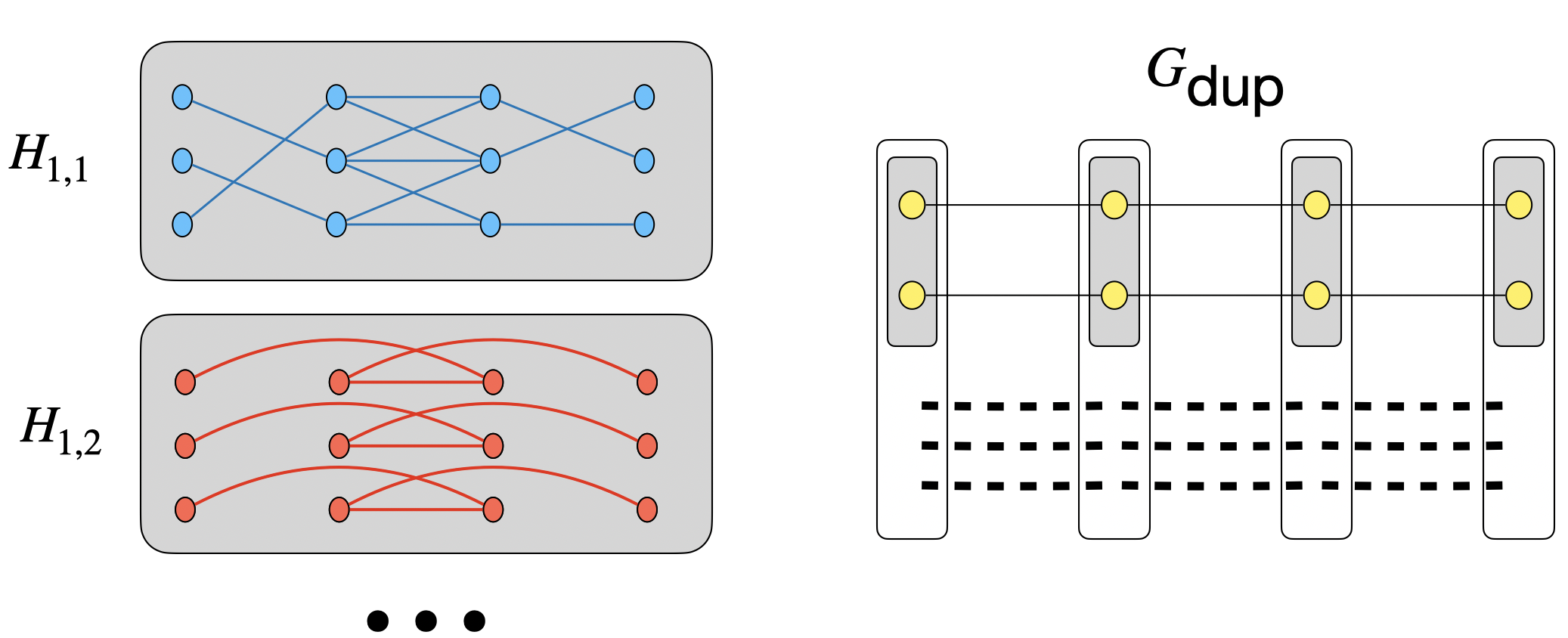}
	} 
	
	\vspace{0.5cm} 
	
	\subcaptionbox{The resulting layered graph as the embedding of $\cH$ into $\Gdup$ -- the drawing only shows the effect of the embedding on one UPC of $\Gdup$ and the remainder of the graph is unspecified.}%
	[1\linewidth]{
	\includegraphics[scale=0.38]{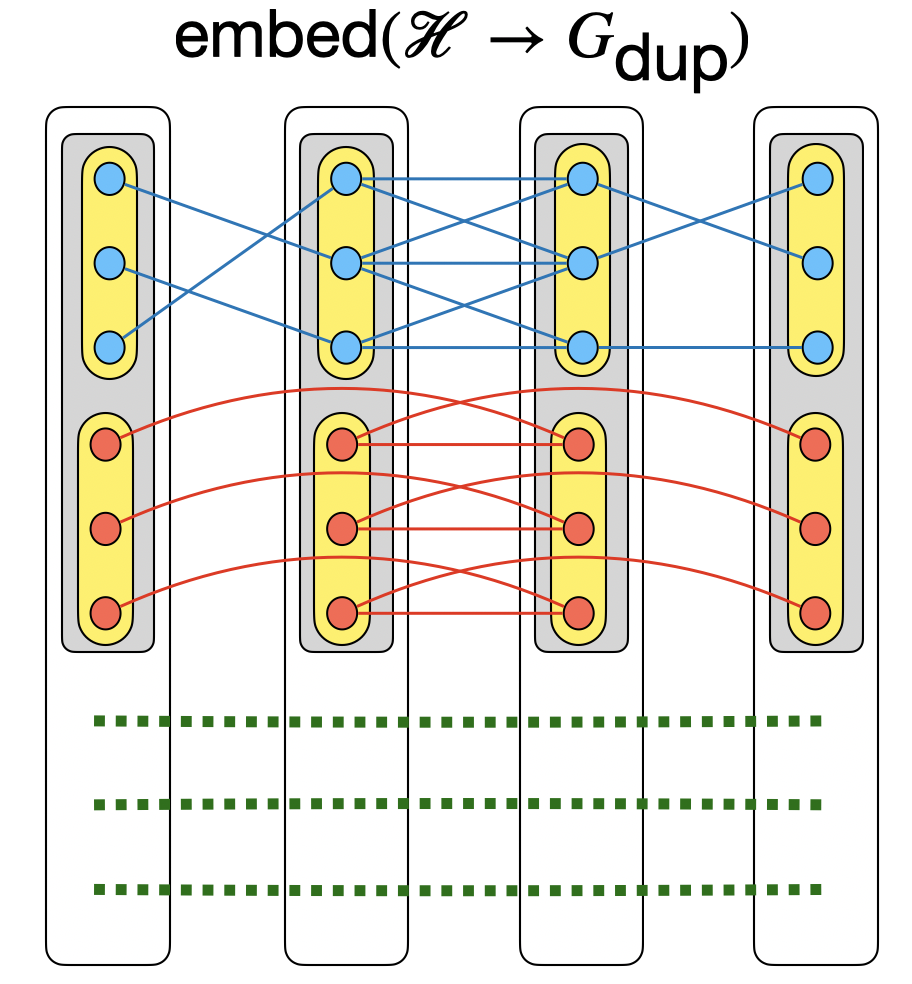}
	}
	\caption{An illustration of the embedding product. For simplicity of exposition, this figure only shows a part of the final graph obtained via embedding inside a single UPC of the DUP graph; in the actual construction, every UPC is replaced by an embedded graph.}
	\label{fig:embed}
\end{figure}

The following lemma captures the main property of these embeddings. 
\begin{lemma}\label{lem:embed-induced}
	Let $\Gdup$ be a $(p,q,k)$-DUP graph for some $p,q,k \geq 1$, $\cH := \set{H_{i,j} \mid i \in [q] , j \in [p]}$ be a family of layered graphs with size $p \cdot q$, and $G := \embed(\cH \rightarrow \Gdup)$. 
	For any $i \in [q]$, the \textbf{\emph{induced subgraph}} of $G$ on vertices corresponding to $\cP_i$, i.e., $\set{(v,*) \mid v \in P_{i,j}~\text{for some $P_{i,j} \in \cP_i$}}$, is a \emph{\textbf{vertex-disjoint union}} of graphs $H_{i,j}$ for $j \in [p]$ in $\cH$. 
\end{lemma}

Before getting to the proof of~\Cref{lem:embed-induced}, an important remark is in order. Recall that our layered graphs can have edges between any pairs of layers and not only consecutive ones. 
As such, to be able to establish the property in~\Cref{lem:embed-induced}, we crucially rely on the fact that DUP graphs not only 
disallow any additional edges between vertices of a UPC, but in fact, no layered paths can connect them also even if the paths are using vertices outside the UPC. This is the key property of DUP graphs that is needed for our constructions. 

\begin{proof}[Proof of~\Cref{lem:embed-induced}]
	For any $i \in [q]$, let $U_i := \set{(v,*) \mid v \in P_{i,j}~\text{for some $P_{i,j} \in \cP_i$}}$ denote the vertices of $G$ corresponding to the UPC $\cP_i$. 
	Let $G_i := G[U_i]$ be the induced subgraph of $G$ on $U_i$. We first prove that $(i)$ $G_i$ contains the vertex-disjoint union of $H_{i,j}$ for $j \in [p]$, and then, that $(ii)$ it does 
	not contain any other edges. See~\Cref{fig:embed-lemma} for an illustration of this proof. 
	
	\paragraph{Part (i).} 
	Fix $j \in [p]$ and observe that all edges of $H_{i,j}$ are added to $G_i$ by the embedding product w.r.t.\ layered path $P_{i,j}$ in a way that ``preserves its structure''. Formally, 
since $P_{i,j} = (u_1, \ldots, u_{k+1})$ is a \emph{layered} path (with edges between consecutive layers and one vertex per layer), the edges between layers $W_\ell$ and $W_{\ell'}$ in $H_{i,j}$ are added as edges between vertices $\{ u_\ell \} \times W_\ell \subseteq V_\ell$ and $\{ u_{\ell'} \} \times W_{\ell'} \subseteq V_{\ell'}$ of $G_i$.  
	Also, since the layered paths of the UPC $\cP_i$ must be vertex-disjoint (\Cref{def:upc}), the edges added w.r.t.\ paths in $\cP_i$ form a vertex-disjoint union of $H_{i,j}$'s for $j \in [p]$ in $G_i$.

	\paragraph{Part (ii).}
	It remains to argue that no other edges are added to $G_i$ by the embedding product. For a contradiction, suppose that $G_i$ contains an edge $e$ that is \emph{not} added w.r.t.\ $\cP_i$. As such, it must be between vertices that correspond to $u \in P_{{i,j_1}}$ and $v \in P_{i,j_2}$ for some $j_1,j_2 \in [p]$ (possibly $j_1 = j_2$), respectively. By the embedding product, the edge $e$ is added w.r.t.\ a layered path $P_{i',j'}$ and layered graph $H_{i',j'}$ for some $i' \neq i \in [q]$ and $j' \in [p]$.
	The layered graph $H_{i',j'}$ may contain any edge between layers even non-consecutive ones. 
	Thus, the existence of edge $e$ implies that there is a \emph{subpath} of $P_{i',j'}$ that connects the vertices $u \in P_{{i,j_1}}$ and $v \in P_{i,j_2}$ in $\Gdup$, which is a strictly-layered graph and only has edges between consecutive layers. This creates a layered path from $\startvert(P_{i,j_1})$ to $\finalvert(P_{i, j_2})$ (or $\startvert(P_{i,j_2})$ to $\finalvert(P_{i, j_1})$) that uses 
	one or more edges of $P_{i',j'}$ and is thus not in $\cP_i$, which is a contradiction by the definition of a UPC (\Cref{def:upc}).
\end{proof}

\begin{figure}[h!]
\centering
\includegraphics[scale=0.30]{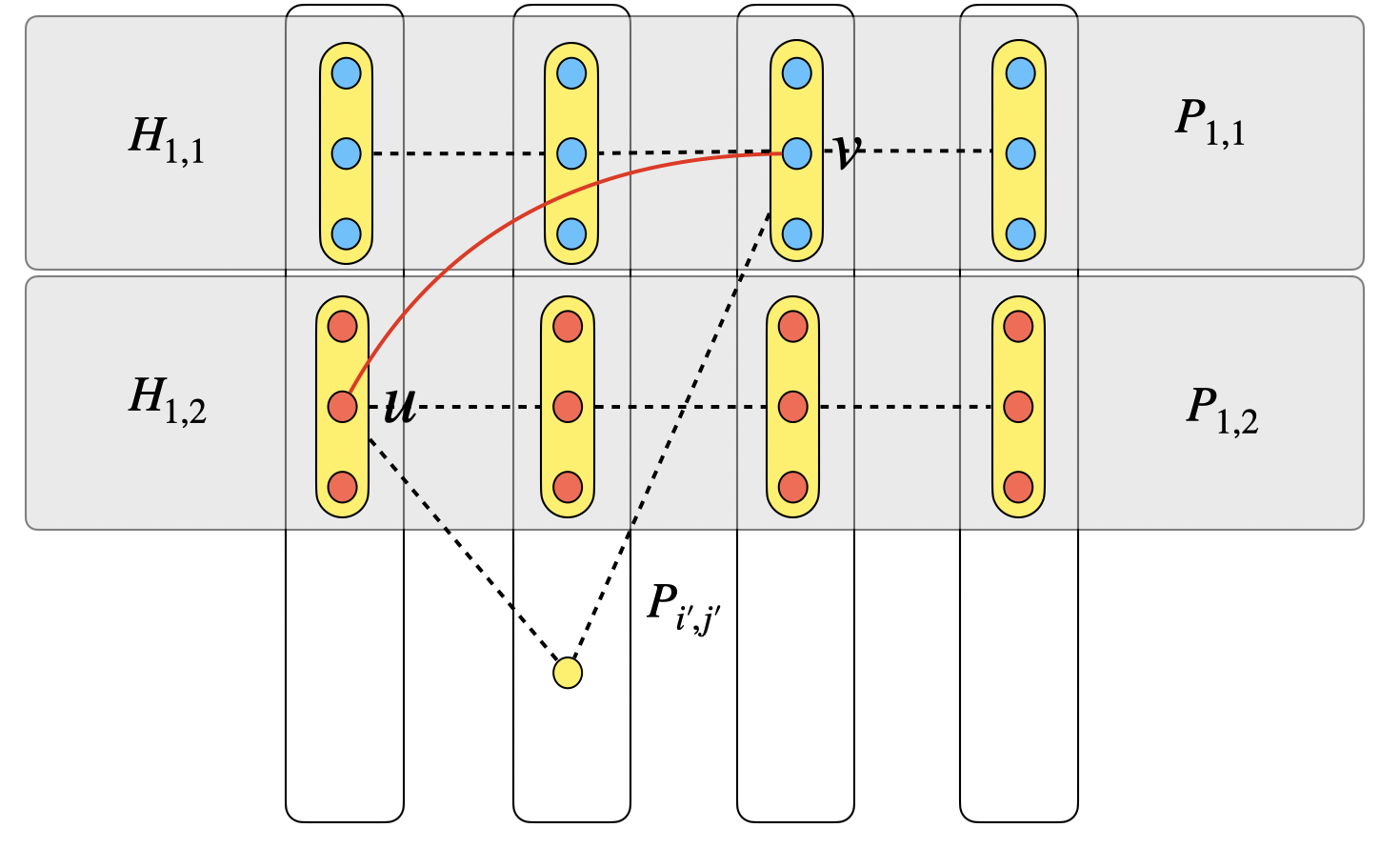}
\caption{The gray boxes denote the subgraphs $H_{1,1}$ and $H_{1,2}$ that are created w.r.t.\ paths $P_{1,1}$ and $P_{1,2}$ in the same UPC of $\Gdup$. \emph{If} there is an edge between vertices $u$ and $v$ that is added
w.r.t.\ another path $P_{i',j'}$, it implies that some subpath of $P_{i',j'}$ connects two vertices in paths $P_{1,1}$ and $P_{1,2}$ together -- this is because any edge of $G$ is added ``alongside'' a path of $\Gdup$, although it is possible that 
it \emph{shortcuts} a couple of layers in between as $H_{i',j'}$ may have edges between non-consecutive layers. This in turn creates another layered path between $\startvert(P_{1,2})$ and $\finalvert(P_{1,1})$, which 
is not allowed in a UPC, a contradiction.}\label{fig:embed-lemma}
\end{figure}

\subsection{A Construction of DUP Graphs}\label{sec:dup-construct}
We now present a construction of DUP graphs and prove~\Cref{prop:dup}. The construction of these graphs follows a standard approach in the literature, e.g., in~\cite{AbboudB16}, for creating graphs with many long disjoint shortest paths. For completeness, we provide the construction to incorporate UPCs explicitly (beyond just bounding the number of ``critical paths'' as in~\cite{AbboudB16}) and fine tune the parameters to the range that we need (which is different from the typical range of parameters used in other applications of these graphs that we are aware of). 

We first need a basic claim on the existence of a large set of ``average-free'' \emph{vectors}. 

\begin{claim}\label{clm:l2}
	For any $\ell, d \geq 1$, there exists a set $A \subseteq [\ell]^d$ with $\card{A} \geq \frac{\ell^d}{d \cdot \ell^2}$ such that for every multi-set of $t \geq 1$ not-all-equal vectors $y_1,\ldots,y_t \in A$, their average $(1/t) \sum_{i=1}^t y_i$ 
	is not in $A$. 
\end{claim}
\begin{proof}
	Let $A$ be the largest subset of $[\ell]^d$ such that all vectors in $A$ have the same $\ell_2$-norm. Since the number of different values possible for squared $\ell_2$-norm of vectors in $[\ell]^d$ 
	is at most $d \cdot \ell^2$, we obtain the lower bound in the claim statement on the size of $A$ by the pigeonhole principle. 
	
	Let $s$ denote the $\ell_2$-norm of the vectors in $A$. Consider any multi-set of $t > 1$ not-all-equal vectors $y_1,\ldots,y_t \in A$. We have, 
	\begin{align*}
	\norm{\frac{1}{t}\sum_{i=1}^{t}y_i}_2^2 &= \frac{1}{t^2} \paren{\sum_{i=1}^{t} \norm{y_i}^2 + \sum_{i \neq j} \inner{y_i}{y_j}} < \frac{1}{t^2} \paren{\sum_{i=1}^{t} \norm{y_i}^2 + \sum_{i \neq j} \norm{y_i}\norm{y_j}} \tag{the Cauchy-Schwartz inequality 
	is strict for the pair $y_i \neq y_j$ which exists as $y_i$'s are not all equal} \\
	 &= \frac{1}{t^2} \paren{t \cdot s^2 + t \cdot (t-1) \cdot s^2} = s^2. 
	\end{align*}
	The $\ell_2$-norm of the average is strictly less than $s$, hence, it cannot be in $A$. 
\end{proof}

We now use the existence of the set $A$ in~\Cref{clm:l2} to construct our DUP graphs. 

\begin{ourbox}
	\textbf{A construction of $(p,q,k)$-DUP graphs:}
	\begin{enumerate}[label=$(\roman*)$]
		\item Set $d := \sqrt{\log{\big(\frac{n}{k+1}\big)}}$ and $\ell: =\paren{\frac{n}{(k+1) \cdot (k+2)^d}}^{1/d}$ and let $A \subseteq [\ell]^d$ be the set of vectors in~\Cref{clm:l2}. 
		\item Fix the vertex set $V = V_1 \cup V_2 \cup \ldots \cup V_{k+1}$ with $V_i := [(k+2) \cdot \ell]^d$ for each $i \in [k+1]$. 
		\item For every $x \in [\ell]^d$ and $y \in A$, and all $i \in [k]$, add an edge between 
		\[
		x+i \cdot y \in V_{i} \quad \text{and} \quad x+(i+1) \cdot y \in V_{i+1}.
		\]
	\end{enumerate}
\end{ourbox}

\begin{proof}[Proof of \Cref{prop:dup}]
	The partition of the edge set into $q = \ell^d$ UPCs of size $p = \card{A} \geq \ell^d/(d \cdot \ell^2)$ each is as follows. 
	For every $x \in [\ell]^d$, we have a UPC $\cP_x := \set{P_{x,y} \mid y \in A}$ of $p$ paths. Each path $P_{x,y} \in \cP_x$ for $y \in A$ is defined as: 
	\[
		P_{x, y} := (x+y~,~ x+2y~,~ x+3y~,~ \ldots~,~ x + (k+1) \cdot y),
	\]		
where edge $(x+i \cdot y, x+(i+1) \cdot y)$ is from the layer $V_{i}$ to layer $V_{i+1}$ for $i \in [k]$. See~\Cref{fig:dup} for an illustration. We now prove that $G$ is indeed a DUP graph with the UPCs $\cP_x$ for $x \in [\ell]^d$. 

\begin{figure}[h!]
	\centering
\includegraphics[width=0.76\textwidth]{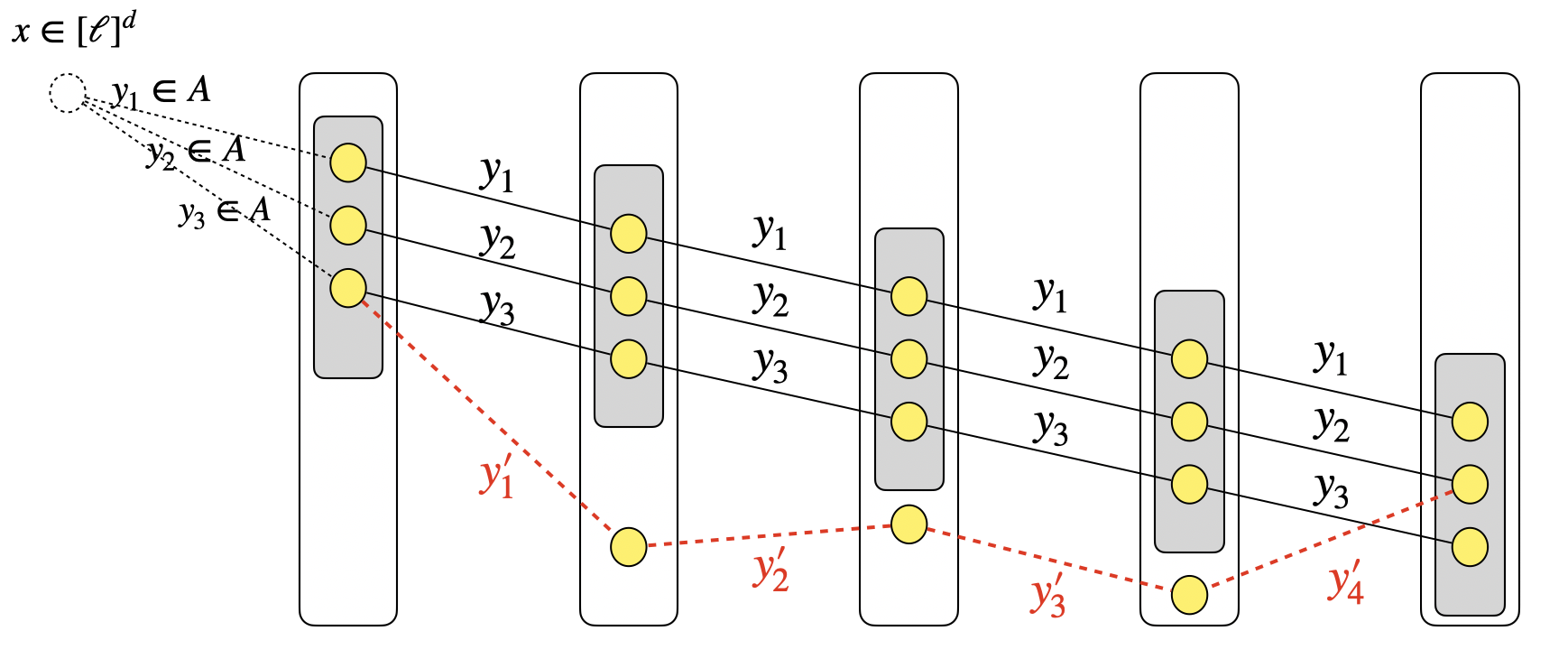}
\caption{An illustration of the construction of DUP graphs. The circle $x$ at left does not correspond to a vertex and instead is used to define a UPC. The red path cannot exist because it requires $y_3 + y'_1+y'_2+y'_3+y'_4 = 5 \cdot y_2$ which
cannot happen because $A$ is average free.}\label{fig:dup}
\end{figure}

Firstly, any edge of the graph $G$ is between some $x+i \cdot y$ and $x + (i+1) \cdot y$ for $x \in [\ell]^d$ and $y \in A$ and $i \in [k]$. This edge thus belongs to the path $P_{x,y}$ in the UPC $\cP_x$. 
	Hence, the edges of the graph are partitioned between the UPCs.
	
Secondly, by the definition, each $P_{x,y}$ is a layered path. Furthermore, the paths in $\cP_x$ are vertex-disjoint since, for any pair of paths $P_{x,y}$ and $P_{x,y'}$ for some $y,y' \in A$, their respective vertices in layer $i$, given by $x + i \cdot y$ and $x + i \cdot y'$, are only the same if $y = y'$.

Finally, consider $z_s := \startvert(P_{x,y})$ and $z_f :=\finalvert(P_{x,y'})$ for some $y,y' \in A$; 
we prove that $z_s$ has no layered path to $z_f$ in $G$ unless $y=y'$, in which case there is only one layered path $P_{x,y}$ (in the entire graph $G$), between $z_s$ and $z_f$.

By definition, we have 
\[
z_s = x+y \qquad \text{and} \qquad z_f = x + (k+1) \cdot y'.
\] 
Suppose there exists a layered path between $z_s$ and $z_f$. As each edge of the graph connects its both endpoints via a vector in $A$, this means that there exist vectors $y_1,\ldots,y_k \in A$ 
such that $z_f - z_s = \sum_{i=1}^{k} y_i$. This implies that 
\[
	x+(k+1) \cdot y' = z_f = z_s + \sum_{i=1}^{k} y_i = x+y+\sum_{i=1}^{k} y_i,
\]
which in turn means
\[
	y' = \frac{1}{k+1} \cdot \paren{y+\sum_{i=1}^{k} y_i}. 
\]
By~\Cref{clm:l2}, since $y' \in A$, the only possibility is that $(y,y_1,\ldots,y_k)$ are all-equal, which means that they are also all equal to $y'$. Thus, the only possible layered path between $z_s$ and $z_t$ is if 
they belong to the same path $P_{x,y} \in \cP_x$. 

To finalize the proof, we can work out the parameters as follows. The number of vertices is
\[
	(k+1) \cdot \paren{(k+2) \cdot \ell}^d = (k+1) \cdot \paren{(k+2) \cdot  \paren{\frac{n}{(k+1) \cdot (k+2)^d}}^{1/d}}^d = n, 
\]
as it should be, hence the choice of $\ell$ and $d$ are consistent with $n$. 
Moreover, by the upper bound on $k$ and the choice of $d$, we have, 
\begin{align*}
	q &= \ell^d = \paren{\paren{\frac{n}{(k+1) \cdot (k+2)^d}}^{1/d}}^d \geq \frac{n}{(k+2)^{d+1}} \geq \frac{n}{2^{\Theta(1) \cdot (\log{n})^{1/4} \cdot (\log{n})^{1/2}}} \geq \frac{n}{2^{\Theta((\log{n})^{3/4})}}; \\
	p &\geq \frac{\ell^d}{\ell^2 \cdot d} = \frac{q}{\ell^2 \cdot d} \geq \frac{q}{n^{2/d} \cdot \sqrt{\log{n}}}  =  \frac{q}{2^{2\log n /\sqrt{\log{n/k}}} \cdot \sqrt{\log{n}}} \geq \frac{q}{2^{\Theta(\sqrt{\log{n}})}}  \geq \frac{n}{2^{\Theta((\log{n})^{3/4})}},
\end{align*}
as desired. This concludes the proof for all graph sizes determined this way from $k$ and $\ell$ by picking $\eta_p$ and $\eta_q$ to match
the hidden constants in the $\Theta$-notation above. 

To conclude the proof, we need the construction to work for all large enough integers $n$ as in the statement of the proposition. 
However, this can be fixed easily using a padding argument, which we postpone to~\Cref{app:prop-dup}. 
\end{proof}

%% file: distribution.tex
\newcommand{\player}[1]{\ensuremath{{Q}_{#1}}}

\newcommand{\MIS}[2]{\ensuremath{\textnormal{\textbf{MIS}}_{#1,#2}}}

\newcommand{\Edup}{\ensuremath{E_{\textsc{dup}}}}

\newcommand{\tG}{\widetilde{G}}

\newcommand{\tE}{\widetilde{E}}

\newcommand{\sG}{G^{\star}}

\renewcommand{\cHstar}{\cH^{\star}}

\renewcommand{\probconst}{\eta_0}

\section{A Round vs Communication Tradeoff for MIS}\label{sec:cc-lower}

We now switch to proving a multi-party communication lower bound for MIS, defined formally as: 

\begin{Definition}\label{def:mis}
	For any integers $n,t \geq 1$, we define $\MIS{n}{t}$ as the communication game of outputting \emph{any} MIS of a given $n$-vertex graph whose edges are partitioned between $t$ players.  
\end{Definition}
\noindent
We prove an almost optimal round vs communication tradeoff for MIS in the model of~\Cref{sec:cc}, formalizing~\Cref{res:cc} from \Cref{sec:intro}. 
\begin{theorem}\label{thm:mis-cc}
	For any $r \geq 1$ and sufficiently large $n \in \IN$, any $r$-round $(r+1)$-party protocol $\prot$ for $\MIS{n}{r+1}$ with any constant probability of success strictly more than zero has communication cost 
	\[
		\cc{\prot} = \Omega\Paren{ \frac1{2^{\Theta(r \log^{5/6}(n))}} \cdot n^{1+1/(2^{r}-1)}}.
	\]
\end{theorem}


The lower bound for semi-streaming algorithms in~\Cref{res:main}, restated below, now follows. 

\begin{corollary}\label{cor:mis-stream}
	For any integer $p \geq 1$ and sufficiently large $n \in \IN$, any $p$-pass streaming algorithm for finding any maximal independent set with any constant probability of success strictly more than zero has space
	\[
		\Omega\Paren{   \frac1{2^{\Theta(p \log^{5/6}(n))}} \cdot n^{1+1/(2^{p}-1)}}.
	\]
	In particular, semi-streaming algorithms require $\Omega(\log\log{n})$ passes. 
\end{corollary}
\begin{proof}
	Follows from~\Cref{thm:mis-cc} and~\Cref{prop:cc-stream} since the number of players and rounds are $O(\log\log{n})$ each, and hence, the space lower bound
	is at most smaller than the communication one by an $O((\log\log{n})^2)$ factor; this is subsumed by the $2^{\Theta(\log^{5/6})}$ term of the lower bound. 
\end{proof}

We define our hard input distribution in the proof of~\Cref{thm:mis-cc} in this section, and list its main properties. We then use this distribution in the next section to conclude the proof of~\Cref{thm:mis-cc}.  

\subsection{A Family of Hard Distributions for MIS}\label{sec:hard-dist}

The distributions are defined recursively as a family $\set{\GG_r(n)}_{r \geq 0}$ where $\GG_r(n)$ is a hard distribution for $r$-round $(r+1)$-party protocols on $n$-vertex graphs. To avoid ambiguity with the notation for paths, 
we use $\player{1},\player{2},\ldots,\player{r+1}$ to denote the players in $\GG_r(n)$. In addition, the input to player $\player{a}$ for $a \in [r+1]$ is denoted by $G_a$. 

\paragraph{Base case for $r=0$.} The base case is defined vacuously for $r=0$ and $0$-round ``protocols'' that are supposed to solve a simple yet non-trivial MIS problem (and are hence trivially impossible). 

\begin{ourbox}\textbf{The distribution $\GG_0(n)$ for $0$-round protocols on $n$-vertex graphs for $n \geq 2$.}
	\begin{enumerate}
		\item Let $G=(V,E)$ be a graph on $p_0 = n/2$ vertices $V = \set{u_i, v_i \mid i \in [p_0]}$. 
		\item For each $i \in [p_0]$, the edge $(u_i, v_i)$ is independently included in $E$ with probability $1/2$ (independent of the inclusion of the other potential edges $(u_j, v_j)$, for $j \neq i$).
		\item Player $\player{1}$ receives $G(1) := G$ as the input (this is the only player when $r=0$). 
	\end{enumerate}
\end{ourbox}
The answer to MIS on $\GG_0$ contains either one of the vertices $u_i$ \emph{or} $v_i$ (but not both) with probability half, and otherwise contains \emph{both} $u_i$ and $v_i$ with the remaining probability for each $i \in [p_0]$ independently. 
Thus, a $0$-round ``protocol'', namely, one that has to commit to a fixed answer always without checking the input, can only succeed with probability $1/2^{p_0}$ on $\GG_0$. 

\paragraph{Distributions for $r \geq 1$.} A  graph generated by $\GG_r(n)$ will be an $(n, 2^{r+1})$-layered graph created by \emph{embedding} many $n_{r-1}$-vertex hard instances for $(r-1)$ rounds into $(p_r, q_r, 2^r - 1)$-DUP graphs with $b_r$ vertices in each layer. 
We will determine the values of the parameters $p_r,q_r,b_r$ and $n_{r-1}$ based on $n$ shortly, after defining the distribution itself. 
Throughout, if clear from the context or not relevant, we may omit the subscripts on $p_r,q_r$, and $b_r$, as well as $n$ from $\GG_r(n)$. 

\begin{ourbox}\textbf{Distribution $\GG_r(n)$ for $r$-round protocols for $r \geq 1$. }
	\begin{enumerate}
			\item Pick a \underline{fixed} $(p_r,q_r,2^r-1)$-DUP graph $\Gdup := (W=W_1 \cup \ldots \cup W_{2^r}, \Edup)$ with $b_r$ vertices in each layer and thus $b_r \cdot 2^r$ vertices in total. 
			\item Sample $p_r \cdot q_r$ \underline{independent} instances $\cH := \set{H_{i,j} \mid i \in [q_r], j \in [p_r]}$ from $\GG_{r-1}(n_{r-1})$. 
			\item Create two \underline{identical} copies of the graph $\tG_L = \tG_R := \embed(\cH \rightarrow \Gdup)$. 
			\item For $a \in [r]$, player $\player{a}$ receives the edges in both $\tG_L$ and $\tG_R$ that
			correspond to $H_{i,j,a}$ for  $i \in [q_r], j \in [p_r]$ where $H_{i,j,a}$ is the subgraph of $H_{i,j}$ given to the $a^\text{th}$ player in $\GG_{r-1}(n_{r-1})$. 
			\item Pick $t$ \underline{uniformly at random} from $[q_r]$ and let $\cP^{\star} := \cP_{t}$ be the $t$-th UPC in $\Gdup$, referred to as the \emph{special UPC}. Let $W^{\star} \subseteq W$ be the set of vertices in $\Gdup$ incident to $\cP^{\star}$. 
			\item Player $\player{r+1}$ receives a bipartite clique between $\set{(u,*) \mid u \in W \setminus W^{\star}}$ in $\tG_L$ and $\tG_R$. 
		\end{enumerate}
\end{ourbox}

\noindent\Cref{fig:dist-image} gives an illustration of our hard distribution. 

\begin{figure}[h!]
	\centering
	\includegraphics[width=0.9\textwidth]{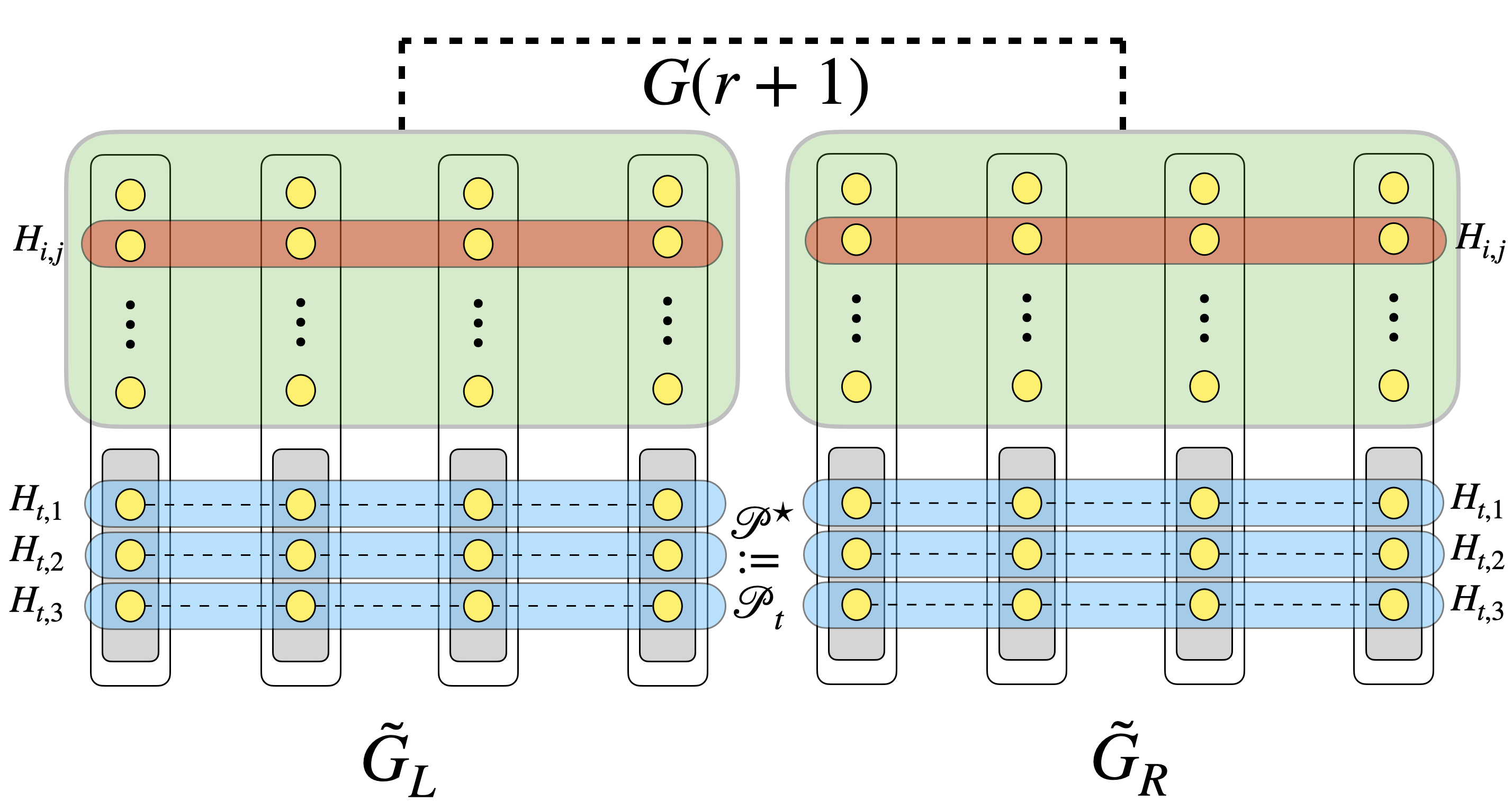}
	\caption{An illustration of the distribution $\GG_r(n)$. The paths at the bottom correspond to the UPC $\cP^{\star}$; each path is ``replaced'' with an instance $H$ from $\GG_{r-1}(n_{r-1})$ via the embedding product. 
	Since the embedding product affects the entire graph, the paths in other UPCs are also replaced by instances from $\GG_{r-1}(n_{r-1})$, e.g., for the top (red) block (remaining instances are not drawn). 
	The last player $\player{r+1}$ receives a bipartite clique between the vertices corresponding to $W \setminus W^{\star}$ in $\tG_L$ and $\tG_R$, denoted by the top (green) blocks. Other players $\player{a}$ for $a \in [r]$ 
	receive the inputs of the $a^\text{th}$ player in \emph{all} instances $H$ sampled from $\GG_{r-1}(n_{r-1})$ when placed inside graphs $\tG_L$ and $\tG_R$. 
	}\label{fig:dist-image}
\end{figure}

\subsubsection*{Definitions and Notation for Distribution $\GG_r(n)$} We set up the following notation for the distribution $\GG_r(n)$ for $r \geq 1$. For any $a \in [r+1]$, let $\GG_{r,a}(n)$ denote 
the distribution of edges in the input graph that are given to player $\player{a}$. Thus, the distribution $\GG_r(n)$ is a joint distribution of $(\GG_{r,1}(n),\ldots,\GG_{r,r+1}(n))$. 

An instance of the $\MIS{n}{r+1}$ sampled from $\GG_r(n)$ is a graph $G$ with $n$ vertices together with its edge partitioning between the $r+1$ players. With a slight abuse of notation, we sometimes refer to $G$ sampled from $\GG_r(n)$, i.e., $G \sim \GG_r(n)$, as the entire instance, where the edge partitioning between the players is implicit. Any $G \sim \GG_r(n)$ contains $q_r \cdot p_r$ instances $H_{1,1},\ldots,H_{q_r,p_r}$ sampled from $\GG_{r-1}(n_{r-1})$ that are embedded as part of both $\tG_L$ and $\tG_R$. 
We refer to these $(r-1)$-round instances as the \textbf{sub-instances} of $G$. For any $i \in [q_r], j \in [p_r]$, we further write $\tE_{L,i,j}$ and $\tE_{R,i,j}$ to denote the edges inserted to, respectively, $\tG_L$ and $\tG_R$ as part of the embedding of $H_{i,j}$
into $\Gdup$.

Among the sub-instances of any $G \sim \GG_r(n)$, there are 
$p_r$ sub-instances that correspond to the paths in the special UPC $\cP^{\star}$; we denote them by $\cHstar$ and call them  the \textbf{special sub-instances} of $G$. 
We enumerate these special sub-instances by $H_{t,j}$ for $t \in [q_r]$ being the index of the special UPC $\cP^{\star} = \cP_t$ and $j$ ranging over $[p_r]$. 

For any $j \in [p_r]$, define the \textbf{special subgraphs} $\sG_{L,j}$ and $\sG_{R,j}$ as the subgraphs of $\tG_{L}$ and $\tG_{R}$ on edges of special sub-instance $H_{t,j}$, i.e., with edges $\tE_{L,t,j}$ and $\tE_{R,t,j}$, respectively. 

\subsubsection*{The Choice of Parameters and Their Ranges}
 We now specify how the parameters $p_r,q_r,b_r$, and $n_{r-1}$ are chosen, the restriction we have on their ranges, and their ``validity'' in their recursive calls (and~\Cref{prop:dup}). 

\paragraph{Parameters $b_r,n_{r-1}$.} Let $n_0$ be an arbitrarily large constant to be determined later (this controls the probability of success of the protocols). 
The distribution $\GG_r(n)$ for $r \geq 0$ rounds requires that
\begin{align}
n \geq \frac{1}{2} \cdot (2n_0)^{{2}^{r}-1}. \label{eq:n-r-relation}
\end{align} 
For $r=1$, $n_0$ is already defined and $b_1$ is defined from $n$ as 
\begin{align}
	b_1 = \paren{\frac{n}{2 \cdot n_0}}. \label{eq:define-b1}
\end{align}
For $r > 1$, $b_r$ and $n_{r-1}$ are defined from $n$ via the following equations (recall that $(a)\uparrow(b) = a^b$)
\begin{align}
\begin{split}
	&b_r = \paren{n_{r-1}} \uparrow \paren{1+\frac{1}{2^{r-1}-1}}, \\ 
	&2 \cdot b_r \cdot n_{r-1} = n,
\end{split} \label{eq:parameter-defining}
\end{align}
which implies that
\begin{align}
\begin{split}
	n_{r-1} &= \paren{\frac{n}{2}} \uparrow \paren{\frac{2^{r-1}-1}{2^r-1}} \\
	b_{r} &= \paren{\frac{n}{2}} \uparrow \paren{\frac{2^{r-1}}{2^{r}-1}}. 
\end{split} \label{eq:parameter-n-b}
\end{align}

In the distribution $\GG_r(n)$ for any $r \geq 1$, we call the distribution $\GG_{r-1}(n_{r-1})$. For $r=1$, $\GG_1(n)$ calls $\GG_0(n_0)$ which satisfies \Cref{eq:n-r-relation} as 
\[
\frac{1}{2} \cdot (2n_0)^{2^{1}-1} = n_0.
\] 
For $r > 1$, $\GG_r(n)$ calls $\GG_{r-1}(n_{r-1})$ and by the choice of $n_{r-1}$ in~\Cref{eq:parameter-n-b}, we have, 
\begin{align}
	n_{r-1} = \paren{\frac{n}{2}} \uparrow \paren{\frac{2^{r-1}-1}{2^r-1}} \Geq{\Cref{eq:n-r-relation}} \paren{\frac{1}{2} \cdot (2n_0)^{{2}^{r}-1}} \uparrow \paren{\frac{2^{r-1}-1}{2^r-1}} \geq \frac{1}{2} \cdot \paren{2n_0}^{2^{r-1}-1}, \label{eq:r-1-n-check-validity}
\end{align}
which means the call to $\GG_{r-1}(n_{r-1})$ also satisfies~\Cref{eq:n-r-relation} for $r-1$. This gives us the following. 

\begin{observation}\label{obs:eq-n-r-relation}
For any $r \geq 1$, if the distribution $\GG_r(n)$ satisfies~\Cref{eq:n-r-relation} for $n$, then any distribution $\GG_{r'}(n')$ for $r' < r$ called recursively in $\GG_r$ also satisfies~\Cref{eq:n-r-relation} for $r',n'$. 
\end{observation}

\paragraph{Parameters $p_r,q_r$ for $r \geq 1$.} Recall that we use a $(p_r,q_r,2^{r}-1)$-DUP graph with $b_r$ vertices in each layer in $\GG_r$. This means the the total number of vertices in this DUP graph is $b_r \cdot 2^{r}$. 
We define the parameters $p_r$ and $q_r$ as follows: 
 \begin{align}
	\begin{split}
   p_r&:= \frac{b_r\cdot 2^r}{\exp(\eta_p \cdot (\ln (b_r \cdot 2^r))^{3/4})}, \\
   q_r&:= \frac{b_r\cdot 2^r}{ \exp(\eta_q \cdot (\ln (b_r \cdot 2^r))^{3/4})}. 
  \end{split}\label{eq:parameter-p-q}
\end{align}
Firstly, these bounds match those of~\Cref{prop:dup} and thus we can apply the proposition to prove the existence of the required DUP graphs in our distribution.  We should also verify that
\[
	2^{r}-1 \leq 2^{(\log{(b_r \cdot 2^r)})^{1/4}},
\]
which is required by~\Cref{prop:dup}. For $r=1$, the LHS is $1$ and this holds trivially. For $r > 1$, we have by~\Cref{eq:parameter-n-b} that
\[
	b_r =\paren{\frac{n}{2}} \uparrow \paren{\frac{2^{r-1}}{2^{r}-1}} \Geq{\Cref{eq:n-r-relation}} \paren{\frac{1}{2} \cdot (2n_0)^{{2}^{r}-1}} \uparrow \paren{\frac{2^{r-1}}{2^{r}-1}} \geq \frac{1}{2} \cdot \paren{(2n_0)^{{2}^{r-1}}}. 
\]
This satisfies the above equation because 
\[
	\log{(2^{r}-1)} \leq r \qquad \text{whereas} \qquad (\log{(b_r \cdot 2^r)})^{1/4} \geq 2^{r-4} \cdot \log{(n_0)}
\]
and thus we obtain the equation for any large constant $n_0$ (with quite some room to spare). 

\begin{observation}\label{obs:graphs-exists}
For any $r \geq 1$, if the distribution $\GG_r(n)$ satisfies~\Cref{eq:n-r-relation} for $n$, then the DUP graphs created in $\GG_r$ (and in any recursive call to $\GG_{r'}$ for $r' < r$) 
with the given parameters $p_r$ and $q_r$ do exist (by~\Cref{prop:dup}). 
\end{observation}



\subsection{Basic Properties of the Distributions $\set{\GG_{r}(n)}_{r \geq 1}$}\label{sec:properties} 

We establish some basic properties of the distribution $\GG_r(n)$ for $r \geq 1$ in this subsection.  
We then use these 
to define \emph{search predicates} that reduce our task of proving the lower bound for MIS to determining the ability of protocols 
in figuring out a certain {predicate} for graphs sampled from $\GG_r$.

The first property identifies the ``real input'' to the players. 

\begin{property}\label{prop:input-create}
	For any $a \in [r]$, the input of player $\player{a}$ in $G_a \sim \GG_{r,a}(n)$ is determined deterministically by the input of the $a^\text{th}$ player in all sub-instances in $\cH$, i.e., $\cH_a := \set{H_{i,j,a} \mid i \in [q_r], j \in [p_r]}$. 
	The input of player $\player{r+1}$ is deterministically determined by the index $t \in [q_r]$ of the special UPC. 
\end{property}
\begin{proof}
	The choice of $\Gdup$ is fixed in $\GG_r(n)$ and does not include any randomness (think of it as being ``hardcoded'' in the definition of the distribution). For any $a \in [r]$, given only $\cH_{a}$, player $\player{a}$ can know
	the edges of $\tG_L$ and $\tG_R$ that belong to its input. Similarly, given only $t$, player $\player{r+1}$ can know what vertices of $\tG_L$ and $\tG_R$ are in $W^\star$ and thus 
	which edges should be added to the bipartite clique between vertices corresponding to $W \setminus W^{\star}$ in $\tG_L$ and $\tG_R$. 
\end{proof}

The next property shows that the distributions $\set{\GG_r(n)}_{r \geq 0}$ are product distributions -- a fact which is used crucially in our lower bound arguments. 

\begin{property}\label{prop:product-dist}
		For any $r \geq 0$, the distribution $\GG_r(n)$ is a product distribution, i.e., 
		\[
			\GG_r(n) = \GG_{r,1}(n) \times \GG_{r,2}(n) \times \cdots \times \GG_{r,r+1}(n). 
		\] 
\end{property}
\begin{proof}
		The proof is by induction on $r$. The base case when $r = 0$ is trivial as the entire graph is given to $\player{1}$. 
		Let us assume that the statement is true for $r = \ell$ and we prove it for $r=\ell+1$. 
		
		The input to players in $\GG_{\ell+1}(n)$, by~\Cref{prop:input-create} is determined by $q_{\ell+1} \cdot p_{\ell+1}$ sub-instances sampled from $\GG_{\ell}(n_{\ell})$ (for players $\player{1}$ to $\player{\ell+1}$) and by $t \in [q_{\ell+1}]$ (for player $\player{\ell+2}$). 
		The distribution of each sub-instance $H_{i,j}$ for $i \in [q_{\ell+1}]$ and $j \in [p_{\ell+1}]$ is a product distribution of the inputs to $\player{a}$ for all $a \in [\ell+1]$ by the induction hypothesis. 
		These sub-instances are also all sampled independently of each other, thus making the collective input to players $\player{1}$ to $\player{\ell+1}$ independent of each other. Finally, 
		the index $t \in [q_{\ell+1}]$ is sampled independent of all other variables, implying the input to player $\player{\ell+2}$ is also independent of the rest, completing the proof. 
\end{proof}

The next two properties together identify the key properties of special subgraphs and their role w.r.t.\ any MIS of the input graph. 

\begin{property}\label{prop:induced-subgraph}
	The special subgraphs are all {{vertex-disjoint}}. Moreover, the induced subgraph of $\tG_L$ (resp. $\tG_R$) on 
	the vertices corresponding to the special subgraphs contains only the edges of these subgraphs, i.e., the edges of $\sG_{L,j}$ (resp. $\sG_{R,j}$) for all $j \in [p_r]$.
\end{property}
\begin{proof}
	The special subgraphs inserted to either $\tG_L$ or $\tG_R$ are part of the embedding product of $\embed{(\cH \rightarrow \Gdup)}$ corresponding to the special UPC $\cP^{\star}$. 
	Thus, by \Cref{lem:embed-induced}, we know that they are vertex-disjoint and no other edges from any other subgraphs are present in the vertices of the special subgraphs. 
\end{proof}

\begin{property}\label{prop:recursive-mis}
	For any graph $G \sim \GG_r(n)$ for $r \geq 1$, any MIS $S$ of $G$ also contains an MIS for every special subgraph $\sG_{L,1},\ldots,\sG_{L,p_r}$ or every special subgraph $\sG_{R,1},\ldots,\sG_{R,p_r}$. 
\end{property}
\begin{proof}
	We consider three possible cases for the proof. 
	
	\paragraph{Case 1.}	Suppose first that there exists a vertex $v \in S$ which (1) belongs to $\tG_L$ and (2) corresponds to some vertex in $W \setminus W^{\star}$. 
	\Cref{fig:mis} gives an illustration of this case.

\begin{figure}[h!]
	\centering
	\includegraphics[width=0.8\textwidth]{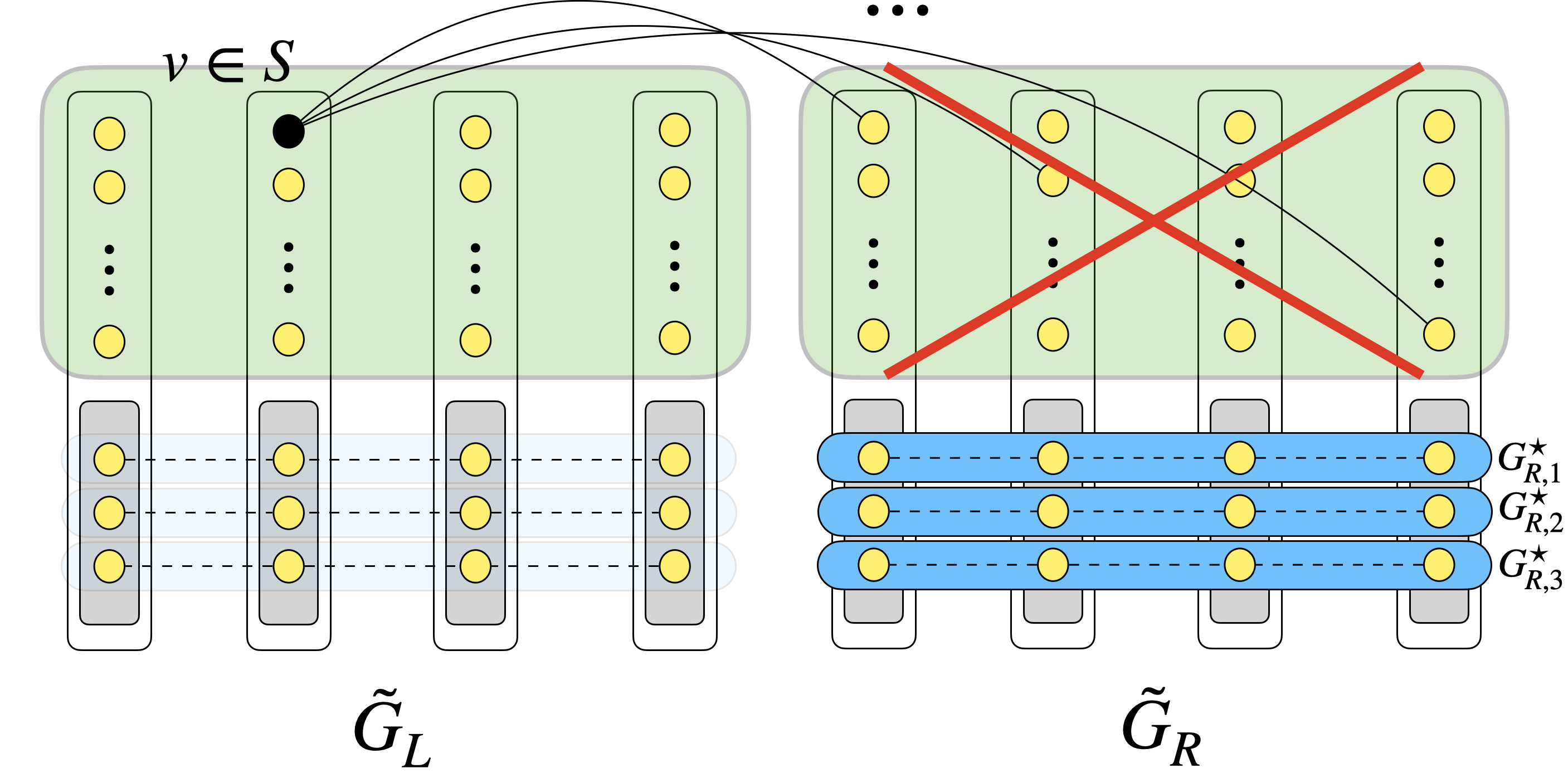}
	\caption{An illustration of case $1$: picking vertex $v$ in the MIS ``removes'' all vertices on $\tG_L$ that are not part of the special subgraphs as the remaining ones are all neighbor to $v$ (not all edges of $v$ are drawn). Since the 
	special subgraphs are now isolated, the MIS of $G$ now needs to contain a separate MIS for each of the special subgraphs in $\tG_R$ in bottom right (drawn in blue).}\label{fig:mis}
\end{figure}
	
	Since $G_{r+1} \subseteq G$ is a bipartite clique between vertices corresponding to $W \setminus W^{\star}$ in $\tG_L$ and $\tG_R$, we have that all vertices corresponding 
	to $W \setminus W^{\star}$ in $\tG_R$ are now incident on $v$ and cannot be part of the MIS $S$. The remaining vertices in $\tG_R$ are incident on $W^{\star}$ and thus $S$ 
	also needs to contain an MIS of the induced subgraph of $\tG_R$ on vertices corresponding to $W^{\star}$ or equivalently $\cP^{\star}$ -- this is because these vertices do not have any edge
	to $\tG_L$ and also no vertex of $\tG_R$ outside these can be part of $S$. 
	
	We can now apply the main property of the embedding product, i.e.,~\Cref{lem:embed-induced}, captured in~\Cref{prop:induced-subgraph} to have that the induced subgraph of $\tG_R$ on the vertices corresponding to the UPC $\cP^{\star}$ 
	is a vertex-disjoint union of subgraphs $\sG_{R,1},\ldots,\sG_{R,p_r}$. This immediately implies that the MIS of this subgraph of $\tG_R$ should be a union of the MISes of $\sG_{R,1},\ldots,\sG_{R,p_r}$, 
	which means the MIS $S$ contains an MIS for each of these special subgraphs. 
	
	\medskip

	\paragraph{Case 2.} Suppose the symmetric case that there exists a vertex $v \in S$ which (1) belongs to $\tG_R$ and (2) corresponds to some vertex in $W \setminus W^{\star}$.  The same exact argument implies that now $S$ should contain an MIS for every special subgraph $\sG_{L,1},\ldots,\sG_{L,p_r}$ instead. 

	\paragraph{Case 3.} Finally, suppose that no vertex in $S$ corresponds to a vertex of $W \setminus W^{\star}$ in either of $\tG_L$ or $\tG_R$. This, similar to the above, implies that $S$ should 
	contain an MIS for the induced subgraph of $\tG_L$ on vertices corresponding to $W^{\star}$ \emph{and} the induced subgraph of $\tG_R$ on vertices corresponding to $W^{\star}$. This in turn, again, as above, implies
	that $S$ now contains an MIS for  every special subgraph $\sG_{L,1},\ldots,\sG_{L,p_r}$ \emph{and} every $\sG_{R,1},\ldots,\sG_{R,p_r}$. 
\end{proof}

\subsection{Search Sequences and Predicates}\label{sec:search-sequence}

Before getting to analyze the distributions $\set{\GG_r(n)}_{r\geq 1}$, we need one other set of definitions, which capture
the notion of \emph{hierarchical embeddings} in our lower bounds. 

 At a high level, the analysis goes as follows. 
We are \emph{hiding} the choice of the special UPC of the instance from the first $r$ players in the first round, then UPC of special sub-instances from the first $r-1$ players in the second round, 
and so on and so forth, until at the end of the last round; at that point, we reach the ``inner most'' graphs sampled from $\GG_0$, whose edges are still hidden from the players, and now there is no more round to compute the answer. 
Making this intuition precise is going take work, but hopefully this provides some intuition for the following definitions.

\paragraph{Notation.}
Fix any $r \geq 1$. We use $p_j,q_j,n_j$ for $0 \leq j < r$ to denote the parameters $p,q,n$ in the instances sampled from the hard distribution for $j$ rounds, i.e., $\GG_j$, in the construction of the hard distribution on $r$ rounds with $n$ vertices, 
i.e., $\GG_r(n)$. 

\begin{Definition}\label{def:search-seq}
	For $r \geq 0$, we define $K = (k_r, k_{r-1}, \ldots, k_{1})$ to be a valid \textbf{search sequence} if $k_i \in [p_i]$ for each $i \in [r]$ ($K=\emptyset$ is the only valid search sequence for $r=0$). We interpret a search sequence on a graph $G=G_r \sim \GG_r(n)$ as: 
	\begin{itemize}
		\item $k_r$ points to the $k_r^{\textnormal{th}}$ special sub-instance of $G_r$, namely, $G_{r-1}:= H_{t_r,k_r} \sim \GG_{r-1}$ where $t_r$ is the index of the special UPC of $G_r$; 
		\item $k_{r-1}$ points to the $k_{r-1}^{\textnormal{th}}$ special sub-instance of $G_{r-1}$, namely, $G_{r-2} := H_{t_{r-1},k_{r-1}} \sim \GG_{r-2}$ where $t_{r-1}$ is the index of the special UPC of $G_{r-1}$; 
		\item We continue like this until $k_1$ points to the $k_{1}^{\textnormal{th}}$ special sub-instance of $G_1$, namely, $G_0 := H_{t_1,k_1} \sim \GG_{0}$ where $t_1$ is the index of the special UPC of $G_1$. 
		\item We now have a unique instance $G_0$ of $\GG_0$ defined by $K$. 
	\end{itemize}
	Finally, we define the \textbf{search predicate} of a $G \sim \GG_r(n)$ and a search sequence $K$ as:
	\[
		\PP_r(G, K) := \textnormal{a string in $\set{0,1}^{p_0}$ where the $i$-th bit is $1$ iff $(u_i,v_i)$ is an edge in $G_0$.}
	\]
\end{Definition}
\Cref{fig:search-predicate} gives an illustration of this definition. 
\begin{figure}[h!]
	\centering
	\includegraphics[scale=0.25]{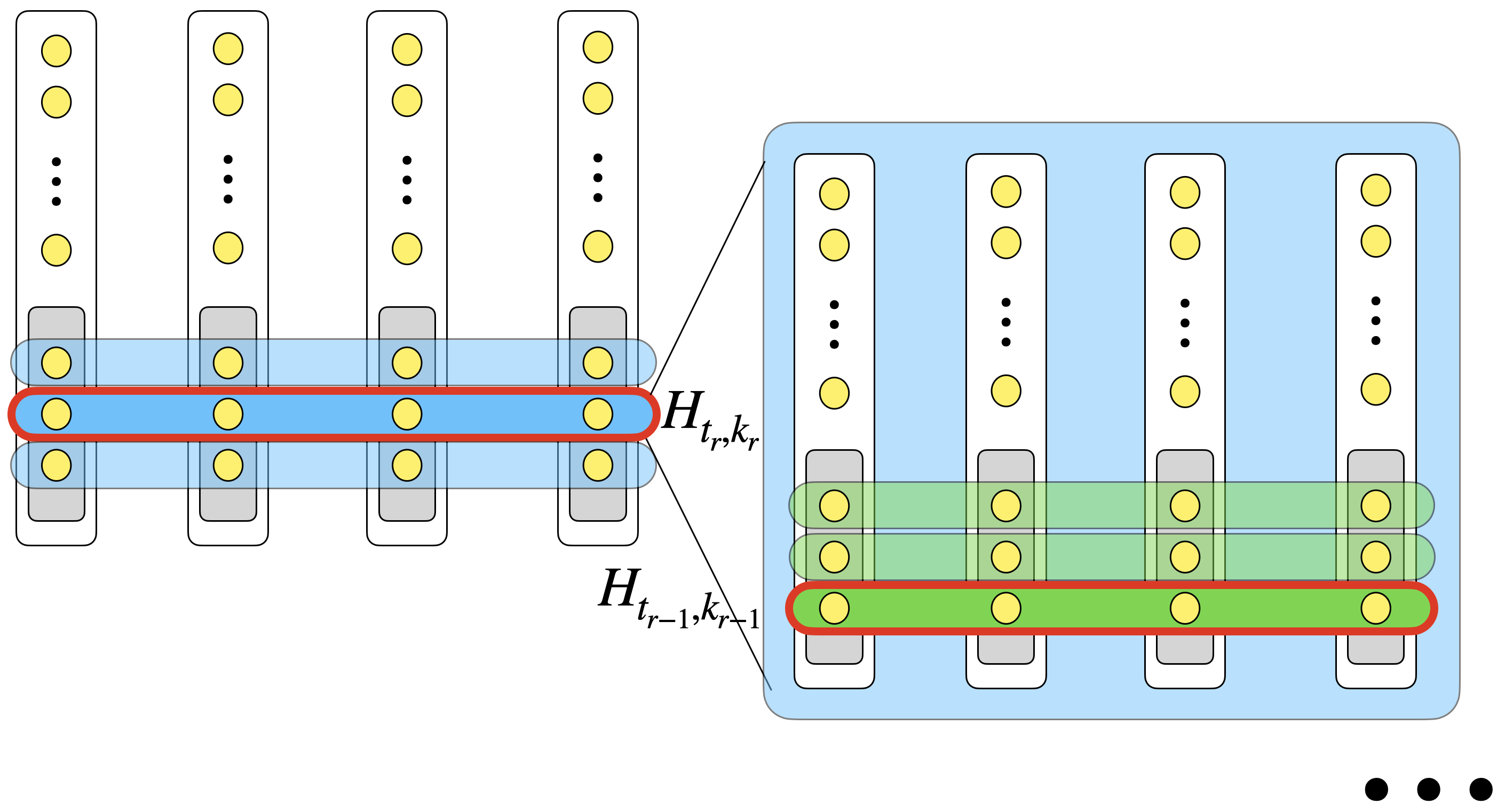}
	\caption{An illustration of a search sequence. Here, $k_r$ picks a special sub-instance $H_{t_r,k_r}$ from $\GG_{r-1}$ (with $t_r$ being the index of the special UPC of the outer instance), and then ``inside'' that instance, 
	$k_{r-1}$ picks another sub-instance $H_{t_{r-1},k_{r-1}}$ from $\GG_{r-2}$ (with $t_{r-1}$ being the index of the special UPC of the instance $H_{t_r,k_r}$), and this continues on.}\label{fig:search-predicate}
\end{figure}

Let $\prot$ be any $(r+1)$-party communication protocol on inputs sampled from $\GG_r(n)$, and let $\Prot = \Prot(G)$ be the transcript of $\prot$ at the end of $r$ rounds on an input $G \sim \GG_r(n)$. 
We say that the protocol $\prot$ \textbf{solves the search predicate} $\PP_r$ for  input graph $G$, if only given the transcript $\Prot$ and any search sequence $K$ (with no further access to $G$), 
$\PP_r(G,K)$ can be deterministically determined for \emph{all} valid search sequences $K$. 
We shall emphasize that the search sequence $K$ is \emph{not} a part of the input to the protocol $\prot$ (and is in fact even independent of the input graph $G$). 

The following lemma shows why search sequences are important to us.

\begin{lemma}\label{lem:reveal-search-bit}
	 For every $\delta \in (0,1)$, any $(r+1)$-party protocol $\prot$ that outputs an MIS of $G \sim \GG_r(n)$ with probability at least $\delta$, 
	 also solves the search predicate $\PP_r(G,K)$ for \emph{every} valid search sequence $K$ with probability at least $\delta$. 
\end{lemma}

\begin{proof}
	We claim that:
	\begin{quote}
		\vspace{-0.5em}
	For any $r \geq 0$, any MIS of $G \sim \GG_r(n)$ uniquely determines $\PP_r(G,K)$ for all valid search sequences $K$. 
		\vspace{-0.5em}
	\end{quote} 
	We can then conclude the proof as follows.
	Protocol $\prot$ will be computing some MIS of $G$ with probability at least $\delta$. Thus, whenever $\prot$ succeeds, the transcript $\Prot$ contains an MIS that can be used to solve the predicate $\PP_r(G,K)$ for 
	all $K$ with probability at least $\delta$ as well. 

	We now prove the above statement by induction on $r$. The base case when $r=0$ 
	is a graph $G_0$ consisting of $n_0 = 2p_0$ vertices $u_i$ and $v_i$ for each $i \in [p_0]$ and $K=\emptyset$. Each such $u_i$ and $v_i$ pair for $i \in [p_0]$ is possibly an edge in $G_0$. If the edge does exist, any MIS necessarily should have exactly one of $u_i$ or $v_i$ but not both, whereas if the edge
	does not exist, then any MIS of $G_0$ should contain both $u_i$ and $v_i$. Thus, given any MIS of $G_0 \sim \GG_0$, we can determine the predicate $\PP_0(G,\emptyset)$. 

	For the induction step, suppose the statement is true for $r-1$ and we prove it for $r$. Let $K:=(k_{r},\ldots,k_1)$ be any search sequence. By~\Cref{prop:recursive-mis}, any MIS of $G \sim \GG_{r}$ also 
	contains an MIS either for \emph{every} special subgraph $\sG_{L,1},\ldots,\sG_{L,p_r}$ or for \emph{every} special subgraph $\sG_{R,1},\ldots,\sG_{R,p_r}$. 
	In particular, it contains an MIS for either $\sG_{L,k_{r}}$ or $\sG_{R,k_{r}}$. Both these graphs are identical and correspond to the $k_{r}^{\textnormal{th}}$ special sub-instance $H_{t,k_{r}}$ in $G$ 
	where $t$ is the index of the special UPC. As $H_{t,k_{r}} \sim \GG_{r-1}$, by the induction hypothesis, an MIS of $H_{t,k_{r}}$ uniquely determines 
	\[
	\PP_{r-1}(H_{t,k_{r}},(k_{r-1},k_{r-2},\ldots,k_1)).
	\]
	But, by definition, we also have 
	\[
		\PP_{r}(G,K) = \PP_{r-1}(H_{t,k_{r}},(k_{r-1},k_{r-2},\ldots,k_1)),
	\] 
	as $k_{r} \in K$ is just pointing to the special sub-instance $H_{t,k_{r}}$. This finalizes the proof. 
\end{proof}


%% file: analysis.tex
\newcommand{\rT}{\rv{t}}
\newcommand{\rH}{\rv{H}}
\newcommand{\rK}{\rv{k}}

\newcommand{\sigmae}{\overline{\sigma}}
\newcommand{\Sigmae}{\overline{\Sigma}}

\newcommand{\rSigma}{\rv{\Sigma}}

\section{Analysis of the Hard Distribution}\label{sec:analysis} 

We present the analysis of the lower bound for distributions $\set{\GG_r}_{r \geq 1}$ in this section, and conclude the proof of~\Cref{thm:mis-cc}. 
By~\Cref{lem:reveal-search-bit}, we need to focus on the ability of protocols for solving the search predicate. The following lemma 
captures our lower bound for this task. 

\begin{lemma}\label{lem:communication-lb}
	For any $r \geq 1$, any $r$-round protocol $\prot$ that given $G \sim \GG_r(n)$ (for $n,r$ satisfying~\Cref{eq:n-r-relation}), can solve $\PP_r$ on input graph $G$ with
	 probability of success at least 
	 \[
	 2^{-p_0} \cdot \paren{1 + \frac{r}{20 \cdot (r+1)}}
	 \]
	 has communication cost
	\[
		\cc{\prot} \geq s_r(n) := \frac1{{n_0}^2 \cdot (2^{p_0} \cdot 40 \cdot \raoconst)^r \cdot ((r+1)!)^2 \cdot e^{3r\ln^{5/6}(n)}} \cdot  \paren{{(n)}\uparrow \paren{1+\frac{1}{2^{r}-1}}}.
	\]
\end{lemma}
We do note that ignoring all extra (and lower order) terms in the lemma (that are needed for a proper inductive argument), the lemma simply says that obtaining any probability of success better than $2^{-p_0}$ (an arbitrarily small constant) 
requires roughly $n^{1+1/(2^r-1)}$ communication.

We prove~\Cref{lem:communication-lb} inductively using a \emph{round elimination} argument: if we have a ``very good'' $r$-round protocol, then we should be able to eliminate its first round and 
also obtain a ``good enough'' $(r-1)$-round protocol; keep doing this then eventually bring us to a ``non-trivial'' protocol for the $r=0$ case, which we know cannot exist. 
The following lemma---which is the heart of the proof---allows us to establish the induction step.

\begin{lemma}\label{lem:round-elim-final}
	For every $0 < \eps<\delta <1$ and integer $s \geq 1$, the following is true.  
	Suppose there is a deterministic $r$-round $(r+1)$-party protocol $\prot_r$ with communication cost $\cc{\prot_r} \leq s$ that solves predicate $\PP_r$ with probability at least $\delta$ for a graph $G \sim \GG_r(n)$ (for $n,r$ satisfying~\Cref{eq:n-r-relation}). 
	Then, for constant $\raoconst$ (from \Cref{eq:msg-compress-final}), there is an $(r-1)$-round $r$-party deterministic protocol $\prot_{r-1}$ with communication cost 
	\[
	\cc{\prot_{r-1}} \leq \frac{\raoconst}{\eps} \cdot (\frac{s}{p_r}+ r^2)
	\]
	that solves $\PP_{r-1}$ for $G' \sim \GG_{r-1}(n_{r-1})$ (for $n_{r-1}$ from~\Cref{eq:parameter-n-b}) with probability of success at least
	\[
	\delta - \eps - \sqrt{\frac{s}{2p_r \cdot q_r}}. 
	\]
\end{lemma}

We spend the bulk of this section in proving~\Cref{lem:round-elim-final}. We then use this lemma to prove~\Cref{lem:communication-lb} easily in \Cref{sec:proof-communication-lb} 
and subsequently use it to conclude the proof of~\Cref{thm:mis-cc} in \Cref{sec:proof-mis-cc}.

\subsection{The Setup for the Proof of~\Cref{lem:round-elim-final}} 

We now start the proof of~\Cref{lem:round-elim-final} which is the most technical part of the paper. Fix any $r \geq 1$ and let $\prot_r$ be a $r$-round $(r+1)$-party protocol for solving $\PP_{r}$ 
with the parameters specified in~\Cref{lem:round-elim-final}. We further define (or recall) the following notation: 

\begin{itemize}
	\item $\GG_{r,a}(n)$ for $a \in [r+1]$ to denote the distribution of the input subgraph $G_a$ given to player $\player{a}$; we further use $H_a$ to denote the input of $\player{a}$ in a specific sub-instance $H \in \cH(G)$. 
	\item $\rT$ to denote the random variable for the index $t \in [q_r]$ of the special UPC picked in $\GG_r(n)$. 
	\item $H_{i,*}$ for any $i \in [q_r] $ to denote the  sub-instances $(H_{i, 1}, H_{i, 2}, \ldots, H_{i, p_r})$ together in $G \sim \GG_r(n)$. The sub-instances given to player $\player{a}$ for $a \in [r+1]$ 
	are denoted by $H_{i,*,a} = (H_{i, 1,a}, H_{i, 2,a}, \ldots, H_{i, p_r,a})$. 
	\item $\Pi = (\Pi_1,\ldots,\Pi_r)$ as the set of messages communicated by the players in rounds $1$ to $r$. Similarly, we use $\Pi_{i,a}$ for $a \in [r+1]$ to denote the messages of player $\player{a}$ in round $i \in [r]$, 
	and $\Pi_{*,a}$ to denote all messages of $\player{a}$. We use $\rProt, \rProt_i,$ and $\rProt_{i,a}$ as the corresponding random variables for these messages. 
\end{itemize}

Our goal is to go from a $r$-round $(r+1)$-party protocol $\prot_r$ for $\PP_r$ with communication cost $s$ and probability of success $ \delta$, to a 
$(r-1)$-round $r$-party protocol $\prot_{r-1}$ for $\PP_{r-1}$ with communication cost $\approx s/(\eps \cdot p_r)$ and success probability $\approx \delta-\eps - \sqrt{s/(p_r \cdot q_r)}$. We do this in two steps: 
\begin{itemize}
\item \textbf{Step 1:} We first \emph{shave off} one player from $\prot_r$ and obtain an intermediate $r$-party protocol $\prot'$ with communication cost $\approx s/(\eps \cdot p_r)$ and success probability $\approx\delta - \eps$ 
for $\PP_{r-1}$. However, $\prot'$ still has $r$ rounds instead of our desired $(r-1)$ rounds. 
\item \textbf{Step 2:} We then \emph{shave off} one round from $\prot'$ and obtain the protocol $\prot_{r-1}$ without increasing the communication but by decreasing the success probability with another $\approx \sqrt{s/(p_r \cdot q_r)}$ term. 
\end{itemize}

We implement each step in the following two subsections. We do emphasize that these steps are not entirely blackbox and this partitioning into the two steps is more for the simplicity of 
exposition (and there will be some intermediate steps as well). \Cref{fig:schematic} gives a schematic organization of these steps, in particular, the protocols we build along the way and their properties.

\begin{figure}[h!]
	\centering
	\begin{subfigure}{.75\textwidth}
	\includegraphics[scale=0.18]{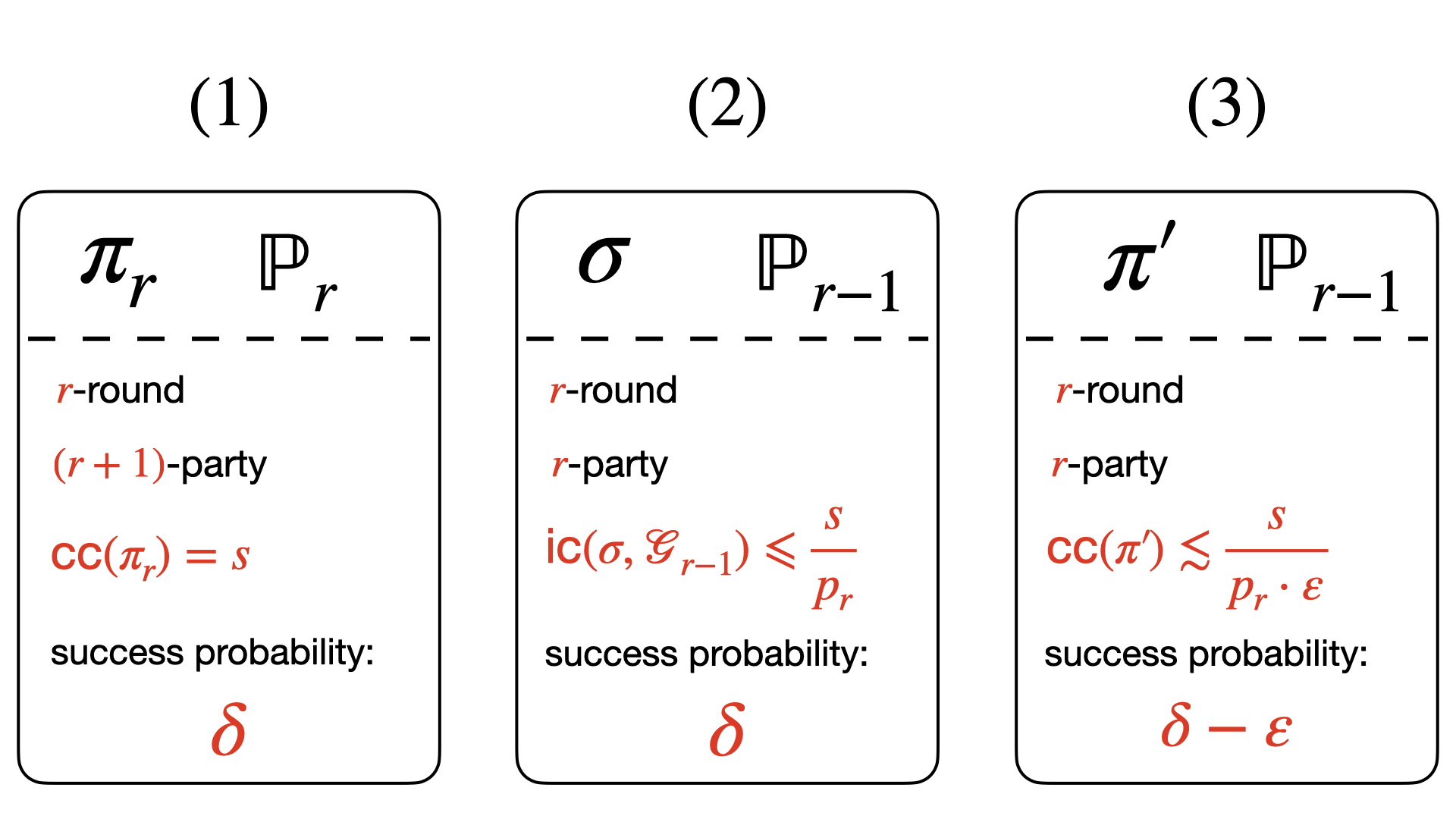}
	\end{subfigure}
	~\!\!\!\!\!\!
	\begin{subfigure}{0.25\textwidth}
	\includegraphics[scale=0.18]{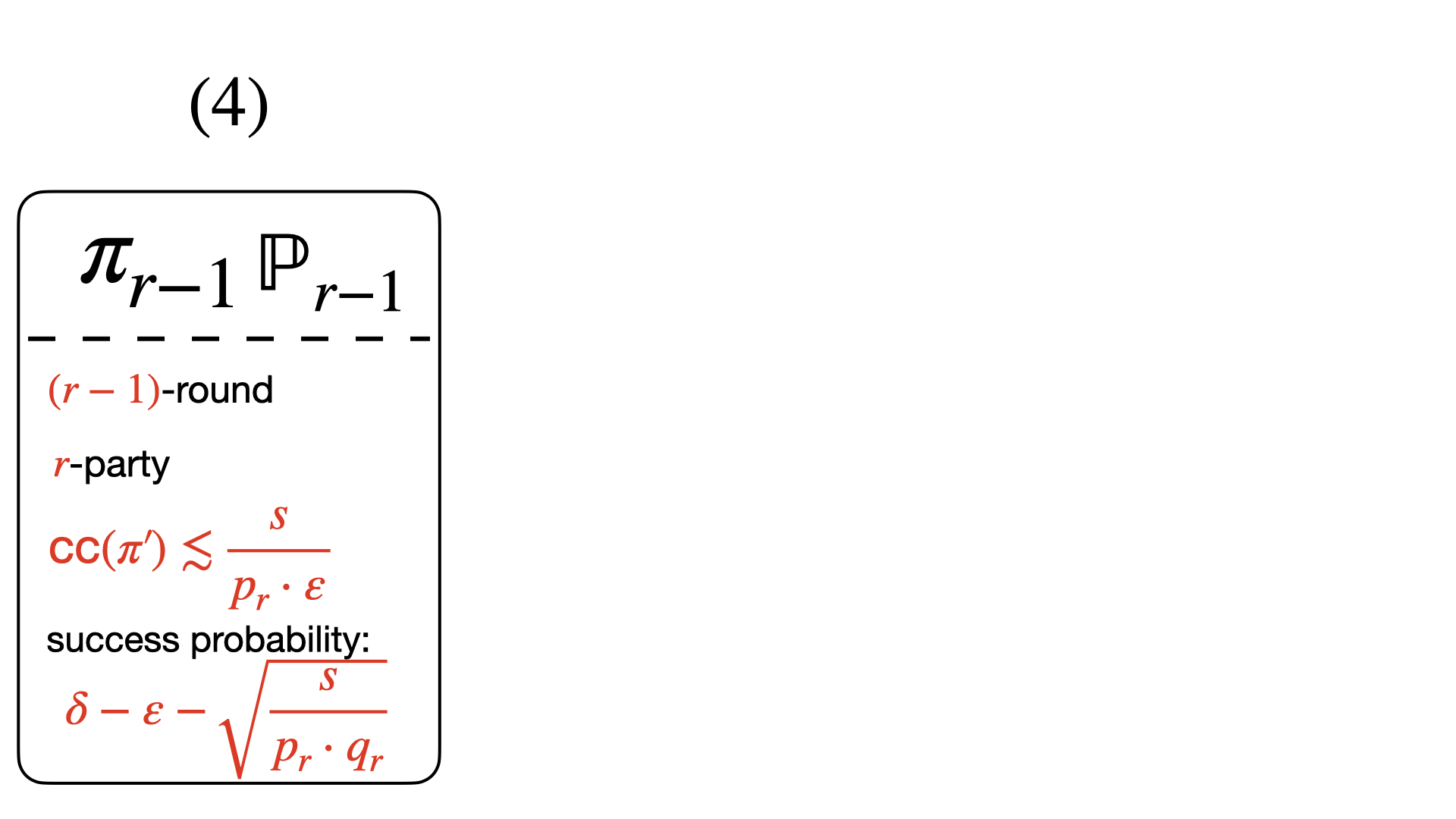}
	\end{subfigure}
	\caption{A schematic organization of the proof of~\Cref{lem:round-elim-final}: step 1 encompasses moving from part 1 to part 3 here, and step 2 is for moving from part 3 to part 4.}
	\label{fig:schematic}
\end{figure}

\subsection{Step 1: A Low Communication Protocol for $\PP_{r-1}$ in $r$ Rounds}\label{sec:step-one}

Fix $r \geq 1$ and $n \in \IN$ satisfying~\Cref{eq:n-r-relation} and let $n_{r-1}$ be determined from $n$ as in~\Cref{eq:parameter-n-b}. 
Suppose we have an instance $G_{r-1}$ sampled from $\GG_{r-1}(n_{r-1})$, and we want to use $\prot_r$, which is for $\GG_r(n)$, to solve $\PP_{r-1}$ for all search sequences $K$ (note that we use $K$ to denote search sequences for $\PP_{r-1}$ and not $\PP_r$). 
Consider the following direct way of doing this: 

\begin{ourbox}
	\textbf{An $r$-round $r$-party protocol $\sigma$ for $\PP_{r-1}$ on input graph $G_{r-1}(n_{r-1})$}: 
	\begin{enumerate}
		\item Sample indices $\kstar \in [p_r]$ and $t \in [q_r]$ uniformly at random  using \underline{public randomness}. 
		\item For every $a \in [r]$, each player $\player{a}$ samples 
		$
		\Paren{G_{r,a} \sim \GG_r \mid \rT=t , \rH_{t,\kstar,a}=G_{r-1,a}}
		$
		using \underline{private randomness}. Let $G_r=(G_{r,1},G_{r,2},\ldots,G_{r,r+1})$ be the 
		input of all players in $\GG_r$ where, using~\Cref{prop:input-create}, the input of the ``simulated'' player $\player{r+1}$ is fixed by $\rT=t$.
		\item The players run $\prot_r(G_r)$ by \underline{simulating} the messages of player $\player{r+1}$ in $\prot_r$ given that $t$ is public knowledge. 
		\item At the end, for any search sequence $K=(k_{r-1},\ldots,k_1)$, to solve $\PP_{r-1}(G_{r-1},K)$, we return the answer of $\prot_r$ for $\PP_r(G_r, (\kstar, k_{r-1},k_{r-2},\ldots,k_1))$.  
	\end{enumerate}
\end{ourbox}

Define the \textbf{distribution $\reald$} as the distribution of the graph $G_r$ obtained in the protocol $\sigma$ when $G_{r-1} \sim \GG_{r-1}$. We claim that this is the ``right'' distribution. 

\begin{observation}\label{obs:reald-is-real}
	Distribution $\reald$ is the same as $\GG_r$. 
\end{observation}
\begin{proof}
	We know from \Cref{prop:product-dist} that $\GG_r$ is a product distribution of $\GG_{r,1} \times \ldots \GG_{r,r+1}$. Moreover, each $\GG_{r,a}$ for $a \in [r]$ is, by~\Cref{prop:input-create}, 
	a collection of \emph{independent} instances $H_{i,j,a} \sim \GG_{r-1,a}$ and $t \in [q_r]$ is chosen uniformly at random in $\GG_r$ to define the input of player $\player{r+1}$. This matches 
	exactly the distribution of $\reald$ in $\sigma$, given that $G_{r-1,a}$ used by player $\player{a}$ is sampled from $\GG_{r-1,a}$ by definition. 
\end{proof}

As protocol $\prot_r$ can solve $\PP_r$, the transcript of $\prot_r$ can determine $\PP_r(G_r, (\kstar, k_{r-1}, \ldots, k_1))$ for every $K = (k_{r-1}, \ldots, k_1)$ and any $\kstar \in [p_r]$.
By~\Cref{obs:reald-is-real}, we know that
the probability of success of $\sigma$ in solving $\PP_{r-1}$ is the same as that of $\prot_r$ for solving $\PP_{r}$ and is thus at least $\delta$ by the statement of \Cref{lem:round-elim-final}. 

It seems however that we have done nothing yet: $\sigma$ is a $r$-round protocol with the \emph{same} communication cost as $\prot_r$, so effectively we made no progress. 
The silver lining is that we can actually prove $\sigma$ has a much lower \emph{information cost} compared to $\prot_r$, using a direct-sum style argument. 

\begin{claim}\label{clm:low-info-cost}
	The information cost of protocol $\sigma$ on the distribution $\GG_{r-1}$ is at most 
	\[
		\ic{\sigma}{\GG_{r-1}} \leq \frac{1}{p_r} \cdot \ic{\prot}{\GG_r}. 
	\]
\end{claim}
\begin{proof}
Let $\rK$ denote the random variable for the index $\kstar$ in $\sigma$ and $\rProt$ denote both the messages of players in $\sigma$ as well as $\prot_r$ -- this is because, 
the players in $\sigma$ communicate exactly the same messages as $\prot_r$ (without loss of generality, we can assume $\player{r}$ also writes the message of $\player{r+1}$ on the board, even though all players 
can calculate that message on their own also). And, by~\Cref{obs:reald-is-real}, we obtain that these messages are also distributed exactly the same way as $\prot_r$ is a deterministic function of the samples from $\GG_r = \reald$. 
Moreover, since $\prot_r$ is a deterministic protocol, the only \emph{public} randomness of $\sigma$ is $\rK$ and $\rT$. 
Thus, 
	\begin{align}
		\ic{\sigma}{\GG_{r-1}} &= \mi{\rG_{r-1}}{\rProt \mid \rK,\rT} \tag{by~\Cref{def:ext-info}} \\
		&= \sum_{k=1}^{p_r} \frac{1}{p_r} \cdot \mi{\rH_{\rT,k}}{\rProt \mid \rT,\rK=k} \tag{by the distribution of $\rK$ and since $G_{r-1} = H_{t,k}$ in $\sigma$ for any choice of $t \in [q_r]$}  \\
		&= \sum_{k=1}^{p_r} \frac{1}{p_r} \cdot \mi{\rH_{\rT,k}}{\rProt \mid \rT}, \label{eq:plugin-1}
		\end{align}
		where the final equality holds because of the following: the joint distribution of $(\rH_{\rT,k},\rProt,\rT)$ is a deterministic function of the choice of $\rG_r \sim \reald$ as fixing the graph $G_r$ also 
		fixes the sub-instance $H_{t,k}$ (where $t$ is the index of the special UPC of $G_r$) as well as all the messages of protocol $\prot_r = \prot_r(G_r)$ which is deterministic. On the other hand, 
		even given a fixed choice of $G_r$, we can still pick $\rK$ uniformly at random from $[p_r]$ as it is entirely independent of $G_r$. Thus, the joint distribution of $(\rH_{\rT,k},\rProt,\rT)$ is independent 
		of the event $\rK = k$ and we can drop the conditioning. 
		
		The next observation is that for every $k \in [p_r]$, we have $\rH_{\rT, k} \perp \rH_{\rT, <k} \mid \rT$ because the sub-instances are sampled independently and, after conditioning on $\rT$, 
		the choice of $\rH_{\rT, k}$ and $\rH_{\rT, <k}$ only depend on the sub-instances. Thus, we can apply~\Cref{prop:info-increase} to each of the mutual information terms above and get:
	\begin{align*}
		\sum_{k=1}^{p_r} \mi{\rH_{\rT, k}}{\rProt \mid \rT}&\leq \sum_{k=1}^{p_r} \mi{\rH_{\rT, k}}{\rProt \mid \rT, \rH_{\rT,<k}} \\
		&= \mi{\rH_{\rT, *}}{\rProt \mid \rT}  \tag{by the chain rule of mutual information in \itfacts{chain-rule}} \\
		&\leq \mi{\rH_{\rT, *}}{\rProt \mid \rT} + \mi{\rH_{-\rT, *}}{\rProt \mid \rH_{\rT, *},\rT} + \mi{\rT}{\rProt} \tag{as mutual information is non-negative (\itfacts{info-zero})} \\
		&= \mi{\set{\rH_{i,j}}_{i \in [q_r],j\in [p_r]}, \rT}{\rProt}  \tag{by the chain rule of mutual information in \itfacts{chain-rule}} \\
		&=\mi{\rG_r}{\rProt} \tag{as $G_r$ fixes $\set{H_{i,j}}_{i,j}, T$ and vice-versa by \Cref{prop:input-create}} \\
		&= \ic{\prot_r}{\GG_r} \tag{by~\Cref{def:ext-info}}. 
	\end{align*}	
	Plugging in this bound in~\Cref{eq:plugin-1}, we conclude that,
	\[
		\ic{\sigma}{\GG_{r-1}}  \leq \frac{1}{p_r} \cdot \ic{\prot_r}{\GG_r}. \qedhere
	\]
\end{proof}

We are not done however as in the next step, we \emph{really} need the communication cost of the protocol to be small and not only its information cost (see~\Cref{rem:faked-ic}). We do this by ``compressing'' the messages of $\sigma$ down to their information cost. 
This part uses standard ideas except for the fact that we are applying them to a multi-party protocol instead of their typical two-party application.

In the following, we design a protocol $\sigmae$ by compressing the messages of $\sigma$ on an input $G_{r-1} \sim \GG_{r-1}$ down to their information cost. To avoid ambiguity, here, for every $a \in [r]$ and $i \in [r]$, 
we use $\Sigma_{i,a}$ to denote the message sent by player $\player{a}$ of $\sigma$ in round $i$. We also use $\rR_{\sigma}$ to denote the public randomness of $\sigma$. 
We define $\Sigmae_{i,a}$ and $\rR_{\sigmae}$ analogously for $\sigmae$. Protocol $\sigmae$ is as follows: 
\begin{ourbox}
	\textbf{A \underline{communication efficient} implementation of $\sigma$ as a protocol $\sigmae$}: 
	\begin{itemize}
		\item For $i=1$ to $r$ rounds and $a=1$ to $r$ players in this order: 
			\begin{enumerate}
				\item Suppose the players at this point \emph{all} know $R_{\sigma}$ as well as
				\[
					\Sigma_{<i} := \Sigma_{<i,r},\ldots,\Sigma_{<i,r}, \quad \text{and} \quad \Sigma_{i,<a} := \Sigma_{i,1},\ldots,\Sigma_{i,a-1},
				\]
				i.e., the public randomness of $\sigma$ and the messages that ``should have been'' communicated by $\sigma$ in the first $i-1$ rounds plus the ones by $\player{1},\ldots,\player{a-1}$ in round $i$. 
				\item Define the distributions 
				\[
					\mathcal{A}_{i,a} := \rSigma_{i,a} \mid \Sigma_{<i},\Sigma_{i,<a}, R_{\sigma}, G_{r-1,a}, \quad \text{and} \quad \mathcal{B}_{i,a} := \rSigma_{i,a} \mid \Sigma_{<i},\Sigma_{i,<a}, R_{\sigma}; 
				\]
				notice that the difference is only that in $\mathcal{A}_{i,a}$ we also condition on the input of player $\player{a}$ in $G_{r-1}$, i.e., $G_{r-1,a}$ but in $\mathcal{B}_{i,a}$, we do not. 
				\item The player $\player{a}$ plays the role of Alice in~\Cref{prop:msg-compress}, who knows both $\mathcal{A}_{i,a}$ and $\mathcal{B}_{i,a}$ and all other players play
				the role of Bob who only knows $\mathcal{B}_{i,a}$. The players use public randomness and Alice (i.e., player $\player{a}$) writes the message of~\Cref{prop:msg-compress}, denoted by $\sigmae_{i,a}$, on the blackboard
				which allows \emph{all} players to sample a \emph{single} message
				\[
					\Sigma_{i,a} \sim \mathcal{A}_{i,a} = \rSigma_{i,a} \mid \Sigma_{<i},\Sigma_{i,<a}, G_{r-1,a},R_{\sigma}. 
				\] 
				This allows the players to obtain the message $\Sigma_i(a)$ and continue this for-loop. 
			\end{enumerate}
	\end{itemize}
\end{ourbox}

The fact that $\sigmae$ is \emph{faithfully} simulating the protocol $\sigma$ follows immediately from~\Cref{prop:msg-compress} as the message compression approach 
has no error. We now bound the length of messages communicated by $\sigmae$ \emph{in expectation}. 

\begin{claim}\label{clm:simulate-prot-1}
		The \emph{expected} length of the messages communicated by the protocol $\sigmae$ is at most 
		\[
			\raoconst \cdot \Paren{\ic{\sigma}{\GG_{r-1}} + r^2}, 
		\]
		where $\raoconst \geq 1$ is an absolute constant defined in~\Cref{eq:msg-compress-final} (derived from \Cref{prop:msg-compress}). 
\end{claim}
\begin{proof}
	With a slight abuse of notation, we denote $\card{\Sigmae}$ to denote the length of the messages communicated by $\sigmae$ (on a particular input and realization of randomness) 
	and $\card{\Sigmae_{i,a}}$ for $i \in [r]$ and $a \in [r]$ as the length of the message $\Sigmae_{i,a}$. Thus, by linearity of expectation, 
	\[
		\Exp\card{\Sigmae} = \sum_{i=1}^{r} \sum_{a=1}^{r} \Exp\card{\Sigmae_{i,a}}. 
	\]
	Moreover, by~\Cref{prop:msg-compress}, for every $i \in [r]$ and $a \in [r]$, we have, 
	\begin{align*}
		\Exp\card{\Sigmae_{i,a}} &\leq \raoconst \cdot \Paren{\Exp_{\Sigma_{<i},\Sigma_{i,<a}, G_{r-1,a},R_{\sigma}} \kl{\mathcal{A}_{i,a}}{\mathcal{B}_{i,a}} + 1} 
		\tag{as the distributions of~\Cref{prop:msg-compress} are $\mathcal{A}_{i,a},\mathcal{B}_{i,a}$ (conditioned on the prior messages)} \\
		&\leq \raoconst \cdot \Paren{\Exp_{\Sigma_{<i},\Sigma_{i,<a}, G_{r-1,a},R_{\sigma}} \kl{\rSigma_{i,a} \mid \Sigma_{<i},\Sigma_{i,<a}, R_{\sigma}, G_{r-1,a}}{\rSigma_{i,a} \mid \Sigma_{<i},\Sigma_{i,<a}, R_{\sigma}} + 1}
		\tag{by the definition of $\mathcal{A}_{i,a},\mathcal{B}_{i,a}$} \\
		&= \raoconst \cdot \Paren{\mi{\rSigma_{i,a}}{\rG_{r-1,a} \mid \rSigma_{<i},\rSigma_{i,<a}, \rR_{\sigma}} + 1}. \tag{by~\Cref{fact:kl-info}}
	\end{align*}
	Plugging these bounds in the expectation-term above implies 
	\begin{align*}
		\Exp\card{\Sigmae} &\leq \sum_{i=1}^{r} \sum_{a=1}^{r}\raoconst \cdot \Paren{\mi{\rSigma_{i,a}}{\rG_{r-1,a} \mid \rSigma_{<i},\rSigma_{i,<a}, \rR_{\sigma}} + 1} \\
		&\leq \sum_{i=1}^{r} \sum_{a=1}^{r}\raoconst \cdot \Paren{\mi{\rSigma_{i,a}}{\rG \mid \rSigma_{<i},\rSigma_{i,<a}, \rR_{\sigma}} + 1} \tag{by~\itfacts{data-processing}  as $\rG_{r-1,a}$ is a deterministic function of $\rG_{r-1}$} \\
		&=  \raoconst \cdot \Paren{\mi{\set{\rSigma_{i,a}}_{i \in [r], a \in [r]}}{\rG \mid \rR_{\sigma}} + r^2} \tag{by the chain rule of mutual information in~\itfacts{chain-rule}} \\
		&= \raoconst \cdot \Paren{\ic{\sigma}{\GG_{r-1}} + r^2},
	\end{align*}
	by the definition of information cost of $\sigma$ in~\Cref{def:ext-info}. 
\end{proof}

Finally, we do a basic clean up of $\sigmae$ to bound its communication cost (in worst-case, not in expectation). This step is entirely straightforward and is basically by truncating the protocol $\sigmae$ whenever
 a player is going to communicate more than their ``budget''. We only provide the full protocol here for completeness. 

\begin{ourbox}
	\textbf{Protocol $\prot'$: the truncated version of the protocol $\sigmae$ for a given parameter $\eps > 0$}
	\begin{itemize}
		\item Run the protocol $\sigmae$ as is. If at any point, a player is going to send a message that results in the communication 
		cost of the protocol to become more than 
		\[
			\frac{\raoconst}{\eps} \cdot \Paren{\ic{\sigma}{\GG_{r-1}} + r^2}, 
		\]
		terminate the protocol and return `fail'. 
	\end{itemize}
\end{ourbox}

We list the properties of the protocol $\prot'$ in the following. 

\begin{claim}\label{clm:prot'}
	For every $\eps > 0$, protocol $\prot'$ is a $r$-round $r$-party protocol for solving $\PP_{r-1}$ on a graph $G_{r-1} \sim \GG_{r-1}$ with probability of success at least $\delta-\eps$ and communication cost
	\[
		\cc{\prot'} \leq \frac{\raoconst}{\eps} \cdot \Paren{\frac{s}{p_r} + r^2}. 
	\]
\end{claim}
\begin{proof}
	The fact that $\prot'$ is $r$-round and $r$-party is by definition. A direct application of Markov bound, plus~\Cref{clm:simulate-prot-1} also implies that the probability that $\prot'$ terminates is at most $\eps$. 
	This in turn implies that the error probability of $\prot'$ compared to that $\sigmae$ and equivalently $\sigma$ is at most $\eps$ more. Thus, $\prot'$ succeeds with probability at least $\delta-\eps$. 
	Finally,
	\[
		\cc{\prot'} \Leq{(1)} \frac{1}{\eps} \cdot \raoconst \cdot \Paren{\ic{\sigma}{\GG_{r-1}} + r^2} \Leq{(2)} \frac{\raoconst}{\eps} \cdot \Paren{\frac{1}{p_r} \cdot \ic{\prot_r}{\GG_r} + r^2} 
		\Leq{(3)}  \frac{\raoconst}{\eps} \cdot \Paren{\frac{s}{p_r} + r^2},
	\]
	where $(1)$ holds because of the truncation step in $\prot'$, $(2)$ holds by~\Cref{clm:low-info-cost}, and $(3)$ holds since by~\Cref{prop:ic-cc}, $\ic{\prot_r}{\GG_r} \leq \cc{\prot_r}$, which is the 
	parameter $s$ in~\Cref{lem:round-elim-final}. 
\end{proof}

\subsection{Step 2: A Protocol for $\PP_{r-1}$ in $(r-1)$ Rounds via Round Elimination}\label{sec:step-two}

Up until this point, we managed to obtain a protocol $\prot'$ which has the desired communication cost and probability of success, but the main issue remains: it still requires $r$ rounds as opposed to $(r-1)$ rounds. 
In the following, we try a more nuanced way of creating an $(r-1)$-round protocol $\prot_{r-1}$ which follows the same approach as $\prot'$ but no longer uses the first round of messages of $\prot_r$ at all -- instead, the players 
simply sample those messages using public randomness. 

Before we proceed however, we need to establish an important property of the original protocol $\prot_r$ using the fact the distribution $\GG_r$ is a product distribution. 

\subsubsection*{Conditional Independence of Inputs Even After Messages}

We can prove that the distribution of the inputs to the players in $\GG_r$ \emph{remains} a product distribution (see \Cref{prop:product-dist}), even conditioned on the messages of the first round of $\prot_r$. 
This is a direct consequence of the rectangle property of protocols (and is a standard fact, which is proven here merely for completeness given we are conditioning on a subset of the input and not all of it).  

\begin{lemma}\label{lem:cond-ind}
	Let $\Prot_1$ denote all messages of $\prot_r$ in the first round. For any $i \in [q_r]$ and $j \in [p_r]$, 
	\[
	\distribution{H_{i, >j}, H_{>i,*} \mid H_{i, \leq j} , \Prot_1, H_{< i,*}} = \prod_{a=1}^{r} \distribution{H_{i, >j,a}, H_{>i,*,a} \mid H_{i, \leq j,a} , \Prot_1, H_{< i,*}} .
	\]
\end{lemma}
 We will first show some necessary independence of distributions before proving the lemma. 

\begin{claim}\label{clm:2-ind-dist}
	For any $a \in [r], j \in [p_r]$, $i \in [q_r]$, and any choice of $\Pi_1, H_{i, \leq j,a}, H_{< i, *}$,
	\[
	\rH_{i, \leq j,-a}, \rH_{i, >j,< a}, \rH_{>i,*,< a} \perp \rH_{i, >j,a},\rH_{>i,*,a} \mid \Pi_1, H_{i, \leq j,a}, H_{< i, *}.
	\]  
\end{claim}
\begin{proof}
	We will show that the following mutual information term is zero, which immediately implies the claim by~\itfacts{info-zero}: 
	\begin{align*}
		&\mi{\rH_{i, \leq j,-a}, \rH_{i, >j,< a}, \rH_{>i,*,< a}}{\rH_{i, >j,a},\rH_{>i,*,a} \mid \rProt_1, \rH_{i, \leq j,a}, \rH_{< i,*}} \\ 
		&\hspace{1cm}\leq \mi{\rG_{r-1,-a}}{\rG_{r-1,a} \mid \rProt_1, \rH_{i, \leq j,a}, \rH_{< i,*}}  \tag{by the data processing inequality of~\itfacts{data-processing} and~\Cref{prop:input-create}}\\
		&\hspace{1cm}= \mi{\rG_{r-1,-a}}{\rG_{r-1,a} \mid \rProt_{1,\leq a},\rProt_{1,>a}, \rH_{i, \leq j,a}, \rH_{< i,*}} \tag{by splitting $\rProt_1 = \rProt_{1,\leq a},\rProt_{1,> a}$} \\
		&\hspace{1cm}\leq \mi{\rG_{r-1,-a}}{\rG_{r-1,a} \mid \rProt_{1,\leq a}, \rH_{i, \leq j,a}, \rH_{< i,*}} \tag{by \Cref{prop:info-decrease} as  $\rProt_{1,>a} \perp \rG_{r-1,a} \mid G_{r-1,-a} , \Prot_{1,\leq a}, H_{i, \leq j,a}, H_{< i, *}$, and  $\rProt_{1,>a}$ is now fixed} \\
		&\hspace{1cm}= \mi{\rG_{r-1,-a}}{\rG_{r-1,a} \mid \rProt_{1,< a},\rProt_{1,a}, \rH_{i, \leq j,a},  \rH_{< i, *}} \tag{by further splitting $\rProt_{1,\leq  a} = \rProt_{1,< a},\rProt_{1,a}$}  \\
		&\hspace{1cm}\leq \mi{\rG_{r-1,-a}}{\rG_{r-1,a} \mid \rProt_{1,< a}, \rH_{i, \leq j,a}, \rH_{< i, *}}  \tag{by \Cref{prop:info-decrease} as $\rProt_{1,a} \perp \rG_{r-1,-a} \mid G_{r-1,a}, \Prot_{1,<a}, H_{i, \leq j,a}, H_{< i, *}$, and $\rProt_{1,a}$ is now fixed} \\
		&\hspace{1cm}\leq \mi{\rG_{r-1,-a}}{\rG_{r-1,a} \mid \rH_{i, \leq j,a}, \rH_{< i, *}} \tag{by \Cref{prop:info-decrease} as $\rProt_{1,< a} \perp \rG_{r-1,a} \mid G_{r-1,-a}, H_{i, \leq j,a}, H_{< i, *}$; $\rProt_{1,<a}$ is now fixed} \\
		&\hspace{1cm}= 0. \tag{by \Cref{prop:product-dist} all $H_{i,j}$'s are independent and we can apply~\itfacts{info-zero}} \qedhere
	\end{align*}
\end{proof}

\Cref{lem:cond-ind} follows easily now.

\begin{proof}[Proof of \Cref{lem:cond-ind}]
	We have that,
	\begin{align*}
		&\distribution{H_{i, >j}, H_{>i,*} \mid H_{i, \leq j}, \Prot_1, H_{< i,*}} \\
		&\hspace{1cm}= \prod_{a=1}^r \distribution{H_{i, >j,a}, H_{>i,*,a} \mid H_{i, >j,<  a}, H_{>i,*,< a}, H_{i, \leq j}, \Prot_1, H_{< i,*}} \tag{by chain rule} \\
		&\hspace{1cm}= \prod_{a=1}^r \distribution{H_{i, >j,a}, H_{>i,*,a}  \mid H_{i, >j,< a}, H_{>i,*,< a}, H_{i, \leq j,-a}, H_{i, \leq j,a}, \Prot_1, H_{<i,*}} \tag{by splitting $H_{i, \leq j} = H_{i, \leq j,-a}, H_{i, \leq j,a}$} \\
		&\hspace{1cm}= \prod_{a=1}^r \distribution{H_{i, >j,a}, H_{>i,*,a} \mid  \Prot_1, H_{i, \leq j,a}, H_{<i,*}} \tag{by~\Cref{clm:2-ind-dist}}, 
	\end{align*}
	completing the proof.
\end{proof}

\subsubsection*{Eliminating the First Round of the Protocol $\prot_r$} 
We are now ready to proceed with designing our $(r-1)$-round $r$-party protocol $\prot_{r-1}$ for $\PP_{r-1}$ using $\prot_r$ (and the intermediate protocol $\prot'$ designed in the previous step). 
Recall the intermediate protocol $\sigma$ in the previous step. This protocol (and subsequent ones $\sigmae$ and $\prot'$) consists of two separate parts: (1) an \emph{embedding} part that created 
an entire graph $G_r$ from $\GG_{r}$ by placing the input $G_{r-1} \sim \GG_{r-1}$ as one of its special sub-instances; and, (2) a \emph{simulation} part that ran the protocol $\prot_r$ (either directly in $\sigma$ or indirectly in $\sigmae$ and $\prot'$) 
on this input. In the following, we will change the embedding part (thoroughly) but stick with the same simulation part except that we will only run those simulations from the second round onwards. Formally, 

\begin{ourbox}
	\textbf{An $(r-1)$-round $r$-party protocol $\prot_{r-1}$ for $\PP_{r-1}(G_{r-1},K)$}: 
	\begin{enumerate}
		\item Sample first-round messages $\Prot_{1,1},\ldots,\Prot_{1,r}$ of $\prot_{r}$ from $\GG_r$ and indices $\kstar \in [p_r]$ and $t \in [q_r]$ uniformly at random and independently using \underline{public randomness}\footnote{We emphasize
		that we are \emph{not} sampling the message $\Prot_{1,r+1}$ in this step (as it is determined by $t$).}. 
		\item Sample $H_{<t, *}, H_{t, <\kstar}$ from $\GG_r \mid \Prot_{1,1},\ldots,\Prot_{1,r},\rT=t$ using \underline{public randomness}.  
		\item For every $a \in [r]$, each player $\player{a}$ sets $H_{t,\kstar,a} = G_{r-1,a}$ and samples
		\[
		\Paren{G_{r,a} \sim \GG_r \mid \Prot_{1,1},\ldots,\Prot_{1,r}, \rT=t , H_{<t, *}, H_{t, <\kstar}, \rH_{t,\kstar,a}=G_{r-1,a}}
		\]
		using \underline{private randomness}. Let $G_r=(G_{r,1},G_{r,2},\ldots,G_{r,r+1})$ be the 
		input of all players where, using~\Cref{prop:input-create}, the input of the simulated player $\player{r+1}$ is fixed by $\rT=t$.
		\item The players run $\prot_r(G_r)$ \underline{from its second round} onwards, following the protocol $\prot'$ of~\Cref{clm:prot'}. 
		\item At the end, for any search sequence $K=(k_{r-1},\ldots,k_1)$, to solve $\PP_{r-1}(G_{r-1},K)$, return the answer of $\prot_r$ for $\PP_r(G_r, (\kstar, k_{r-1},k_{r-2},\ldots,1))$ (which is possible given
		that $\prot'$ generates the messages of $\prot_r(G_r)$ implicitly).  
	\end{enumerate}
\end{ourbox}

Similar to step one and distribution $\reald$, we also define the \textbf{distribution $\faked$} as the distribution of the graph $G_r$ obtained in the protocol $\prot_{r-1}$ when $G_{r-1} \sim \GG_{r-1}$. 
Unlike before, however, it is no longer the case that $\faked$ is actually the ``right'' distribution of the input that $\prot_r$ expects; for instance, we now
\emph{embedded} $H_{t,\kstar} = G_{r-1}$ \emph{without} conditioning on the messages $\Prot_{1,1},\ldots,\Prot_{1,r}$ in the protocol\footnote{This is inevitable: we are given $G_{r-1} \sim \GG_{r-1}$ as part of the input
and \emph{not} from $\GG_{r-1} \mid \Prot_1$ (and we have to sample $\Prot_1$ instead of spending a whole round computing it). So, even though clearly in the protocol $\prot_r$, $\Prot_1 \not\perp H_{t,\kstar}$ can happen, these two variables are always independent in $\faked$.}. In particular, we can write these distributions as the following (where the random variables on the RHS of each term is distributed according to $\GG_r$): 
\begin{alignat}{2}
	&\distribution{\Prot_1,\kstar, G_r \sim \reald} :=&&\paren{\rProt_1, \rT, \rK} \times \paren{\rH_{<\rT,*},\rH_{\rT,<\rK} \mid \rProt_1,\rT,\rK} \notag \\
	& &&\times \paren{\rH_{\rT,\rK} \mid \rProt_1,\rT,\rK,  \rH_{<\rT,*},\rH_{\rT,<\rK}} \times \paren{\rG_r \mid \rProt_1,\rT,\rK,  \rH_{<\rT,*},\rH_{\rT,<\rK}, \rH_{\rT,\rK}} \label{eq:reald} \\ 
	\notag \\
	&\distribution{\Prot_1,\kstar, G_r \sim \faked} :=&&\paren{\rProt_1, \rT, \rK} \times \paren{\rH_{<\rT,*},\rH_{\rT,<\rK} \mid \rProt_1,\rT,\rK} \notag \\ 
	& &&\times \paren{\rH_{\rT,\rK}} \times \prod_{a=1}^{r} \paren{\rG_{r,a} \mid \rProt_1,\rT,\rK,  \rH_{<\rT,*},\rH_{\rT,<\rK}, \rH_{\rT,\rK,a}}. \label{eq:faked}
\end{alignat}
In the above, with a slight abuse of notation, we wrote a triple $(\Prot,t,\kstar,G_r) \sim \reald$ (or $\sim \faked$) 
to denote the joint distribution of \emph{all} these variables when the input graph $G_r$ is sampled from $\reald$ versus when they are sampled from $\faked$. 

Nevertheless, we are going to prove that $\faked$ is not that different from $\reald$ either. In particular, we prove the following lemma. 
\begin{lemma}[``distributions induced by $\reald$ and $\faked$ are close'']\label{lem:reald-faked}
	\[
		\tvd{\distribution{\Prot_1,\kstar, G_r \sim \reald}}{\distribution{\Prot_1,\kstar, G_r \sim \faked}} \leq \sqrt{\frac{s}{2\,p_r \cdot q_r}}. 
	\]
\end{lemma}

We will bound this difference term by term in~\Cref{eq:reald} and~\Cref{eq:faked}. The first two terms are equal in both. 
The next, and the main, claim bounds the difference between the third terms. This is yet another application of direct-sum style arguments along the lines of~\Cref{clm:low-info-cost}, 
although fundamentally different as we need a much stronger guarantee (that only holds for the first round of the protocol $\prot_r$). We emphasize that the following claim is
talking about $\prot_r$ and \emph{not} $\prot_{r-1}$ (as the variables in RHS of~\Cref{eq:reald} and~\Cref{eq:faked} are distributed according to $\prot_r$). 

\begin{claim}\label{clm:first-message}
	\[
	\Exp_{\rProt_1,\,\rT,\,\rK,\, \rH_{<t,*}, \, \rH_{t,<\kstar}} \tvd{(\rH_{t,\kstar} \mid \Prot_1 , H_{<t, *}, H_{t, <\kstar})}{\rH_{t,\kstar}} \leq \sqrt{\frac{s}{2\, p_r \cdot q_r}}. 
	\]
\end{claim}

\begin{proof}
	We follow the standard plan of bounding the KL-divergence of the above distributions and applying Pinsker's inequality (\Cref{fact:pinskers}) at the end. In this lemma, 
	\emph{all} variables are with respect to the protocol $\prot_r$ on the distribution $\GG_r$. The only exception is that of $\rK$ which is chosen uniformly from $[p_r]$ (and is not defined explicitly 
	in $\GG_r$ beforehand). 
	
	By~\Cref{prop:product-dist}, the first $r$ players do not have any information about random variable $\rT$, and thus their messages in the first round cannot reveal too much about \emph{all} of the special sub-instances, and in particular 
	$H_{t,\kstar}$. The last player knows $\rT$ but has no access to \emph{any} of the sub-instances, so should not be able to reveal anything about the special sub-instances. Overall, the first message 
	cannot change the distribution of $\rH_{t,\kstar}$, for a random chosen $t \in [q_r]$ and $\kstar \in [p_r]$ by much. 
	Our proof formalizes this idea to get a bound on the KL-Divergence. In the following, we denote 
	\[
		\Prot_{1,\leq r} := (\Prot_{1,1},\Prot_{1,2},\ldots,\Prot_{1,r}),
	\]
	i.e., the first-round messages of all players except for the last one. 
	
	We claim that, for every $i \in [q_r]$ and $j \in [p_r]$,  
	\begin{align}
		\rH_{i,j} \perp \rT=i, \rK=j \mid \Prot_1, \rH_{<i, *}, \rH_{i, <j} \qquad \text{and} \qquad  \rH_{i,j} \perp  \rH_{<i, *}, \rH_{i, <j} , \rT=i, \rK=j, \label{eq:cond-2}
	\end{align}
	because in both cases, sub-instances are chosen independently in $\GG_r$ from each other, as well as the special UPC selected by $\rT$; the choice of $\rK$ is also independent of the entire graph. 
	
	We can now use this and have, 
	\begin{align*}
		&\Exp_{\rProt_1,\,\rT,\,\rK,\,\rH_{<t,*}, \, \rH_{t,<\kstar}} \kl{\rH_{t,\kstar} \mid \Prot_1 , H_{<t, *}, H_{t, <\kstar}}{\rH_{t,\kstar}} \\
		&\hspace{1cm}= \Exp_{\rProt_1,\,\rT,\,\rK,\,\rH_{<t,*}, \, \rH_{t,<\kstar}} \kl{\rH_{t,\kstar} \mid \Prot_1 , H_{<t, *}, H_{t, <\kstar},\rT=t,\rK=\kstar}{\rH_{t,\kstar} \mid H_{<t, *}, H_{t, <\kstar}, \rT=t,\rK=\kstar} 
		\tag{the distributions are the same by part one and part two of~\Cref{eq:cond-2}, respectively} \\
		&\hspace{1cm}= \Exp_{\rProt_{1,\leq r},\,\rT,\,\rK,\,\rH_{<t,*}, \, \rH_{t,<\kstar}} \kl{\rH_{t,\kstar} \mid \Prot_{1,\leq r} , H_{t, <\kstar}, \rT=t, \rK = \kstar}{\rH_{t,\kstar} \mid H_{<t, *}, H_{t, <\kstar}, \rT=t, \rK=k} 
		\tag{because $\Prot_{1,r+1}$ is fixed by $t$ and $\Prot_{1,\leq r}$ because it fixes the input of $\player{r+1}$ by~\Cref{prop:input-create}}\\
		&\hspace{1cm}= \mi{\rProt_{1,\leq r}}{\rH_{\rT,\rK} \mid \rH_{<\rT, *}, \rH_{\rT, <\rK}, \rT,\rK}. \tag{by~\Cref{fact:kl-info}}
	\end{align*}
	In words, the change in the KL-divergence of the distribution of the sub-instance $H_{t,\kstar}$ from its original distribution by the messages of first round is bounded by the information revealed by $\Prot_{1,\leq r}$ 
	about this sub-instance. 
	
	We now bound this information term as follows: 
	\begin{align*}
		\mi{\rProt_{1,\leq r}}{\rH_{\rT,\rK} \mid \rH_{<\rT,*}, \rH_{\rT,<\rK}, \rT,\rK} &= \Exp_{\rT,\rK} \Bracket{\mi{\rProt_{1,\leq r}}{\rH_{t,\kstar} \mid \rH_{<t, *}, \rH_{t, <\kstar}, \rT=t,\rK=\kstar}} \tag{by the definition of conditional mutual information} \\
		&= \Exp_{\rT,\rK} \Bracket{\mi{\rProt_{1,\leq r}}{\rH_{t,\kstar} \mid \rH_{<t, *}, \rH_{t, <\kstar}}} 
		\tag{for the same reason as~\Cref{eq:cond-2} since $(\rProt_{1,\leq r},\rH_{t,\kstar},\rH_{<t, *}, \rH_{t, <\kstar})$ are all independent of $\rT,\rK$} \\
		&= \frac{1}{q_r} \cdot \frac{1}{p_r} \cdot \sum_{t=1}^{q_r} \sum_{\kstar=1}^{p_r} \mi{\rProt_{1,\leq r}}{\rH_{t,\kstar} \mid \rH_{<t, *}, \rH_{t, <\kstar}} \tag{as the distribution of both $\rT,\rK$ is uniform and independent of each other} \\
		&= \frac{1}{q_r \cdot p_r} \cdot \mi{\rProt_{1,\leq r}}{\set{\rH_{i,j}}_{i \in [q_r], j\in [p_r]}}\tag{by the chain rule of mutual information in~\itfacts{chain-rule}}  \\
		&\leq \frac{1}{q_r \cdot p_r} \cdot \mi{\rProt}{\rG} \tag{by the data processing inequality of~\itfacts{data-processing} and~\Cref{prop:input-create} on inputs of players} \\
		&= \frac{1}{q_r \cdot p_r} \cdot \ic{\prot_{r}}{\GG_r} \tag{by the definition of external information (\Cref{def:ext-info})} \\
		&\leq \frac{1}{q_r \cdot p_r} \cdot \cc{\prot_r} \tag{by~\Cref{prop:ic-cc}}
	\end{align*}
	which is $s/(q_r \cdot p_r)$ as $s$ is the total communication cost of $\prot_r$. We thus have, 
	\begin{align}
	\Exp_{\rProt_1,\,\rT,\,\rK,\,\rH_{t,<\kstar},\, \rH_{<t,*}} \kl{\rH_{t,\kstar} \mid \Prot_1 , H_{t,<\kstar}, H_{<t,*}}{\rH_{t,\kstar}} \leq  \frac{s}{q_r \cdot p_r}. \label{eq:kl-bound-2} 
	\end{align}
	We can finish the proof as follows: 
\begin{align*}
&\Exp_{\rProt_1,\,\rT,\,\rK,\,\rH_{t,<\kstar},\, \rH_{<t,*}} \tvd{(\rH_{t,\kstar} \mid \Prot_1 , H_{t,<\kstar}, H_{<t,*})}{\rH_{t,\kstar}} \\
&\hspace{1cm}\leq \Exp_{\rProt_1,\,\rT,\,\rK,\,\rH_{t,<\kstar},\, \rH_{<t,*}}\Bracket{\sqrt{1/2 \cdot \kl{\rH_{t,\kstar} \mid \Prot_1 , H_{t,<\kstar}, H_{<t,*}}{\rH_{t,\kstar}}}} 
	\tag{by Pinsker's inequality of~\Cref{fact:pinskers}} \\
&\hspace{1cm}\leq \sqrt{1/2 \cdot \Exp_{\rProt_1,\,\rT,\,\rK,\,\rH_{t,<\kstar},\, \rH_{<t,*}}\Bracket{\kl{\rH_{t,\kstar} \mid \Prot_1 , H_{t,<\kstar}, H_{<t,*}}{\rH_{t,\kstar}}}} \tag{by Jensen's inequality as $\sqrt{\cdot}$ is concave} \\
&\hspace{1cm}\leq \sqrt{\frac{s}{2q_r \cdot p_r}}, \tag{by~\Cref{eq:kl-bound-2}} 
\end{align*}
concluding the proof. 
\end{proof}

In the following claim, we bound the fourth terms of~\Cref{eq:reald} and~\Cref{eq:faked} by showing those distributions are actually equivalent. This is a direct corollary 
of the conditional independences we established earlier in~\Cref{lem:cond-ind}. The statement of the following claim is written in a rather indirect way, by stating the distance is zero instead of simply claiming the two distributions are equivalent, 
to make its application in the later part of the proof completely transparent.

\begin{claim}\label{clm:cor-cond-ind}
	\begin{align*}
		&\hspace{5.7cm} \Exp_{\rProt_1,\rT,\rK,  \rH_{<\rT,*},\rH_{\rT,<\rK}, \rH_{\rT,\rK}} \\
		&\tvd{\paren{\rG_r \mid \Prot_1,t,\kstar, H_{<t,*},H_{t,<\kstar}, H_{t,\kstar}}}{\prod_{a=1}^{r} \paren{\rG_{r,a} \mid \Prot_1, t, \kstar, H_{<t,*}, H_{t,<\kstar}, H_{t,\kstar,a}}} = 0.
	\end{align*}
\end{claim}
\begin{proof}
	The distribution of $\rG_r$ conditioned on $\rT=t, H_{<t,*},H_{t,<\kstar}, H_{t,\kstar}$ is a function of sub-instances 
	\[
	\set{\rH_{t,j}}_{j > \kstar} \qquad \text{and} \qquad \set{\rH_{i,j}}_{i > t, j \in [p_r]};
	\]
	this follows immediately from~\Cref{prop:input-create}. Moreover, these distributions are independent of $\rT=t$ even conditioned on $\rProt_1$, since $\rT=t$ fixes
	$\rProt_{1,r+1}$ (again by~\Cref{prop:input-create}) but the rest of $\rProt_{1,\leq r}$ is independent of $\rT=t$ -- however, the sub-instances are only a function of $\rProt_{1,\leq r}$
	and thus remain independent of $\rT=t$. Finally, $\rK = \kstar$ is independent of the entire graph; hence, we have, 
	\begin{align*}
		&\paren{\rG_r \mid \Prot_1,t,\kstar, H_{<t,*},H_{t,<\kstar}, H_{t,\kstar}} = \paren{\rH_{t,>\kstar},\rH_{>t,*} \mid \Prot_1, H_{<t,*}, H_{t,<\kstar}, H_{t,\kstar}} \\
		&\prod_{a=1}^{r} \paren{\rG_{r,a} \mid \Prot_1, t, \kstar, H_{<t,*}, H_{t,<\kstar}, H_{t,\kstar,a}} = \prod_{a=1}^{r} \paren{\rH_{t,>\kstar,a},\rH_{>t,*,a} \mid \Prot_1, H_{<t,*}, H_{t,<\kstar}, H_{t,\kstar,a}}. 
	\end{align*}
	The two RHS are now equal by~\Cref{lem:cond-ind}, proving the claim. 
\end{proof}

\Cref{lem:reald-faked} now follows from~\Cref{clm:first-message} and~\Cref{clm:cor-cond-ind} and the chain rule (upper bound) for total variation distance (\Cref{fact:tvd-chain-rule}). 

We can now conclude the proof of~\Cref{lem:round-elim-final}, by showing that the protocol $\prot_{r-1}$ can solve $\PP_{r-1}$ with the desired probability.  

\begin{proof}[Proof of~\Cref{lem:round-elim-final}]
	Suppose we sample a graph $G_{r-1} \sim \GG_{r-1}$ and run protocol $\prot'$ of step one in~\Cref{sec:step-one} on this graph by creating the messages and instance $(\Prot_1,G_r) \sim \reald$; then, by~\Cref{clm:prot'}, 
	\begin{align}
		\Prob_{\substack{G_{r-1} \sim \GG_{r-1} \\ \\ (\Prot_1,\kstar,G_r) \sim \reald}}\Paren{\text{$\prot'$ succeeds in solving $\PP_{r-1}$}} \geq \delta - \eps, \label{eq:wishful1}
	\end{align}
	for the parameters $\delta$ and $\eps$ in the lemma's statement. 
	
	On the other hand, what we are actually doing in $\prot_{r-1}$ is to sample $G_{r-1} \sim \GG_{r-1}$ but then run $\prot'$ of step one on this graph by creating messages and instance $(\Prot_1,G_r) \sim \faked$ instead. 
	Thus, 
	\begin{align}
		\Prob_{G_{r-1} \sim \GG_{r-1}}\Paren{\text{$\prot_{r-1}$ succeeds in solving $\PP_{r-1}$}} = \Prob_{\substack{G_{r-1} \sim \GG_{r-1} \\ \\ (\Prot_1,\kstar,G_r) \sim \faked}}\Paren{\text{$\prot'$ succeeds in solving $\PP_{r-1}$}}.  \label{eq:actual1}
	\end{align}
	Nevertheless, the RHS of this equation and the LHS of the above one are quite close to each other by~\Cref{lem:reald-faked} so we should be able to extend the first equation to the second one as well. Formally, 
	\begin{align*}
		&\Prob_{\substack{G_{r-1} \sim \GG_{r-1} \\ \\ (\Prot_1,\kstar,G_r) \sim \faked}}\Paren{\text{$\prot'$ succeeds in solving $\PP_{r-1}$}} \\
		&\hspace{1cm}\geq \hspace{-0.25cm}\Prob_{\substack{G_{r-1} \sim \GG_{r-1} \\ (\Prot_1,\kstar,G_r) \sim \reald}} \hspace{-0.25cm} \Paren{\text{$\prot'$ succeeds in solving $\PP_{r-1}$}}-\tvd{{(\Prot_1,\kstar, G_r) \sim \reald}}{{(\Prot_1,\kstar, G_r) \sim \faked}} \tag{by~\Cref{fact:tvd-small}} \\
		&\hspace{1cm}\geq \delta-\eps - \sqrt{\frac{s}{2\,p_r \cdot q_r}} \tag{by~\Cref{eq:wishful1} for the first term and~\Cref{lem:reald-faked} for the second one}. 
	\end{align*}
	Plugging in this bound in~\Cref{eq:actual1} implies that $\prot_{r-1}$ solves $\PP_{r-1}$ with the desired probability. 
	
	Also, again, by~\Cref{clm:prot'}, we have that $\prot_{r-1}$ has communication cost 
	\[
		\cc{\prot_{r-1}} \leq \frac{\raoconst}{\eps} \cdot \Paren{\frac{s}{p_r} + r^2}. 
	\]
	Finally, we can fix the randomness of $\prot_{r-1}$ by an averaging argument to obtain a deterministic algorithm with the same performance. We now have the desired deterministic $(r-1)$-round $r$-party protocol for $\PP_{r-1}$, 
	concluding the proof. 
\end{proof}

\begin{Remark}\label{rem:faked-ic}
{The reason we needed the message compression arguments and the protocols $\sigmae$ and $\prot'$ of step one, instead of working
with the protocol $\sigma$ right away is the \emph{very last step} of the proof above. Had we directly used $\sigma$, then, in the last step, we should have bounded the information cost 
of $\sigmae$ on the \emph{new} distribution $\faked$ instead of $\reald$ -- while, by~\Cref{lem:reald-faked}, these two distributions are statistically close, this does \emph{not} imply
that the information cost of $\sigmae$ on $\faked$ will be small also, which stops us from applying the induction hypothesis.} 
\end{Remark}

\subsection{Proof of~\Cref{lem:communication-lb}}\label{sec:proof-communication-lb}

We now use~\Cref{lem:round-elim-final} to complete the proof of~\Cref{lem:communication-lb} (restated below). 
At this stage, the main arguments have been made and the remainder of the proof, for the most part, is a tedious calculation based on the parameters set in~\Cref{eq:parameter-n-b} and~\Cref{eq:parameter-p-q}. 

\begin{lemma*}[Restatement of~\Cref{lem:communication-lb}]
	For any $r \geq 1$, any $r$-round protocol $\prot$ that given $G \sim \GG_r(n)$ (for $n,r$ satisfying~\Cref{eq:n-r-relation}), can solve $\PP_r$ on input graph $G$ with
	 probability of success at least 
	 \[
	 2^{-p_0} \cdot \paren{1 + \frac{r}{20 \cdot (r+1)}}
	 \]
	 has communication cost
	\[
		\cc{\prot} \geq s_r(n) := \frac1{{n_0}^2 \cdot (2^{p_0} \cdot 40 \cdot \raoconst)^r \cdot ((r+1)!)^2 \cdot e^{3r\ln^{5/6}(n)}} \cdot  \paren{{(n)}\uparrow \paren{1+\frac{1}{2^{r}-1}}}.
	\]
\end{lemma*}

\renewcommand{\probconst}{2^{p_0}}

\begin{proof}
	Given that there is a distribution over $G$ and we are working with communication cost (which is a worst-case measure), 
	it suffices to prove~\Cref{lem:communication-lb} for \emph{deterministic} protocols; the result for randomized protocols then 
	follows immediately by the easy direction of Yao's minimax principle (namely, by fixing the randomness of the protocol to its ``best'' choice using an averaging argument). 
	We prove the lemma using induction on $r$. 
	
	\paragraph{Base Case.}
	We prove the base case when $r = 1$. Assume towards a contradiction that there exists a protocol $\pi_1$ for $\GG_1(n)$ which uses communication  
	\[
	\cc{\prot_1}< s_1(n) = \frac{1}{n_0^2 \cdot (2^{p_0} \cdot 40 \raoconst) \cdot (2!)^2 \cdot e^{3 \ln^{5/6}(n)}} \cdot n^2,
	\]
	 and success probability at least $2^{-p_0} \cdot 41/40$. We will use \Cref{lem:round-elim-final} on $\prot_1$ to get a $0$-round protocol.
	
	We know that a $0$-round protocol cannot change the prior distribution of $\PP_0(G,\emptyset)$ for $G \sim \GG_0(n_0)$ 
	which is uniform over $\set{0,1}^{p_0}$. Thus, the probability that the $0$-round protocol can solve the predicate $\PP_0$ is $2^{-p_0}$ as it returns a fixed answer to  $\PP_0(G,\emptyset)$ independent of $G$.  
	
	We use \Cref{lem:round-elim-final} on $\prot_1$ with the following parameters (from \Cref{eq:parameter-p-q} for $r=1$):
	\begin{align*}
		\epsilon = \frac1{160 \cdot \probconst} \qquad p_1 = \frac{2b_1}{\exp(\eta_p \cdot (\ln (2b_1))^{3/4})} \qquad q_1 = \frac{2b_1}{\exp(\eta_q \cdot (\ln(2b_1))^{3/4})}.
	\end{align*}

We get a $0$-round protocol $\prot_0$, which cannot communicate any bits. We will just analyze the probability of success of $\prot_0$. We have, 
\begin{align*}
	\frac{s_1(n)}{p_1 \cdot q_1} &= \underbrace{\frac{1}{n_0^2 \cdot (\probconst \cdot 40 \raoconst) \cdot (2!)^2 \cdot e^{3 \ln^{5/6}(n)}} \cdot n^2}_{=s_1} \cdot \underbrace{\frac{\exp(\eta_p \cdot (\ln(2b_1))^{3/4})}{2b_1}}_{=1/p_1} \cdot \underbrace{\frac{\exp(\eta_q \cdot (\ln(2b_1))^{3/4})}{2b_1}}_{=1/q_1} \\
	&= \frac1{16\cdot 40 \cdot \probconst\cdot \raoconst} \cdot \frac{(2 \cdot b_1 n_0)^2}{4 \cdot b_1^2 \cdot n_0^2} \cdot \exp((\eta_p+\eta_q) \cdot \ln^{3/4}(2b_1)-3\ln^{5/6}(n)) \tag{by choice of $b_1=(n/2n_0)$ from \Cref{eq:define-b1}}  \\
	&\leq \frac1{16\cdot 40 \cdot \probconst\cdot \raoconst} \cdot \exp(-2 \ln^{5/6}(n)) \tag{as $\eta_p = O(1), \eta_q = O(1)$, $2b_1 < n$, and $O(\ln^{3/4}(n))< \ln^{5/6}(n)$ for large $n$} \\
	&\leq \paren{\frac1{160\cdot \probconst}}^2. \tag{taking a loose but sufficient upper bound}
\end{align*}
Thus, by~\Cref{lem:round-elim-final}, the probability of success of $\prot_0$ is at least,
\begin{align*}
	\frac{41}{40\cdot  \probconst} - \underbrace{\frac1{160 \cdot \probconst}}_{\eps} - \underbrace{\frac1{160 \cdot\probconst}}_{\sqrt{s_1(n)/(p_1 \cdot q_1)}} =\frac{81}{80 \cdot \probconst} > \frac1{2^{p_0}},
\end{align*}
which is a contradiction. This proves the induction base for $r=1$. 
	
	\paragraph{Induction Step.}
	We now prove the induction step wherein we assume the statement is true for $r-1$ and we prove it for $r > 1$. 
	Suppose towards a contradiction that there exists a protocol $\prot_r$ with probability of success
	\[
	\delta :=\frac1{\probconst}\paren{1 +  \frac{r}{20 \cdot (r+1)}}
	\]
	for solving $\PP_{r}$ as in the lemma statement, but with 
	communication cost $\cc{\prot_r} < s_r(n)$ for the parameter $s_r(n)$ in the lemma statement.
	Now, recall from~\Cref{eq:parameter-n-b} and \Cref{eq:parameter-p-q} that any graph sampled from distribution $\GG_r$ has the following parameters: 
	\begin{align}
	\begin{split}
			n&= 2 \cdot b_r\cdot n_{r-1} = 2 \cdot \paren{n_{r-1}}\uparrow\paren{\frac{2^r-1}{2^{r-1}-1}}, \\ 
			b_r &= \paren{n_{r-1}} \uparrow \paren{1+\frac{1}{2^{r-1}-1}} \\
			p_r&= \frac{b_r\cdot 2^r}{\exp(\eta_p \cdot (\ln (b_r \cdot 2^r))^{3/4})} = \frac{\paren{n_{r-1}} \uparrow \paren{1+\frac{1}{2^{r-1}-1}} \cdot 2^r}{\exp\paren{\eta_p \cdot (r\ln2+ \ln (b_r)}^{3/4})}, \\
			q_r&= \frac{b_r\cdot 2^r}{ \exp(\eta_q \cdot (\log (b_r \cdot 2^r))^{3/4})} = \frac{\paren{n_{r-1}} \uparrow \paren{1+\frac{1}{2^{r-1}-1}} \cdot 2^r}{\exp\paren{\eta_q \cdot (r\ln2+ \log (b_{r})}^{3/4})}. 
	\end{split}\label{eq:final-para}
	\end{align}

	We use~\Cref{lem:round-elim-final} on protocol $\prot_r$ with parameter 
	\[
	\eps := \frac1{\probconst \cdot 40 \cdot (r+1)^2}
	\]
	to get a protocol $\prot_{r-1}$ for $\PP_{r-1}$ over $\GG_{r-1}(n_{r-1})$.
	We want to analyze the communication cost and probability of success of this protocol. 
	We start by getting an upper bound on $s_r(n)/p_r$. 
\begin{align*}
	\frac{s_r(n)}{p_r} &= \underbrace{\frac{{{(n)}\uparrow \paren{1+\frac{1}{2^{r}-1}}}}{n_0^2 \cdot (\probconst \cdot 40\raoconst)^r \cdot ((r+1)!)^2 \cdot e^{3r\ln^{5/6}(n)}}}_{=s_r} \cdot \underbrace{\frac{\exp\paren{\eta_p \cdot (r\ln 2+ \ln (b_{r}))^{3/4}}}{b_r \cdot 2^r}}_{=1/p_r} \\
	&= \frac{2^{\frac{2^r}{2^r-1}} \cdot \paren{n_{r-1}}\uparrow\paren{\frac{2^r}{2^{r-1}-1}}}{n_0^2 \cdot (\probconst \cdot 40\raoconst)^r \cdot ((r+1)!)^2 \cdot 2^{r}} \cdot \frac{1}{b_r} \cdot \exp\paren{\eta_p \cdot (r\ln 2+ \ln (b_{r}))^{3/4} -3r \ln^{5/6}(n)} 
	\tag{by the choice of $n_{r-1}$ from \Cref{eq:final-para} and moving around the terms $2^r$ and $b_r$ in the denominator} \\
	&= \frac{2^{\frac{2^r}{2^r-1}} \cdot \paren{n_{r-1}}\uparrow\paren{\frac{2^r}{2^{r-1}-1} - \frac{2^{r-1}}{2^{r-1}-1}}}{n_0^2 \cdot (\probconst \cdot 40\raoconst)^r \cdot ((r+1)!)^2 \cdot 2^{r}} \cdot \exp\paren{\eta_p \cdot (r\ln 2+ \ln (b_{r}))^{3/4} -3r \ln^{5/6}(n)} 
	\tag{by the choice of $b_{r}$ from \Cref{eq:final-para} and writing it in the exponent of $n_{r-1}$} \\
	&= \frac{2^{\frac{2^r}{2^r-1}} \cdot \paren{n_{r-1}}\uparrow\paren{1+\frac{1}{2^{r-1}-1}}}{n_0^2 \cdot (\probconst \cdot 40\raoconst)^r \cdot ((r+1)!)^2 \cdot 2^{r}} \cdot \exp\paren{\eta_p \cdot (r\ln 2+ \ln (b_{r}))^{3/4} -3r \ln^{5/6}(n)} 
	\tag{by canceling the terms in the exponent of $n_{r-1}$} \\
	&\leq \frac{\paren{n_{r-1}}\uparrow\paren{1+\frac{1}{2^{r-1}-1}}}{n_0^2 \cdot (\probconst \cdot 40\raoconst)^r \cdot ((r+1)!)^2 \cdot 2^{r}} \cdot \exp\paren{O(  \ln^{3/4}(n)) -3r \ln^{5/6}(n)} 
	\tag{as $2^{\frac{2^r}{2^r-1}} = O(1)$, $b_r \leq n$ by \Cref{eq:final-para}, $\eta_p = \Theta(1)$ by~\Cref{prop:dup}, and $r \ll \ln{n}$ by~\Cref{eq:n-r-relation}} \\
	&\leq \frac{\paren{n_{r-1}}\uparrow\paren{1+\frac{1}{2^{r-1}-1}}}{n_0^2 \cdot (\probconst \cdot 40\raoconst)^r \cdot ((r+1)!)^2 \cdot 2^{r}} \cdot  \exp\paren{-(3r-1) \ln^{5/6}(n)} 
	 \tag{as $O(\log^{3/4}(n)) < \ln^{5/6}(n_r)$ for large $n$}.
\end{align*}
	By~\Cref{lem:round-elim-final}, the communication cost of $\prot_{r-1}$ is at most,
	\begin{align*}
	\cc{\prot_{r-1}} &\leq \frac{\raoconst}{\eps} \cdot \paren{\frac{s_r(n)}{p_r}  +r^2} \\ 
	&\leq	\frac{\raoconst}{ \eps} \cdot \paren{\frac{2 \cdot \paren{n_{r-1}}\uparrow\paren{1+\frac{1}{2^{r-1}-1}}}{n_0^2 \cdot (\probconst \cdot 40\raoconst)^r \cdot ((r+1)!)^2 \cdot 2^{r}} \cdot  \exp\paren{-(3r-1) \ln^{5/6}(n)}} 
	\tag{by our bound on $s_r(n)/p_r$ from above and since $r^2$ is (much) smaller than this bound by \Cref{eq:n-r-relation} } \\
	&= \raoconst \cdot   \underbrace{\probconst \cdot40 \cdot (r+1)^2}_{=1/\eps} \cdot \paren{ \frac{\paren{n_{r-1}}\uparrow\paren{1+\frac{1}{2^{r-1}-1}}}{n_0^2 \cdot (\probconst \cdot 40\raoconst)^r \cdot ((r+1)!)^2 \cdot 2^{r-1}} \cdot  \exp\paren{-(3r-2) \ln^{5/6}(n)}} 
	\tag{by our choice of $\eps$ and canceling the $2$-term in the nominator via $2^r$ in the denominator} \\
	&= \frac{\paren{n_{r-1}}\uparrow\paren{1+\frac{1}{2^{r-1}-1}}}{n_0^2 \cdot (\probconst \cdot 40\raoconst)^{r-1} \cdot (r!)^2 \cdot 2^{r-1}} \cdot  \exp\paren{-(3r-1) \ln^{5/6}(n)} 
	\tag{by canceling the extra terms via the ones in the denominator} \\
	&\leq \frac{\paren{n_{r-1}}\uparrow\paren{1+\frac{1}{2^{r-1}-1}}}{n_0^2 \cdot (\probconst \cdot 40\raoconst)^{r-1} \cdot (r!)^2 \cdot 2^{r-1} \cdot e^{3 \cdot (r-1) \cdot \ln^{5/6}(n_{r-1})}}
	\tag{by dropping an extra $e^{2\ln^{5/6}(n_{r-1})}$-term and using the fact $n_{r-1} < n$} \\
	&= s_{r-1}(n_{r-1}). 
	\end{align*}
	Thus, the communication cost of $\prot_{r-1}$ is less than the bounds of the induction hypothesis for $r-1$. 
	
	It remains to bound the probability of success of $\prot_{r-1}$. 
	By~\Cref{lem:round-elim-final}, 
	\begin{align*}
		\Pr\paren{\text{$\prot_{r-1}$ solves $\PP_{r-1}$}} &\geq  \delta - \eps - \sqrt{\frac{s_r(n)}{p_r \cdot q_r}} \\
		&\geq \frac{1}{\probconst}\paren{1 + \frac{r}{20(r+1)}} - \frac1{\probconst \cdot40(r+1)^2} - \sqrt{\frac{s_r(n)}{p_r \cdot q_r}} \tag{by the choice of $\eps$ and $\delta$ earlier in the proof}. 
	\end{align*}
	We thus only need to bound the last term to conclude the proof. 
	
	We can upper bound this last term as follows: 
\begin{align*}
\frac{s_r(n)}{p_r \cdot q_r} &\leq \frac{\paren{n_{r-1}}\uparrow\paren{1+\frac{1}{2^{r-1}-1}}}{n_0^2 \cdot (\probconst \cdot 40\raoconst)^r \cdot ((r+1)!)^2 \cdot 2^{r}} \cdot  \exp\paren{-(3r-1) \ln^{5/6}(n)} \cdot 
\underbrace{\frac{\exp\paren{\eta_q \cdot (r\ln2+ \ln b_{r}}^{3/4})}{b_{r} \cdot 2^r}}_{=1/q_r} 
\tag{by our upper bound on $s_r(n)/p_r$ established above and the choice of $q_r$ in~\Cref{eq:final-para}} \\
&= \frac{ \exp\paren{-(3r-1) \ln^{5/6}(n)+\eta_q \cdot (r\ln2+ \ln b_{r})^{3/4}}}{n_0^2 \cdot (\probconst \cdot 40\raoconst)^r \cdot ((r+1)!)^2 \cdot 2^{2r}} 
\tag{by the choice of $b_r = \paren{n_{r-1}}\uparrow\paren{1+\frac{1}{2^{r-1}-1}}$ in~\Cref{eq:final-para} and moving around the terms} \\
&\leq  \frac{ \exp\paren{-(3r-1) \ln^{5/6}(n)+O(\ln^{3/4}{n})}}{n_0^2 \cdot (\probconst \cdot 40\raoconst)^r \cdot ((r+1)!)^2 \cdot 2^{2r}} 
\tag{as $b_r \leq n$ by \Cref{eq:final-para}, $\eta_q = \Theta(1)$ by~\Cref{prop:dup}, and $r \ll \ln{n}$ by~\Cref{eq:n-r-relation}} \\
&\leq  \frac{ \exp(-(3r-2)\ln^{5/6}(n))}{n_0^2 \cdot (\probconst \cdot 40\raoconst)^r \cdot ((r+1)!)^2 \cdot 2^{2r}}  \cdot  \tag{as $ O(\ln^{3/4}(n))<\log^{5/6}(n)$ for large $n$ } \\ 
&\ll \paren{\frac{1}{\probconst \cdot40 \cdot (r+1)^2}}^2. \tag{using a crude upper bound sufficient for the analysis using $\ln^{5/6} \gg r$ by~\Cref{eq:n-r-relation}} 
\end{align*}
Thus, the probability that $\prot_{r-1}$ solves $\PP_{r-1}$ is at least 
\[
	\frac{1}{\probconst}\paren{1 + \frac{r}{20(r+1)}} - \frac1{ \probconst \cdot 40(r+1)^2} -  \frac1{ \probconst \cdot 40(r+1)^2} > \frac{1}{\probconst}\paren{1 + \frac{r-1}{20r}}. 
\]
But, now $\prot_{r-1}$ is a $(r-1)$-round protocol with communication cost $\cc{\prot_{r-1}} < s_{r-1}(n_{r-1})$ and probability of success strictly larger than $1/\probconst \cdot (1+ (r-1)/20r)$ which contradicts the induction hypothesis for $r-1$ 
(recall that by~\Cref{obs:eq-n-r-relation}, $n_{r-1}$ and $(r-1)$ also satisfy~\Cref{eq:n-r-relation} and thus we can indeed apply the induction hypothesis). 
This implies that our choice of the protocol $\prot_r$ 
was contradictory, completing the proof. 
\end{proof}

\subsection{Concluding the Proof of~\Cref{thm:mis-cc}} \label{sec:proof-mis-cc}

\Cref{thm:mis-cc} follows from combining the results in the previous section and the following argument. 

 \Cref{lem:reveal-search-bit} implies that $\MIS{n}{r+1}$ is as hard as solving $\PP_{r}(G,K)$ for $G \sim \GG_r$ and every valid search sequence $K$.
 \Cref{lem:communication-lb} then shows that, for the $r$-round protocol $\prot$ in the theorem statement, there is a lower bound of 
	\[
		\cc{\prot} \geq s_r(n) :=  \frac1{n_0^2 \cdot (\probconst \cdot 40\raoconst)^r \cdot ((r+1)!)^2 \cdot e^{3r\ln^{5/6}(n)}} \cdot  \paren{(n)\uparrow \paren{1+\frac{1}{2^{r}-1}}}
	\]
	for solving $\PP_r$ on an $n$-vertex input graph $G \sim \GG_r(n)$ with probability of success at least 
	\[
	\frac1{\probconst}\paren{1+\frac{r}{20(r+1)}} \leq \frac{21}{20\cdot \probconst} < \frac{21}{20 \cdot 2^{(n_0/2)-1}}
	\]
	 for any $r \geq 1$. We choose $n_0$ to be a large enough constant to handle any constant probability of success strictly larger than zero. 
	Thus, for $\MIS{n}{r+1}$, we obtain a lower bound of 
	\[
		\Omega\Paren{\frac{1}{2^{\Theta(r \ln^{5/6}(n))}} \cdot n^{1+1/(2^{r}-1)}}
	\]
	bits on the communication cost of the problem for any $r$-round protocol as long as $n \geq (n_0)^{2^r-1}$. 

	This concludes the proof of~\Cref{thm:mis-cc}. \qed

%% file: info.tex
\section{Background on Information Theory}\label{app:info}

We now briefly introduce some definitions and facts from information theory that are needed in this thesis. We refer the interested reader to the text by Cover and Thomas~\cite{CoverT06} for an excellent introduction to this field, 
and the proofs of the statements used in this Appendix. 

For a random variable $\rA$, we use $\supp{\rA}$ to denote the support of $\rA$ and $\distribution{\rA}$ to denote its distribution. 
When it is clear from the context, we may abuse the notation and use $\rA$ directly instead of $\distribution{\rA}$, for example, write 
$A \sim \rA$ to mean $A \sim \distribution{\rA}$, i.e., $A$ is sampled from the distribution of random variable $\rA$. 

\begin{itemize}
\item We denote the \emph{Shannon Entropy} of a random variable $\rA$ by
$\en{\rA}$, which is defined as: 
\begin{align}
	\en{\rA} := \sum_{A \in \supp{\rA}} \Pr\paren{\rA = A} \cdot \log{\paren{1/\Pr\paren{\rA = A}}} \label{eq:entropy}
\end{align} 

\item The \emph{conditional entropy} of $\rA$ conditioned on $\rB$ is denoted by $\en{\rA \mid \rB}$ and defined as:
\begin{align}
\en{\rA \mid \rB} := \Ex_{B \sim \rB} \bracket{\en{\rA \mid \rB = B}}, \label{eq:cond-entropy}
\end{align}
where 
$\en{\rA \mid \rB = B}$ is defined in a standard way by using the distribution of $\rA$ conditioned on the event $\rB = B$ in Eq~(\ref{eq:entropy}).

\item The \emph{mutual information} of two random variables $\rA$ and $\rB$ is denoted by
$\mi{\rA}{\rB}$ and is defined:
\begin{align}
\mi{\rA}{\rB} := \en{A} - \en{A \mid  B} = \en{B} - \en{B \mid  A}. \label{eq:mi}
\end{align}

\item The \emph{conditional mutual information} is defined as $\mi{\rA}{\rB \mid \rC}:= \en{\rA \mid \rC} - \en{\rA \mid \rB,\rC}$.
\end{itemize}

\subsection{Useful Properties of Entropy and Mutual Information}\label{sec:prop-en-mi}

We shall use the following basic properties of entropy and mutual information throughout. 
Proofs of these properties mostly follow from convexity of the entropy function
and Jensen's inequality and can be found in~\cite[Chapter~2]{CoverT06}. 

\begin{fact}\label{fact:it-facts}
  Let $\rA$, $\rB$, $\rC$, and $\rD$ be four (possibly correlated) random variables.
   \begin{enumerate}
  \item \label{part:uniform} $0 \leq \en{\rA} \leq \log{\card{\supp{\rA}}}$. The right equality holds
    iff $\distribution{\rA}$ is uniform.
  \item \label{part:info-zero} $\mi{\rA}{\rB \mid \rC} \geq 0$. The equality holds iff $\rA$ and
    $\rB$ are \emph{independent} conditioned on $\rC$.
  \item \label{part:cond-reduce} \emph{Conditioning on a random variable reduces entropy}:
    $\en{\rA \mid \rB,\rC} \leq \en{\rA \mid  \rB}$.  The equality holds iff $\rA \perp \rC \mid \rB$.
  \item \label{part:chain-rule} \emph{Chain rule for mutual information}: $\mi{\rA,\rB}{\rC \mid \rD} = \mi{\rA}{\rC \mid \rD} + \mi{\rB}{\rC \mid  \rA,\rD}$.
  \item \label{part:data-processing} \emph{Data processing inequality}: for a function $f(\rA)$ of $\rA$, $\mi{f(\rA)}{\rB \mid \rC} \leq \mi{\rA}{\rB \mid \rC}$. 
   \end{enumerate}
\end{fact}


\noindent
We also use the following two standard propositions on effect of conditioning on mutual information.

\begin{proposition}\label{prop:info-increase}
  For random variables $\rA, \rB, \rC, \rD$, if $\rA \perp \rD \mid \rC$, then, 
  \[\mi{\rA}{\rB \mid \rC} \leq \mi{\rA}{\rB \mid  \rC,  \rD}.\]
\end{proposition}
 \begin{proof}
  Since $\rA$ and $\rD$ are independent conditioned on $\rC$, by
  \itfacts{cond-reduce}, $\HH(\rA \mid  \rC) = \HH(\rA \mid \rC, \rD)$ and $\HH(\rA \mid  \rC, \rB) \ge \HH(\rA \mid  \rC, \rB, \rD)$.  We have,
	 \begin{align*}
	  \mi{\rA}{\rB \mid  \rC} &= \HH(\rA \mid \rC) - \HH(\rA \mid \rC, \rB) = \HH(\rA \mid  \rC, \rD) - \HH(\rA \mid \rC, \rB) \\
	  &\leq \HH(\rA \mid \rC, \rD) - \HH(\rA \mid \rC, \rB, \rD) = \mi{\rA}{\rB \mid \rC, \rD}. \qed
	\end{align*}
	
\end{proof}

\begin{proposition}\label{prop:info-decrease}
  For random variables $\rA, \rB, \rC,\rD$, if $ \rA \perp \rD \mid \rB,\rC$, then, 
  \[\mi{\rA}{\rB \mid \rC} \geq \mi{\rA}{\rB \mid \rC, \rD}.\]
\end{proposition}
 \begin{proof}
 Since $\rA \perp \rD \mid \rB,\rC$, by \itfacts{cond-reduce}, $\HH(\rA \mid \rB,\rC) = \HH(\rA \mid \rB,\rC,\rD)$. Moreover, since conditioning can only reduce the entropy (again by \itfacts{cond-reduce}), 
  \begin{align*}
 	\mi{\rA}{\rB \mid  \rC} &= \HH(\rA \mid \rC) - \HH(\rA \mid \rB,\rC) \geq \HH(\rA \mid \rD,\rC) - \HH(\rA \mid \rB,\rC) \\
	&= \HH(\rA \mid \rD,\rC) - \HH(\rA \mid \rB,\rC,\rD) = \mi{\rA}{\rB \mid \rC,\rD}. \qed
 \end{align*}

\end{proof}

\subsection{Measures of Distance Between Distributions}\label{sec:prob-distance}

We use two main measures of distance (or divergence) between distributions, namely the \emph{Kullback-Leibler divergence} (KL-divergence) and the \emph{total variation distance}. 

\paragraph{KL-divergence.} For two distributions $\mu$ and $\nu$ over the same probability space, the \textbf{Kullback-Leibler (KL) divergence} between $\mu$ and $\nu$ is denoted by $\kl{\mu}{\nu}$ and defined as: 
\begin{align}
\kl{\mu}{\nu}:= \Ex_{a \sim \mu}\Bracket{\log\frac{\mu(a)}{{\nu}(a)}}. \label{eq:kl}
\end{align}
We also have the following relation between mutual information and KL-divergence. 
\begin{fact}\label{fact:kl-info}
	For random variables $\rA,\rB,\rC$, 
	\[\mi{\rA}{\rB \mid \rC} = \Ex_{(B,C) \sim {(\rB,\rC)}}\Bracket{ \kl{\distribution{\rA \mid \rB=B,\rC=C}}{\distribution{\rA \mid \rC=C}}}.\] 
\end{fact}


\paragraph{Total variation distance.} We denote the \textbf{total variation distance} between two distributions $\mu$ and $\nu$ on the same 
support $\Omega$ by $\tvd{\mu}{\nu}$, defined as: 
\begin{align}
\tvd{\mu}{\nu}:= \max_{\Omega' \subseteq \Omega} \paren{\mu(\Omega')-\nu(\Omega')} = \frac{1}{2} \cdot \sum_{x \in \Omega} \card{\mu(x) - \nu(x)}.  \label{eq:tvd}
\end{align}
\noindent
We use the following basic properties of total variation distance. 
\begin{fact}\label{fact:tvd-small}
	Suppose $\mu$ and $\nu$ are two distributions for $\event$, then, 
	$
	{\mu}(\event) \leq {\nu}(\event) + \tvd{\mu}{\nu}.
$
\end{fact}

We also have the following (chain-rule) bound on the total variation distance of joint variables.

\begin{fact}\label{fact:tvd-chain-rule}
	For any distributions $\mu$ and $\nu$ on $n$-tuples $(X_1,\ldots,X_n)$, 
	\[
		\tvd{\mu}{\nu} \leq \sum_{i=1}^{n} \Exp_{X_{<i} \sim \mu} \tvd{\mu(X_i \mid X_{<i})}{\nu(X_i \mid X_{<i})}. 
	\]
\end{fact}

The following Pinsker's inequality bounds the total variation distance between two distributions based on their KL-divergence, 

\begin{fact}[Pinsker's inequality]\label{fact:pinskers}
	For any distributions $\mu$ and $\nu$, 
	$
	\tvd{\mu}{\nu} \leq \sqrt{\frac{1}{2} \cdot \kl{\mu}{\nu}}.
	$ 
\end{fact}


%% file: appendix-dup.tex
\section{Extending \Cref{prop:dup} to all Graph Sizes}\label{app:prop-dup} 

Recall that in the proof of~\Cref{prop:dup}, the graph $G$ constructed only exists for all integers $\ell \geq 1$, which does not include all large enough integers $n$ as in the statement of the proposition. 

To fix this, we simply take the largest possible $G$ with fewer than $n$ vertices and pad each layer with extra vertices to have exactly $n$ vertices overall as required.
This padding does not change the parameters $p$ and $q$. As such, we show that for any choice of large enough number of vertices $n$ there exists a choice of $n'$ that is consistent with some integer $\ell$ such that $G$ exists and $n \geq n' \geq {n}/{2^{\Theta(\sqrt{\log n})}}$. This implies a negligible loss on the bounds of $p$ and $q$ when $G$ is padded.

To that end, in the construction of $G$, we have that $d^2 = \log \frac{n}{k+1} = d \cdot \log ((k+2) \cdot \ell)$ and thus
\begin{align}
	& \log (k\ell) \leq d_\ell \leq \log (3k\ell), \text{~and} \label{eq:dbounds1} \\
	& \log (k(\ell+1)) \leq d_{\ell + 1} \leq \log (3k(\ell + 1)) \leq \log (6k\ell) \label{eq:dbounds2}
\end{align}
where $d_\ell :=  \log ((k+2) \cdot \ell)$ refers to the choice of $d$ that depends only on integer $\ell$ (since $k$ is fixed).
Let $n$ be as in the statement of the proposition where $n = b \cdot (k+1)$ for some large enough integer $b$ since the required DUP graph is a $(k+1)$-layered graph.
Fix a choice of $\ell$ such that 
\begin{align}
	b' = ((k+2) \cdot \ell)^{d_\ell} \leq b \leq ((k+2) \cdot (\ell + 1))^{d_{\ell + 1}} \label{eq:bbounds}
\end{align}
where $b'$ also defines the integer $n' = b' \cdot (k+1)$.
Then, using the bounds in \Cref{eq:dbounds1}, \Cref{eq:dbounds2}, and \Cref{eq:bbounds}, we have that
\begin{align*}
	\frac{n}{n'} = \frac{b}{b'} &\leq \frac{((k+2) \cdot (\ell + 1))^{\log(6k\ell)}}{((k+2) \cdot \ell)^{\log(k\ell)}} \leq \paren{1 + \frac{1}{\ell}}^{\log (k\ell)} \cdot ((k+2) \cdot (\ell + 1))^{\log 6} \\
	& \leq \exp\paren{\frac{\log(k\ell)}{\ell}} \cdot (6k\ell)^3 \leq 2^{\Theta\paren{\log (k\ell)}}.
\end{align*}
Finally, since $(k\ell)^{\log(k\ell)} \leq \paren{(k+2)\cdot \ell}^{d_\ell} = b' \leq b \leq n$, we have that 
\[
\log(k\ell) = \sqrt{\log (k\ell)^{\log(k\ell)}} \leq \sqrt{\log n}.
\] 
This concludes the proof for all graph sizes.

%% file: main.bbl
\newcommand{\etalchar}[1]{$^{#1}$}
\begin{thebibliography}{BRWY13b}
\expandafter\ifx\csname url\endcsname\relax
  \def\url#1{\texttt{#1}}\fi
\expandafter\ifx\csname doi\endcsname\relax
  \def\doi#1{\burlalt{doi:#1}{http://dx.doi.org/#1}}\fi
\expandafter\ifx\csname urlprefix\endcsname\relax\def\urlprefix{URL }\fi
\expandafter\ifx\csname href\endcsname\relax
  \def\href#1#2{#2}\fi
\expandafter\ifx\csname burlalt\endcsname\relax
  \def\burlalt#1#2{\href{#2}{#1}}\fi

\bibitem[A22]{Assadi22}
S.~Assadi.
\newblock A two-pass (conditional) lower bound for semi-streaming maximum
  matching.
\newblock In J.~S. Naor and N.~Buchbinder, editors, {\em Proceedings of the
  2022 {ACM-SIAM} Symposium on Discrete Algorithms, {SODA} 2022, Virtual
  Conference / Alexandria, VA, USA, January 9 - 12, 2022}, pages 708--742.
  {SIAM}, 2022.

\bibitem[A23]{Assadi23}
S.~Assadi.
\newblock Recent advances in multi-pass graph streaming lower bounds.
\newblock {\em {SIGACT} News}, 54(3):48--75, 2023.

\bibitem[A24]{Assadi24}
S.~Assadi.
\newblock A simple (1-{$\eps$})-approximation semi-streaming algorithm for
  maximum (weighted) matching.
\newblock {\em CoRR}, abs/2307.02968. To appear in SOSA 2024, 2023.

\bibitem[AB16]{AbboudB16}
A.~Abboud and G.~Bodwin.
\newblock The 4/3 additive spanner exponent is tight.
\newblock In D.~Wichs and Y.~Mansour, editors, {\em Proceedings of the 48th
  Annual {ACM} {SIGACT} Symposium on Theory of Computing, {STOC} 2016,
  Cambridge, MA, USA, June 18-21, 2016}, pages 351--361. {ACM}, 2016.

\bibitem[Abl93]{Ablayev93}
F.~M. Ablayev.
\newblock Lower bounds for one-way probabilistic communication complexity.
\newblock In {\em Automata, Languages and Programming, 20nd International
  Colloquium, ICALP93, Lund, Sweden, July 5-9, 1993, Proceedings}, pages
  241--252, 1993.

\bibitem[ABR23]{AzarmehrB24}
A.~Azarmehr, S.~Behnezhad, and M.~Roghani.
\newblock Fully dynamic matching: (2-{$\sqrt{2}$})-approximation in polylog
  update time.
\newblock {\em CoRR}, abs/2307.08772. To appear in SODA 2024., 2023.

\bibitem[ACG{\etalchar{+}}15]{AhnCGMW15}
K.~J. Ahn, G.~Cormode, S.~Guha, A.~McGregor, and A.~Wirth.
\newblock Correlation clustering in data streams.
\newblock In {\em Proceedings of the 32nd International Conference on Machine
  Learning, {ICML} 2015, Lille, France, 6-11 July 2015}, pages 2237--2246,
  2015.

\bibitem[ACK19a]{AssadiCK19b}
S.~Assadi, Y.~Chen, and S.~Khanna.
\newblock Polynomial pass lower bounds for graph streaming algorithms.
\newblock In {\em Proceedings of the 51st Annual {ACM} {SIGACT} Symposium on
  Theory of Computing, {STOC} 2019, Phoenix, AZ, USA, June 23-26, 2019}, pages
  265--276, 2019.

\bibitem[ACK19b]{AssadiCK19a}
S.~Assadi, Y.~Chen, and S.~Khanna.
\newblock Sublinear algorithms for ({\(\Delta\)} + 1) vertex coloring.
\newblock In {\em Proceedings of the Thirtieth Annual {ACM-SIAM} Symposium on
  Discrete Algorithms, {SODA} 2019, San Diego, California, USA, January 6-9,
  2019}, pages 767--786, 2019.

\bibitem[ACN05]{AilonCN05}
N.~Ailon, M.~Charikar, and A.~Newman.
\newblock Aggregating inconsistent information: ranking and clustering.
\newblock In H.~N. Gabow and R.~Fagin, editors, {\em Proceedings of the 37th
  Annual {ACM} Symposium on Theory of Computing, Baltimore, MD, USA, May 22-24,
  2005}, pages 684--693. {ACM}, 2005.

\bibitem[AD21]{AssadiD21}
S.~Assadi and A.~Dudeja.
\newblock Ruling sets in random order and adversarial streams.
\newblock In S.~Gilbert, editor, {\em 35th International Symposium on
  Distributed Computing, {DISC} 2021, October 4-8, 2021, Freiburg, Germany
  (Virtual Conference)}, volume 209 of {\em LIPIcs}, pages 6:1--6:18. Schloss
  Dagstuhl - Leibniz-Zentrum f{\"{u}}r Informatik, 2021.

\bibitem[AG18]{AhnG18}
K.~J. Ahn and S.~Guha.
\newblock Access to data and number of iterations: Dual primal algorithms for
  maximum matching under resource constraints.
\newblock {\em {ACM} Trans. Parallel Comput.}, 4(4):17:1--17:40, 2018.

\bibitem[AGM12]{AhnGM12}
K.~J. Ahn, S.~Guha, and A.~McGregor.
\newblock Analyzing graph structure via linear measurements.
\newblock In {\em Proceedings of the Twenty-third Annual ACM-SIAM Symposium on
  Discrete Algorithms}, SODA '12, pages 459--467, 2012.

\bibitem[AJJ{\etalchar{+}}22]{AssadiJJST22}
S.~Assadi, A.~Jambulapati, Y.~Jin, A.~Sidford, and K.~Tian.
\newblock Semi-streaming bipartite matching in fewer passes and optimal space.
\newblock In J.~S. Naor and N.~Buchbinder, editors, {\em Proceedings of the
  2022 {ACM-SIAM} Symposium on Discrete Algorithms, {SODA} 2022, Virtual
  Conference / Alexandria, VA, USA, January 9 - 12, 2022}, pages 627--669.
  {SIAM}, 2022.

\bibitem[AKL17]{AssadiKL17}
S.~Assadi, S.~Khanna, and Y.~Li.
\newblock On estimating maximum matching size in graph streams.
\newblock In {\em Proceedings of the Twenty-Eighth Annual {ACM-SIAM} Symposium
  on Discrete Algorithms, {SODA} 2017, Barcelona, Spain, Hotel Porta Fira,
  January 16-19}, pages 1723--1742, 2017.

\bibitem[AKLY16]{AssadiKLY16}
S.~Assadi, S.~Khanna, Y.~Li, and G.~Yaroslavtsev.
\newblock Maximum matchings in dynamic graph streams and the simultaneous
  communication model.
\newblock In {\em Proceedings of the Twenty-Seventh Annual {ACM-SIAM} Symposium
  on Discrete Algorithms, {SODA} 2016, Arlington, VA, USA, January 10-12,
  2016}, pages 1345--1364, 2016.

\bibitem[AKO20]{AssadiKO20}
S.~Assadi, G.~Kol, and R.~Oshman.
\newblock Lower bounds for distributed sketching of maximal matchings and
  maximal independent sets.
\newblock In Y.~Emek and C.~Cachin, editors, {\em {PODC} '20: {ACM} Symposium
  on Principles of Distributed Computing, Virtual Event, Italy, August 3-7,
  2020}, pages 79--88. {ACM}, 2020.

\bibitem[AKZ22]{AssadiKZ22}
S.~Assadi, G.~Kol, and Z.~Zhang.
\newblock Rounds vs communication tradeoffs for maximal independent sets.
\newblock In {\em 63rd {IEEE} Annual Symposium on Foundations of Computer
  Science, {FOCS} 2022, Denver, CO, USA, October 31 - November 3, 2022}, pages
  1193--1204. {IEEE}, 2022.

\bibitem[ALT21]{AssadiLT21}
S.~Assadi, S.~C. Liu, and R.~E. Tarjan.
\newblock An auction algorithm for bipartite matching in streaming and
  massively parallel computation models.
\newblock In H.~V. Le and V.~King, editors, {\em 4th Symposium on Simplicity in
  Algorithms, {SOSA} 2021, Virtual Conference, January 11-12, 2021}, pages
  165--171. {SIAM}, 2021.

\bibitem[AMS96]{AlonMS96}
N.~Alon, Y.~Matias, and M.~Szegedy.
\newblock The space complexity of approximating the frequency moments.
\newblock In {\em STOC}, pages 20--29. ACM, 1996.

\bibitem[AMS12]{AlonMS12}
N.~Alon, A.~Moitra, and B.~Sudakov.
\newblock Nearly complete graphs decomposable into large induced matchings and
  their applications.
\newblock In {\em Proceedings of the 44th Symposium on Theory of Computing
  Conference, {STOC} 2012, New York, NY, USA, May 19 - 22, 2012}, pages
  1079--1090, 2012.

\bibitem[ANRW15]{AlonNRW15}
N.~Alon, N.~Nisan, R.~Raz, and O.~Weinstein.
\newblock Welfare maximization with limited interaction.
\newblock In {\em {IEEE} 56th Annual Symposium on Foundations of Computer
  Science, {FOCS} 2015, Berkeley, CA, USA, 17-20 October, 2015}, pages
  1499--1512, 2015.

\bibitem[AOSS18]{AssadiOSS18}
S.~Assadi, K.~Onak, B.~Schieber, and S.~Solomon.
\newblock Fully dynamic maximal independent set with sublinear update time.
\newblock In I.~Diakonikolas, D.~Kempe, and M.~Henzinger, editors, {\em
  Proceedings of the 50th Annual {ACM} {SIGACT} Symposium on Theory of
  Computing, {STOC} 2018, Los Angeles, CA, USA, June 25-29, 2018}, pages
  815--826. {ACM}, 2018.

\bibitem[AOSS19]{AssadiOSS19}
S.~Assadi, K.~Onak, B.~Schieber, and S.~Solomon.
\newblock Fully dynamic maximal independent set with sublinear in n update
  time.
\newblock In T.~M. Chan, editor, {\em Proceedings of the Thirtieth Annual
  {ACM-SIAM} Symposium on Discrete Algorithms, {SODA} 2019, San Diego,
  California, USA, January 6-9, 2019}, pages 1919--1936. {SIAM}, 2019.

\bibitem[AR20]{AssadiR20}
S.~Assadi and R.~Raz.
\newblock Near-quadratic lower bounds for two-pass graph streaming algorithms.
\newblock In S.~Irani, editor, {\em 61st {IEEE} Annual Symposium on Foundations
  of Computer Science, {FOCS} 2020, Durham, NC, USA, November 16-19, 2020},
  pages 342--353. {IEEE}, 2020.

\bibitem[ARVX12]{AlonRVX12}
N.~Alon, R.~Rubinfeld, S.~Vardi, and N.~Xie.
\newblock Space-efficient local computation algorithms.
\newblock In Y.~Rabani, editor, {\em Proceedings of the Twenty-Third Annual
  {ACM-SIAM} Symposium on Discrete Algorithms, {SODA} 2012, Kyoto, Japan,
  January 17-19, 2012}, pages 1132--1139. {SIAM}, 2012.

\bibitem[AS19]{AssadiS19}
S.~Assadi and S.~Solomon.
\newblock When algorithms for maximal independent set and maximal matching run
  in sublinear time.
\newblock In C.~Baier, I.~Chatzigiannakis, P.~Flocchini, and S.~Leonardi,
  editors, {\em 46th International Colloquium on Automata, Languages, and
  Programming, {ICALP} 2019, July 9-12, 2019, Patras, Greece}, volume 132 of
  {\em LIPIcs}, pages 17:1--17:17. Schloss Dagstuhl - Leibniz-Zentrum f{\"{u}}r
  Informatik, 2019.

\bibitem[AS23]{AssadiS23}
S.~Assadi and J.~Sundaresan.
\newblock Hidden permutations to the rescue: Multi-pass semi-streaming lower
  bounds for approximate matchings.
\newblock In {\em {IEEE} Annual Symposium on Foundations of Computer Science,
  {FOCS} 2023, Santa Cruz, CA, USA}. {IEEE}, 2023.

\bibitem[AW22]{AssadiW22}
S.~Assadi and C.~Wang.
\newblock Sublinear time and space algorithms for correlation clustering via
  sparse-dense decompositions.
\newblock In M.~Braverman, editor, {\em 13th Innovations in Theoretical
  Computer Science Conference, {ITCS} 2022, January 31 - February 3, 2022,
  Berkeley, CA, {USA}}, volume 215 of {\em LIPIcs}, pages 10:1--10:20. Schloss
  Dagstuhl - Leibniz-Zentrum f{\"{u}}r Informatik, 2022.

\bibitem[BBCR10]{BarakBCR10}
B.~Barak, M.~Braverman, X.~Chen, and A.~Rao.
\newblock How to compress interactive communication.
\newblock In L.~J. Schulman, editor, {\em Proceedings of the 42nd {ACM}
  Symposium on Theory of Computing, {STOC} 2010, Cambridge, Massachusetts, USA,
  5-8 June 2010}, pages 67--76. {ACM}, 2010.

\bibitem[BBD{\etalchar{+}}19]{DBehnezhadBDFHKU19}
S.~Behnezhad, S.~Brandt, M.~Derakhshan, M.~Fischer, M.~Hajiaghayi, R.~M. Karp,
  and J.~Uitto.
\newblock Massively parallel computation of matching and {MIS} in sparse
  graphs.
\newblock In P.~Robinson and F.~Ellen, editors, {\em Proceedings of the 2019
  {ACM} Symposium on Principles of Distributed Computing, {PODC} 2019, Toronto,
  ON, Canada, July 29 - August 2, 2019}, pages 481--490. {ACM}, 2019.

\bibitem[BBH{\etalchar{+}}19]{BalliuBHORS19}
A.~Balliu, S.~Brandt, J.~Hirvonen, D.~Olivetti, M.~Rabie, and J.~Suomela.
\newblock Lower bounds for maximal matchings and maximal independent sets.
\newblock In D.~Zuckerman, editor, {\em 60th {IEEE} Annual Symposium on
  Foundations of Computer Science, {FOCS} 2019, Baltimore, Maryland, USA,
  November 9-12, 2019}, pages 481--497. {IEEE} Computer Society, 2019.

\bibitem[BCMT22]{BehnezhadCMT22}
S.~Behnezhad, M.~Charikar, W.~Ma, and L.~Tan.
\newblock Almost 3-approximate correlation clustering in constant rounds.
\newblock In {\em 63rd {IEEE} Annual Symposium on Foundations of Computer
  Science, {FOCS} 2022, Denver, CO, USA, October 31 - November 3, 2022}, pages
  720--731. {IEEE}, 2022.

\bibitem[BCMT23]{BehnezhadCMT23}
S.~Behnezhad, M.~Charikar, W.~Ma, and L.~Tan.
\newblock Single-pass streaming algorithms for correlation clustering.
\newblock In N.~Bansal and V.~Nagarajan, editors, {\em Proceedings of the 2023
  {ACM-SIAM} Symposium on Discrete Algorithms, {SODA} 2023, Florence, Italy,
  January 22-25, 2023}, pages 819--849. {SIAM}, 2023.

\bibitem[BDH{\etalchar{+}}19a]{BehnezhadDHKS19}
S.~Behnezhad, M.~Derakhshan, M.~Hajiaghayi, M.~Knittel, and H.~Saleh.
\newblock Streaming and massively parallel algorithms for edge coloring.
\newblock In M.~A. Bender, O.~Svensson, and G.~Herman, editors, {\em 27th
  Annual European Symposium on Algorithms, {ESA} 2019, September 9-11, 2019,
  Munich/Garching, Germany}, volume 144 of {\em LIPIcs}, pages 15:1--15:14.
  Schloss Dagstuhl - Leibniz-Zentrum f{\"{u}}r Informatik, 2019.

\bibitem[BDH{\etalchar{+}}19b]{BehnezhadDHSS19}
S.~Behnezhad, M.~Derakhshan, M.~Hajiaghayi, C.~Stein, and M.~Sudan.
\newblock Fully dynamic maximal independent set with polylogarithmic update
  time.
\newblock In D.~Zuckerman, editor, {\em 60th {IEEE} Annual Symposium on
  Foundations of Computer Science, {FOCS} 2019, Baltimore, Maryland, USA,
  November 9-12, 2019}, pages 382--405. {IEEE} Computer Society, 2019.

\bibitem[Beh21]{Behnezhad21}
S.~Behnezhad.
\newblock Time-optimal sublinear algorithms for matching and vertex cover.
\newblock In {\em 62nd {IEEE} Annual Symposium on Foundations of Computer
  Science, {FOCS} 2021, Denver, CO, USA, February 7-10, 2022}, pages 873--884.
  {IEEE}, 2021.

\bibitem[BEPS12]{BarenboimEPS12}
L.~Barenboim, M.~Elkin, S.~Pettie, and J.~Schneider.
\newblock The locality of distributed symmetry breaking.
\newblock In {\em 53rd Annual {IEEE} Symposium on Foundations of Computer
  Science, {FOCS} 2012, New Brunswick, NJ, USA, October 20-23, 2012}, pages
  321--330. {IEEE} Computer Society, 2012.

\bibitem[BFS12]{BlellochFS12}
G.~E. Blelloch, J.~T. Fineman, and J.~Shun.
\newblock Greedy sequential maximal independent set and matching are parallel
  on average.
\newblock In G.~E. Blelloch and M.~Herlihy, editors, {\em 24th {ACM} Symposium
  on Parallelism in Algorithms and Architectures, {SPAA} '12, Pittsburgh, PA,
  USA, June 25-27, 2012}, pages 308--317. {ACM}, 2012.

\bibitem[BG14]{BravermanG14}
M.~Braverman and A.~Garg.
\newblock Public vs private coin in bounded-round information.
\newblock In J.~Esparza, P.~Fraigniaud, T.~Husfeldt, and E.~Koutsoupias,
  editors, {\em Automata, Languages, and Programming - 41st International
  Colloquium, {ICALP} 2014, Copenhagen, Denmark, July 8-11, 2014, Proceedings,
  Part {I}}, volume 8572 of {\em Lecture Notes in Computer Science}, pages
  502--513. Springer, 2014.

\bibitem[BH22]{BodwinH22}
G.~Bodwin and G.~Hoppenworth.
\newblock New additive spanner lower bounds by an unlayered obstacle product.
\newblock In {\em 63rd {IEEE} Annual Symposium on Foundations of Computer
  Science, {FOCS} 2022, Denver, CO, USA, October 31 - November 3, 2022}, pages
  778--788. {IEEE}, 2022.

\bibitem[BHP12]{BernsHP12}
A.~Berns, J.~Hegeman, and S.~V. Pemmaraju.
\newblock Super-fast distributed algorithms for metric facility location.
\newblock In A.~Czumaj, K.~Mehlhorn, A.~M. Pitts, and R.~Wattenhofer, editors,
  {\em Automata, Languages, and Programming - 39th International Colloquium,
  {ICALP} 2012, Warwick, UK, July 9-13, 2012, Proceedings, Part {II}}, volume
  7392 of {\em Lecture Notes in Computer Science}, pages 428--439. Springer,
  2012.

\bibitem[BRWY13a]{BravermanRWY13b}
M.~Braverman, A.~Rao, O.~Weinstein, and A.~Yehudayoff.
\newblock Direct product via round-preserving compression.
\newblock In F.~V. Fomin, R.~Freivalds, M.~Z. Kwiatkowska, and D.~Peleg,
  editors, {\em Automata, Languages, and Programming - 40th International
  Colloquium, {ICALP} 2013, Riga, Latvia, July 8-12, 2013, Proceedings, Part
  {I}}, volume 7965 of {\em Lecture Notes in Computer Science}, pages 232--243.
  Springer, 2013.

\bibitem[BRWY13b]{BravermanRWY13a}
M.~Braverman, A.~Rao, O.~Weinstein, and A.~Yehudayoff.
\newblock Direct products in communication complexity.
\newblock In {\em 54th Annual {IEEE} Symposium on Foundations of Computer
  Science, {FOCS} 2013, 26-29 October, 2013}, pages 746--755, 2013.

\bibitem[BS23]{BehnezhadS23}
S.~Behnezhad and M.~Saneian.
\newblock Streaming edge coloring with asymptotically optimal colors.
\newblock {\em CoRR}, abs/2305.01714, 2023.

\bibitem[BW15]{BravermanW15}
M.~Braverman and O.~Weinstein.
\newblock An interactive information odometer and applications.
\newblock In {\em Proceedings of the Forty-Seventh Annual {ACM} on Symposium on
  Theory of Computing, {STOC} 2015, June 14-17, 2015}, pages 341--350, 2015.

\bibitem[CCM08]{ChakrabartiCM08}
A.~Chakrabarti, G.~Cormode, and A.~McGregor.
\newblock Robust lower bounds for communication and stream computation.
\newblock In {\em Proceedings of the 40th Annual {ACM} Symposium on Theory of
  Computing, Victoria, British Columbia, Canada, May 17-20, 2008}, pages
  641--650, 2008.

\bibitem[CDK14]{ChierichettiDK14}
F.~Chierichetti, N.~N. Dalvi, and R.~Kumar.
\newblock Correlation clustering in mapreduce.
\newblock In S.~A. Macskassy, C.~Perlich, J.~Leskovec, W.~Wang, and R.~Ghani,
  editors, {\em The 20th {ACM} {SIGKDD} International Conference on Knowledge
  Discovery and Data Mining, {KDD} '14, New York, NY, {USA} - August 24 - 27,
  2014}, pages 641--650. {ACM}, 2014.

\bibitem[CDK19]{CormodeDK19}
G.~Cormode, J.~Dark, and C.~Konrad.
\newblock Independent sets in vertex-arrival streams.
\newblock In {\em 46th International Colloquium on Automata, Languages, and
  Programming, {ICALP} 2019, July 9-12, 2019, Patras, Greece}, pages
  45:1--45:14, 2019.

\bibitem[CE06]{CoppersmithE06}
D.~Coppersmith and M.~Elkin.
\newblock Sparse sourcewise and pairwise distance preservers.
\newblock {\em {SIAM} J. Discret. Math.}, 20(2):463--501, 2006.

\bibitem[CKL{\etalchar{+}}24]{CambusKLPU24}
M.~Cambus, F.~Kuhn, E.~Lindy, S.~Pai, and J.~Uitto.
\newblock A {$(3 + \eps)$}-approximate correlation clustering algorithm in
  dynamic streams.
\newblock In D.~Woodruff, editor, {\em Proceedings of the {ACM-SIAM} Symposium
  on Discrete Algorithms, {SODA} 2024}. {SIAM}, 2024.

\bibitem[CKP{\etalchar{+}}21a]{ChenKPSSY21}
L.~Chen, G.~Kol, D.~Paramonov, R.~R. Saxena, Z.~Song, and H.~Yu.
\newblock Almost optimal super-constant-pass streaming lower bounds for
  reachability.
\newblock In S.~Khuller and V.~V. Williams, editors, {\em {STOC} '21: 53rd
  Annual {ACM} {SIGACT} Symposium on Theory of Computing, Virtual Event, Italy,
  June 21-25, 2021}, pages 570--583. {ACM}, 2021.

\bibitem[CKP{\etalchar{+}}21b]{ChenKPSSY21b}
L.~Chen, G.~Kol, D.~Paramonov, R.~R. Saxena, Z.~Song, and H.~Yu.
\newblock Near-optimal two-pass streaming algorithm for sampling random walks
  over directed graphs.
\newblock In N.~Bansal, E.~Merelli, and J.~Worrell, editors, {\em 48th
  International Colloquium on Automata, Languages, and Programming, {ICALP}
  2021, July 12-16, 2021, Glasgow, Scotland (Virtual Conference)}, volume 198
  of {\em LIPIcs}, pages 52:1--52:19. Schloss Dagstuhl - Leibniz-Zentrum
  f{\"{u}}r Informatik, 2021.

\bibitem[CKPU23]{CambusKPU23}
M.~Cambus, F.~Kuhn, S.~Pai, and J.~Uitto.
\newblock Time and space optimal massively parallel algorithm for the 2-ruling
  set problem.
\newblock {\em CoRR}, abs/2306.00432, 2023.

\bibitem[CLM{\etalchar{+}}21]{Cohen-AddadLMNP21}
V.~Cohen{-}Addad, S.~Lattanzi, S.~Mitrovic, A.~Norouzi{-}Fard, N.~Parotsidis,
  and J.~Tarnawski.
\newblock Correlation clustering in constant many parallel rounds.
\newblock In M.~Meila and T.~Zhang, editors, {\em Proceedings of the 38th
  International Conference on Machine Learning, {ICML} 2021, 18-24 July 2021,
  Virtual Event}, volume 139 of {\em Proceedings of Machine Learning Research},
  pages 2069--2078. {PMLR}, 2021.

\bibitem[CMZ23]{ChechikMZ23}
S.~Chechik, D.~Mukhtar, and T.~Zhang.
\newblock Streaming edge coloring with subquadratic palette size.
\newblock {\em CoRR}, abs/2305.07090, 2023.

\bibitem[CSWY01]{ChakrabartiSWY01}
A.~Chakrabarti, Y.~Shi, A.~Wirth, and A.~C. Yao.
\newblock Informational complexity and the direct sum problem for simultaneous
  message complexity.
\newblock In {\em 42nd Annual Symposium on Foundations of Computer Science,
  {FOCS} 2001, 14-17 October 2001, Las Vegas, Nevada, {USA}}, pages 270--278.
  {IEEE} Computer Society, 2001.

\bibitem[CT06]{CoverT06}
T.~M. Cover and J.~A. Thomas.
\newblock {\em Elements of information theory {(2.} ed.)}.
\newblock Wiley, 2006.

\bibitem[CZ19]{ChechikZ19}
S.~Chechik and T.~Zhang.
\newblock Fully dynamic maximal independent set in expected poly-log update
  time.
\newblock In D.~Zuckerman, editor, {\em 60th {IEEE} Annual Symposium on
  Foundations of Computer Science, {FOCS} 2019, Baltimore, Maryland, USA,
  November 9-12, 2019}, pages 370--381. {IEEE} Computer Society, 2019.

\bibitem[Dar20]{Dark20}
J.~Dark.
\newblock {\em Finding structure in data streams: correlations, independent
  sets, and matchings}.
\newblock PhD thesis, University of Warwick, 2020.

\bibitem[DKO14]{DruckerKO14}
A.~Drucker, F.~Kuhn, and R.~Oshman.
\newblock On the power of the congested clique model.
\newblock In M.~M. Halld{\'{o}}rsson and S.~Dolev, editors, {\em {ACM}
  Symposium on Principles of Distributed Computing, {PODC} '14, Paris, France,
  July 15-18, 2014}, pages 367--376. {ACM}, 2014.

\bibitem[FGH{\etalchar{+}}24]{FlinGHKN24}
M.~Flin, M.~Ghaffari, M.~M. Halld{\'{o}}rsson, F.~Kuhn, and A.~Nolin.
\newblock A distributed palette sparsification theorem.
\newblock In D.~Woodruff, editor, {\em Proceedings of the {ACM-SIAM} Symposium
  on Discrete Algorithms, {SODA} 2024}. {SIAM}, 2024.

\bibitem[FKM{\etalchar{+}}05]{FeigenbaumKMSZ05}
J.~Feigenbaum, S.~Kannan, A.~McGregor, S.~Suri, and J.~Zhang.
\newblock On graph problems in a semi-streaming model.
\newblock {\em Theor. Comput. Sci.}, 348(2-3):207--216, 2005.

\bibitem[FN18]{FischerN18}
M.~Fischer and A.~Noever.
\newblock Tight analysis of parallel randomized greedy {MIS}.
\newblock In A.~Czumaj, editor, {\em Proceedings of the Twenty-Ninth Annual
  {ACM-SIAM} Symposium on Discrete Algorithms, {SODA} 2018, New Orleans, LA,
  USA, January 7-10, 2018}, pages 2152--2160. {SIAM}, 2018.

\bibitem[GG23]{GhaffariG23}
M.~Ghaffari and C.~Grunau.
\newblock Faster deterministic distributed {MIS} and approximate matching.
\newblock In B.~Saha and R.~A. Servedio, editors, {\em Proceedings of the 55th
  Annual {ACM} Symposium on Theory of Computing, {STOC} 2023, Orlando, FL, USA,
  June 20-23, 2023}, pages 1777--1790. {ACM}, 2023.

\bibitem[GGK{\etalchar{+}}18]{GhaffariGKMR18}
M.~Ghaffari, T.~Gouleakis, C.~Konrad, S.~Mitrovic, and R.~Rubinfeld.
\newblock Improved massively parallel computation algorithms for mis, matching,
  and vertex cover.
\newblock In {\em Proceedings of the 2018 {ACM} Symposium on Principles of
  Distributed Computing, {PODC} 2018, July 23-27, 2018}, pages 129--138, 2018.

\bibitem[Gha16]{Ghaffari16}
M.~Ghaffari.
\newblock An improved distributed algorithm for maximal independent set.
\newblock In R.~Krauthgamer, editor, {\em Proceedings of the Twenty-Seventh
  Annual {ACM-SIAM} Symposium on Discrete Algorithms, {SODA} 2016, Arlington,
  VA, USA, January 10-12, 2016}, pages 270--277. {SIAM}, 2016.

\bibitem[Gha22]{Ghaffari22}
M.~Ghaffari.
\newblock Local computation of maximal independent set.
\newblock In {\em 63rd {IEEE} Annual Symposium on Foundations of Computer
  Science, {FOCS} 2022, Denver, CO, USA, October 31 - November 3, 2022}, pages
  438--449. {IEEE}, 2022.

\bibitem[GKK12]{GoelKK12}
A.~Goel, M.~Kapralov, and S.~Khanna.
\newblock On the communication and streaming complexity of maximum bipartite
  matching.
\newblock In Y.~Rabani, editor, {\em Proceedings of the Twenty-Third Annual
  {ACM-SIAM} Symposium on Discrete Algorithms, {SODA} 2012, Kyoto, Japan,
  January 17-19, 2012}, pages 468--485. {SIAM}, 2012.

\bibitem[GM08]{GuhaM08}
S.~Guha and A.~McGregor.
\newblock Tight lower bounds for multi-pass stream computation via pass
  elimination.
\newblock In L.~Aceto, I.~Damg{\aa}rd, L.~A. Goldberg, M.~M. Halld{\'{o}}rsson,
  A.~Ing{\'{o}}lfsd{\'{o}}ttir, and I.~Walukiewicz, editors, {\em Automata,
  Languages and Programming, 35th International Colloquium, {ICALP} 2008,
  Reykjavik, Iceland, July 7-11, 2008, Proceedings, Part {I:} Tack {A:}
  Algorithms, Automata, Complexity, and Games}, volume 5125 of {\em Lecture
  Notes in Computer Science}, pages 760--772. Springer, 2008.

\bibitem[GO16]{GuruswamiO16}
V.~Guruswami and K.~Onak.
\newblock Superlinear lower bounds for multipass graph processing.
\newblock {\em Algorithmica}, 76(3):654--683, 2016.

\bibitem[GS23]{GhoshS23}
P.~Ghosh and M.~Stoeckl.
\newblock Low-memory algorithms for online and w-streaming edge coloring.
\newblock {\em CoRR}, abs/2304.12285, 2023.

\bibitem[GU19]{GhaffariU19}
M.~Ghaffari and J.~Uitto.
\newblock Sparsifying distributed algorithms with ramifications in massively
  parallel computation and centralized local computation.
\newblock In T.~M. Chan, editor, {\em Proceedings of the Thirtieth Annual
  {ACM-SIAM} Symposium on Discrete Algorithms, {SODA} 2019, San Diego,
  California, USA, January 6-9, 2019}, pages 1636--1653. {SIAM}, 2019.

\bibitem[Hes03]{Hesse03}
W.~Hesse.
\newblock Directed graphs requiring large numbers of shortcuts.
\newblock In {\em Proceedings of the Fourteenth Annual {ACM-SIAM} Symposium on
  Discrete Algorithms, January 12-14, 2003, Baltimore, Maryland, {USA}}, pages
  665--669. {ACM/SIAM}, 2003.

\bibitem[HJMR07]{HarshaJMR07}
P.~Harsha, R.~Jain, D.~A. McAllester, and J.~Radhakrishnan.
\newblock The communication complexity of correlation.
\newblock In {\em 22nd Annual {IEEE} Conference on Computational Complexity
  {(CCC} 2007), 13-16 June 2007, San Diego, California, {USA}}, pages 10--23.
  {IEEE} Computer Society, 2007.

\bibitem[HP21]{HuangP21}
S.~Huang and S.~Pettie.
\newblock Lower bounds on sparse spanners, emulators, and diameter-reducing
  shortcuts.
\newblock {\em {SIAM} J. Discret. Math.}, 35(3):2129--2144, 2021.

\bibitem[HZ23]{HaqiZ23}
A.~Haqi and H.~Zarrabi{-}Zadeh.
\newblock Almost optimal massively parallel algorithms for k-center clustering
  and diversity maximization.
\newblock In K.~Agrawal and J.~Shun, editors, {\em Proceedings of the 35th
  {ACM} Symposium on Parallelism in Algorithms and Architectures, {SPAA} 2023,
  Orlando, FL, USA, June 17-19, 2023}, pages 239--247. {ACM}, 2023.

\bibitem[JRS09]{JainRS09}
R.~Jain, J.~Radhakrishnan, and P.~Sen.
\newblock A property of quantum relative entropy with an application to privacy
  in quantum communication.
\newblock {\em J. {ACM}}, 56(6):33:1--33:32, 2009.

\bibitem[Kap13]{Kapralov13}
M.~Kapralov.
\newblock Better bounds for matchings in the streaming model.
\newblock In {\em Proceedings of the Twenty-Fourth Annual {ACM-SIAM} Symposium
  on Discrete Algorithms, {SODA} 2013, New Orleans, Louisiana, USA, January
  6-8, 2013}, pages 1679--1697, 2013.

\bibitem[Kap21]{Kapralov21}
M.~Kapralov.
\newblock Space lower bounds for approximating maximum matching in the edge
  arrival model.
\newblock In D.~Marx, editor, {\em Proceedings of the 2021 {ACM-SIAM} Symposium
  on Discrete Algorithms, {SODA} 2021, Virtual Conference, January 10 - 13,
  2021}, pages 1874--1893. {SIAM}, 2021.

\bibitem[KMVV13]{KumarMVV13}
R.~Kumar, B.~Moseley, S.~Vassilvitskii, and A.~Vattani.
\newblock Fast greedy algorithms in mapreduce and streaming.
\newblock In {\em 25th {ACM} Symposium on Parallelism in Algorithms and
  Architectures, {SPAA} '13, Montreal, QC, Canada - July 23 - 25, 2013}, pages
  1--10, 2013.

\bibitem[KMW16]{KuhnMW16}
F.~Kuhn, T.~Moscibroda, and R.~Wattenhofer.
\newblock Local computation: Lower and upper bounds.
\newblock {\em J. {ACM}}, 63(2):17:1--17:44, 2016.

\bibitem[KN97]{KushilevitzN97}
E.~Kushilevitz and N.~Nisan.
\newblock {\em Communication complexity}.
\newblock Cambridge University Press, 1997.

\bibitem[KN21]{KonradN21}
C.~Konrad and K.~K. Naidu.
\newblock On two-pass streaming algorithms for maximum bipartite matching.
\newblock In M.~Wootters and L.~Sanit{\`{a}}, editors, {\em Approximation,
  Randomization, and Combinatorial Optimization. Algorithms and Techniques,
  {APPROX/RANDOM} 2021, August 16-18, 2021, University of Washington, Seattle,
  Washington, {USA} (Virtual Conference)}, volume 207 of {\em LIPIcs}, pages
  19:1--19:18. Schloss Dagstuhl - Leibniz-Zentrum f{\"{u}}r Informatik, 2021.

\bibitem[KN24]{KonradK24}
C.~Konrad and K.~K. Naidu.
\newblock An unconditional lower bound for two-pass streaming algorithms for
  maximum matching approximation.
\newblock In {\em Proceedings of the 2024 {ACM-SIAM} Symposium on Discrete
  Algorithms, {SODA} 2024}. {SIAM}, 2024.

\bibitem[KNR95]{KremerNR95}
I.~Kremer, N.~Nisan, and D.~Ron.
\newblock On randomized one-round communication complexity.
\newblock In {\em Proceedings of the Twenty-Seventh Annual {ACM} Symposium on
  Theory of Computing, 29 May-1 June 1995, Las Vegas, Nevada, {USA}}, pages
  596--605, 1995.

\bibitem[Kon18a]{Konrad18}
C.~Konrad.
\newblock {MIS} in the congested clique model in o(log log {\(\Delta\)})
  rounds.
\newblock {\em CoRR}, abs/1802.07647, 2018.

\bibitem[Kon18b]{Konrad18b}
C.~Konrad.
\newblock A simple augmentation method for matchings with applications to
  streaming algorithms.
\newblock In I.~Potapov, P.~G. Spirakis, and J.~Worrell, editors, {\em 43rd
  International Symposium on Mathematical Foundations of Computer Science,
  {MFCS} 2018, August 27-31, 2018, Liverpool, {UK}}, volume 117 of {\em
  LIPIcs}, pages 74:1--74:16. Schloss Dagstuhl - Leibniz-Zentrum f{\"{u}}r
  Informatik, 2018.

\bibitem[KPRR19]{KonradPRR19}
C.~Konrad, S.~V. Pemmaraju, T.~Riaz, and P.~Robinson.
\newblock The complexity of symmetry breaking in massive graphs.
\newblock In J.~Suomela, editor, {\em 33rd International Symposium on
  Distributed Computing, {DISC} 2019, October 14-18, 2019, Budapest, Hungary},
  volume 146 of {\em LIPIcs}, pages 26:1--26:18. Schloss Dagstuhl -
  Leibniz-Zentrum f{\"{u}}r Informatik, 2019.

\bibitem[Lin87]{Linial87}
N.~Linial.
\newblock Distributive graph algorithms-global solutions from local data.
\newblock In {\em 28th Annual Symposium on Foundations of Computer Science, Los
  Angeles, California, USA, 27-29 October 1987}, pages 331--335. {IEEE}
  Computer Society, 1987.

\bibitem[LMSV11]{LattanziMSV11}
S.~Lattanzi, B.~Moseley, S.~Suri, and S.~Vassilvitskii.
\newblock Filtering: a method for solving graph problems in mapreduce.
\newblock In {\em {SPAA} 2011: Proceedings of the 23rd Annual {ACM} Symposium
  on Parallelism in Algorithms and Architectures, San Jose, CA, USA, June 4-6,
  2011 (Co-located with {FCRC} 2011)}, pages 85--94, 2011.

\bibitem[Lub85]{Luby85}
M.~Luby.
\newblock A simple parallel algorithm for the maximal independent set problem.
\newblock In {\em Proceedings of the 17th Annual {ACM} Symposium on Theory of
  Computing, May 6-8, 1985, Providence, Rhode Island, {USA}}, pages 1--10,
  1985.

\bibitem[LWWX22]{LuWWX22}
K.~Lu, V.~V. Williams, N.~Wein, and Z.~Xu.
\newblock Better lower bounds for shortcut sets and additive spanners via an
  improved alternation product.
\newblock In J.~S. Naor and N.~Buchbinder, editors, {\em Proceedings of the
  2022 {ACM-SIAM} Symposium on Discrete Algorithms, {SODA} 2022, Virtual
  Conference / Alexandria, VA, USA, January 9 - 12, 2022}, pages 3311--3331.
  {SIAM}, 2022.

\bibitem[MNSW95]{MiltersenNSW95}
P.~B. Miltersen, N.~Nisan, S.~Safra, and A.~Wigderson.
\newblock On data structures and asymmetric communication complexity.
\newblock In F.~T. Leighton and A.~Borodin, editors, {\em Proceedings of the
  Twenty-Seventh Annual {ACM} Symposium on Theory of Computing, 29 May-1 June
  1995, Las Vegas, Nevada, {USA}}, pages 103--111. {ACM}, 1995.

\bibitem[NO08]{NguyenO08}
H.~N. Nguyen and K.~Onak.
\newblock Constant-time approximation algorithms via local improvements.
\newblock In {\em 49th Annual {IEEE} Symposium on Foundations of Computer
  Science, {FOCS} 2008, October 25-28, 2008, Philadelphia, PA, {USA}}, pages
  327--336. {IEEE} Computer Society, 2008.

\bibitem[OSSW18]{OnakSSW18}
K.~Onak, B.~Schieber, S.~Solomon, and N.~Wein.
\newblock Fully dynamic {MIS} in uniformly sparse graphs.
\newblock In I.~Chatzigiannakis, C.~Kaklamanis, D.~Marx, and D.~Sannella,
  editors, {\em 45th International Colloquium on Automata, Languages, and
  Programming, {ICALP} 2018, July 9-13, 2018, Prague, Czech Republic}, volume
  107 of {\em LIPIcs}, pages 92:1--92:14. Schloss Dagstuhl - Leibniz-Zentrum
  f{\"{u}}r Informatik, 2018.

\bibitem[RG20]{DRozhonG20}
V.~Rozhon and M.~Ghaffari.
\newblock Polylogarithmic-time deterministic network decomposition and
  distributed derandomization.
\newblock In K.~Makarychev, Y.~Makarychev, M.~Tulsiani, G.~Kamath, and
  J.~Chuzhoy, editors, {\em Proccedings of the 52nd Annual {ACM} {SIGACT}
  Symposium on Theory of Computing, {STOC} 2020, Chicago, IL, USA, June 22-26,
  2020}, pages 350--363. {ACM}, 2020.

\bibitem[RS78]{RuzsaS78}
I.~Z. Ruzsa and E.~Szemer{\'e}di.
\newblock Triple systems with no six points carrying three triangles.
\newblock {\em Combinatorics (Keszthely, 1976), Coll. Math. Soc. J. Bolyai},
  18:939--945, 1978.

\bibitem[RTVX11]{RubinfeldTVX11}
R.~Rubinfeld, G.~Tamir, S.~Vardi, and N.~Xie.
\newblock Fast local computation algorithms.
\newblock In B.~Chazelle, editor, {\em Innovations in Computer Science - {ICS}
  2011, Tsinghua University, Beijing, China, January 7-9, 2011. Proceedings},
  pages 223--238. Tsinghua University Press, 2011.

\bibitem[RVW16]{RoughgardenVW16}
T.~Roughgarden, S.~Vassilvitskii, and J.~R. Wang.
\newblock Shuffles and circuits: (on lower bounds for modern parallel
  computation).
\newblock In C.~Scheideler and S.~Gilbert, editors, {\em Proceedings of the
  28th {ACM} Symposium on Parallelism in Algorithms and Architectures, {SPAA}
  2016, Asilomar State Beach/Pacific Grove, CA, USA, July 11-13, 2016}, pages
  1--12. {ACM}, 2016.

\bibitem[RY20]{RaoY20}
A.~Rao and A.~Yehudayoff.
\newblock {\em Communication Complexity: and Applications}.
\newblock Cambridge University Press, 2020.
\newblock \doi{10.1017/9781108671644}.

\bibitem[Suo20]{Suomela20}
J.~Suomela.
\newblock Using round elimination to understand locality.
\newblock {\em {SIGACT} News: Distributed Computing Column 79}, 51(3):62, 2020.

\bibitem[Wei15]{Weinstein15}
O.~Weinstein.
\newblock Information complexity and the quest for interactive compression.
\newblock {\em {SIGACT} News}, 46(2):41--64, 2015.

\bibitem[WXX23]{WilliamsXX23}
V.~V. Williams, Y.~Xu, and Z.~Xu.
\newblock Simpler and higher lower bounds for shortcut sets.
\newblock {\em CoRR}, abs/2310.12051, 2023.

\bibitem[YYI09]{YoshidaYI09}
Y.~Yoshida, M.~Yamamoto, and H.~Ito.
\newblock An improved constant-time approximation algorithm for maximum
  matchings.
\newblock In M.~Mitzenmacher, editor, {\em Proceedings of the 41st Annual {ACM}
  Symposium on Theory of Computing, {STOC} 2009, Bethesda, MD, USA, May 31 -
  June 2, 2009}, pages 225--234. {ACM}, 2009.

\end{thebibliography}
